\documentclass[reqno,a4paper]{amsart}

\usepackage{amssymb}

\usepackage{adjustbox}

\makeatletter
\def\blfootnote{\xdef\@thefnmark{}\@footnotetext}
\makeatother

\newcommand{\degree}{\ensuremath{^\circ}{}}

\newcounter{mnotecount}[section]

\newcounter{mymnotecount}[section]
\renewcommand{\themymnotecount}{\thesection.\arabic{mymnotecount}}
\newcommand{\mymnote}[1]{\protect{\stepcounter{mnotecount}}${\raisebox{0.5\baselineskip}[0pt]{\makebox[0pt][c]{\color{blue}{\tiny\em$\bullet$\themymnotecount}}}}$\marginpar{\raggedright\tiny\em$\!\bullet$\themymnotecount:
\blue{#1}}\ignorespaces}

\renewcommand{\mymnote}[1]{}

\usepackage{comment}
\includecomment{SKETCH} 
\excludecomment{TOOLONG}

\newcommand{\red}[1]{{\color{red}#1}}

\newcommand{\blue}[1]{{\color{blue}#1}}

\usepackage{mathrsfs}
\usepackage{eucal}

\usepackage{ipa} 

\DeclareMathOperator{\tho}{\text{\rm \thorn}}
\DeclareMathOperator{\edt}{\text{\rm \eth}}

\usepackage{color}
\usepackage{graphicx}
\usepackage{wrapfig}
\usepackage{subcaption}
\usepackage{lpic}
\usepackage{verbatim}

\numberwithin{equation}{section}

\usepackage[linktocpage]{hyperref}

\allowdisplaybreaks[2]

\theoremstyle{plain}
\newtheorem{theorem}{Theorem}[section]
\newtheorem{proposition}[theorem]{Proposition}
\newtheorem{lemma}[theorem]{Lemma}

\newtheorem{definition}[theorem]{Definition}
\newtheorem{remark}[theorem]{Remark}
\newtheorem{example}[theorem]{Example}

\renewcommand{\Re}{{\mathbb R}}         %

\newcommand{\Co}{{\mathbb C}}         %
\newcommand{\half}{\tfrac{1}{2}}         %
\newcommand{\eps}{\epsilon}

\newcommand{\diag}{\operatorname{diag}}

\newcommand{\Mink}{\mathbb M}

\newcommand{\NatNum}{\mathbb N}
   
\newcommand{\Mcal}{\mathcal M}

\newcommand{\sDiv}{\mathscr{D}}
\newcommand{\sCurl}{\mathscr{C}}
\newcommand{\sCurlDagger}{\mathscr{C}^\dagger}
\newcommand{\sTwist}{\mathscr{T}}
\newcommand{\sExt}{\mathscr{E}}

\newcommand{\SL}{\mathrm{SL}}
\newcommand{\SO}{\mathrm{SO}}

\newcommand{\Spin}{\text{Spin}}

\newcommand{\KillSpin}{\mathcal{KS}}

\newcommand{\PetrovD}{D}
\newcommand{\PetrovN}{N}
\newcommand{\PetrovO}{O}

\newcommand{\NPl}{l}
\newcommand{\NPn}{n}
\newcommand{\NPm}{m}
\newcommand{\NPmbar}{\bar m}

\newcommand{\KSigma}{\Sigma} 
\newcommand{\KDelta}{\Delta}

\usepackage{centernot}
\newcommand{\angDelta}{\centernot{\Delta}}

\newcommand{\curlyR}{\mathcal{R}}

\newcommand{\Lie}{\mathcal L} 

\newcommand{\SpaceSlice}{\Sigma} 
\newcommand{\Spaceh}{h}
\newcommand{\SpaceK}{k}
\newcommand{\ordo}{o}
\newcommand{\met}{g} 

\newcommand{\ba}{\mathbf{a}}
\newcommand{\bb}{\mathbf{b}}
\newcommand{\bc}{\mathbf{c}}
\newcommand{\bA}{\mathbf{A}}
\newcommand{\bB}{\mathbf{B}}
\newcommand{\bC}{\mathbf{C}}

\newcommand{\SymSpin}{S}
\newcommand{\SymSpinSec}{\mathcal S}

\newcommand{\GenVec}{\nu} %

\newcommand{\EMTensorT}{X}
\newcommand{\etavarphi}{\varsigma}

\newcommand{\ua}{{\underline{a}}}
\newcommand{\ub}{{\underline{b}}}

\newcommand{\DiffCurlyRTilde}{{\tilde{\mathcal{R}}'}{}}

\newcommand{\newcommandMG}[3]{\newcommand{#1}{#3}} %
\newcommandMG{\fnMna} {{f_{1}}} {z}
\newcommandMG{\fnMnb} {{f_{2}}} {w}
\newcommandMG{\fnMca} {f_{1,1}}   {z_1}
\newcommandMG{\fnMcb} {f_{2,1}}   {w_1}
\newcommandMG{\fnMda} {f_{1,2}}   {z_2}
\newcommandMG{\fnMdb} {f_{2,2}}   {w_2}

\newcommand{\rp}{r_+}
\newcommand{\rs}{r_*}

\newcommand{\dr}{\partial_r}

\newcommand{\di}{\mathrm{d}} %

\newcommand{\Geodesic}{\gamma}

\newcommand{\GeodesicMass}{\boldsymbol{\mu}}
\newcommand{\GeodesicEnergy}{\boldsymbol{e}}
\newcommand{\GeodesicLz}{\boldsymbol{\ell_z}}
\newcommand{\GeodesicLsquared}{\boldsymbol{L}^2}
\newcommand{\GeodesicQ}{\boldsymbol{q}}
\newcommand{\TensorQ}{Q}
\newcommand{\GeodesicKCarter}{\boldsymbol{k}}
\newcommand{\TensorKCarter}{K}
\newcommand{\OperatorQ}{Q}
\newcommand{\OperatorKCarter}{K}

\newcommand{\Rop}{\mathrm{R}}
\newcommand{\Sop}{\mathrm{S}}
\newcommand{\Rsol}{R}
\newcommand{\Ssol}{S}

\usepackage{upgreek}
\newcommand{\bmu}{\upmu}

\newcommand{\GenEnergyGeodesic}[1]{e_{#1}}
\newcommand{\rorbit}{r_{o}} %
\newcommand{\vecMGeodesic}{A}
\newcommand{\fnMrGeodesic}{\mathcal{F}}
\newcommand{\fnMrMorawetz}{\mathcal{F}}
\newcommand{\fnMpGeodesic}{q_{\text{reduced}}}

\newcommand{\gMetric}{g}

\newcommand{\vecMprimary}{A}
\newcommand{\scalMprimary}{q}

\newcommand{\HypersurfaceGeneral}{\Sigma}

\newcommand{\vecfont}[1]{#1}

\newcommand{\vecX}{\vecfont{X}}

\makeatletter
  \def\moverlay{\mathpalette\mov@rlay}
  \def\mov@rlay#1#2{\leavevmode\vtop{%
     \baselineskip\z@skip \lineskiplimit-\maxdimen
     \ialign{\hfil$#1##$\hfil\cr#2\crcr}}}
\makeatother

\newcommand{\squareTME}{\moverlay{\square\cr {\scriptscriptstyle \mathrm T}}}

\newcommand{\hme}{\widehat{g}}
\newcommand{\id}{\mathbf{i}}
\newcommand{\Action}{I}
\newcommand{\tr}{\text{tr}}
\newcommand{\Ucal}{\mathcal{U}}
\newcommand{\Vcal}{\mathcal{V}}
\newcommand{\Tcal}{\mathcal{T}}
\newcommand{\Rcal}{\mathcal{R}}
\newcommand{\Scri}{\mathcal{I}}
\newcommand{\Horizon}{\mathcal{H}}

\def\EMTensorH{\mathbf{H}}

\newcommand{\OrthFrameT}{\widehat T}
\newcommand{\OrthFrameX}{\widehat X}
\newcommand{\OrthFrameY}{\widehat Y}
\newcommand{\OrthFrameZ}{\widehat Z}
\def\MorawetzCurrJ{\mathbf{P}}

\title{Geometry of black hole spacetimes}

\author[L. Andersson]{Lars Andersson} 
\email{laan@aei.mpg.de}
\address{Albert Einstein Institute, Am M\"uhlenberg 1, D-14476 Potsdam,
Germany 
}

\author[T. B\"ackdahl]{Thomas B\"ackdahl} \email{thobac@chalmers.se}
\address{Mathematical Sciences, Chalmers University of Technology and University of Gothenburg, SE-412 96 Gothenburg, Sweden
\and 
The School of Mathematics, University of Edinburgh, James Clerk Maxwell Building, 
Peter Guthrie Tait Road, Edinburgh
EH9 3FD, UK
}

\author[P. Blue]{Pieter Blue} \email{P.Blue@ed.ac.uk}
\address{The School of Mathematics and the Maxwell Institute, University of Edinburgh, James Clerk Maxwell Building, 
Peter Guthrie Tait Road, Edinburgh
EH9 3FD,UK}

\date{October 11, 2016}

\begin{document} 

\begin{abstract} These notes, based on lectures given at the summer school on Asymptotic Analysis in General Relativity, collect material on the Einstein equations, the geometry of black hole spacetimes, and the analysis of fields on black hole backgrounds. The Kerr model of a rotating black hole in vacuum is expected to be unique and stable. The problem of proving these fundamental facts provides the background for the material presented in these notes. 
 
Among the many topics which are relevant for the uniqueness and stability problems are the theory of fields on black hole spacetimes, in particular for gravitational perturbations of the Kerr black hole, and more generally, the study of  nonlinear field equations in the presence of trapping. The study of these questions requires tools from several different fields, including Lorentzian geometry, hyperbolic differential equations and spin geometry, which are all relevant to the black hole stability problem. 
\end{abstract}

\maketitle
\blfootnote{Based on lectures given by the first named author at the 2014 Summer School on Asymptotic Analysis in General Relativity, held at Institut Fourier, Grenoble}

\tableofcontents

\section{Introduction} \label{sec:newintro}
A short time after Einstein published his field equations for general relativity in 1915, 
Karl Schwarzschild discovered an exact 
and explicit solution of the Einstein vacuum equations describing
the gravitational field of a spherical body at rest. 
In analyzing Schwarzschild's solution, one finds that if the central body
is sufficiently concentrated, light emitted from its surface cannot reach an
observer at infinity. It was not until the 1950's that the global structure of the Schwarzschild spacetime was understood. By this time causality  theory and the Cauchy problem for the Einstein equations was firmly established, although many important problems remained open. 
Observations of highly energetic phenomena occurring within small spacetime regions, eg. quasars, made it plausible that black holes played a significant role in astropysics, and by the late 1960's these objects were part of mainstream astronomy and astrophysics. The term ``black hole'' for this type of object came into use in the 1960's. According to our current understanding, black holes are ubiquitous in the universe, in particular most galaxies have a supermassive black hole at their center, and these play an important role in the life of the galaxy. Also our galaxy has at its center a very compact object, Sagittarius A*, with a diameter of less than one astronomical unit, and a mass estimated to be $10^6$ $M_{\odot}$. Evidence for this includes observations of the orbits of stars in its vicinity.

Recall that a solution to the Einstein vacuum equations is a Lorentzian spacetime $(\Mcal, \met_{ab})$, satisfying $R_{ab} = 0$, where $R_{ab}$ is the Ricci tensor of $\met_{ab}$. The Einstein equation is the Euler-Lagrange equation of the diffeomorphism invariant Einstein-Hilbert action functional, given by the integral of the scalar curvature of $(\Mcal, \met_{ab})$, 
$$
\int_{\Mcal} R d\mu_\met.
$$
The diffeomorphism invariance, or general covariance, of the action has the consequence that Cauchy data for the Einstein equation must satisfy a set of constraint equations, and that the principal symbol of the Euler-Lagrange equation is degenerate\footnote{From a hyperbolic PDE perspective, the Einstein equations are both over- and under-determined. Contracting the Einstein equation against the normal to a smooth spacelike hypersurface gives elliptic equations that must be satisfied on the hypersurface; these are called the constraint equations. After introducing suitable gauge conditions, the combination of the gauge conditions and the remaining Einstein equations form a hyperbolic system of evolution equations. Furthermore, if the initial data satisfies the constraint equations, then the solution to this hyperbolic system, when restricted to any spacelike hypersurface, also satisfies the constraint equations. If the initial hypersurface is null, the situation becomes more complicated to summarise but simpler to treat in full detail.}. After introducing suitable gauge conditions, the Einstein equations can be reduced to a hyperbolic system of evolution equations. It is known that for any set of sufficiently regular Cauchy data satisfying the constraints, the Cauchy problem for the Einstein equation has a unique solution which is maximal among all regular, vacuum Cauchy developments. This general result, however, does not give any detailed information about the properties of the maximal development. 

There are two main conjectures about the maximal development. The strong cosmic censorship conjecture (SCC) states that a generic maximal development is inextendible, as a regular vacuum spacetime. There are examples where the maximal development is extendible, and has non-unique extensions, which furthermore may contain closed timelike curves. In these cases, predictability fails for the Einstein equations, but if SCC holds, they are non-generic. At present, this is only known to hold in the context of  families of spacetimes with symmetry restrictions, see \cite{lrr-2010-2,MR2098914} and references therein.
Further, some non-linear stability results without symmetry assumptions, including stability of Minkowski space and stability of quotients of the Milne model (also known as L\"obell spacetimes, see \cite{MR1729952,MR2863911} and references therein), can be viewed as giving support to SCC. 
The weak cosmic censorship conjecture states that for a generic isolated system (i.e. an asymptotically flat solution of the Einstein equations), any singularity is hidden from observers at infinity. In this case, the spacetime contains a black hole region, i.e. the complement of the part of the spacetime visible to observers at infinity. The black hole region is bounded by the event horizon, the boundary of the region of spacetime which can be seen by observers at future infinity. 
Both of these conjectures remain wide open, although there has been limited progress on some problems related to them. The weak cosmic censorship conjecture is most relevant for the purpose of these notes, see \cite{1997gr.qc....10068W}.

The Schwarzschild solution is static, spherically symmetric, asymptotically flat, and has a single free parameter $M$ which represents the mass of the black hole. By Birkhoff's theorem it is the unique solution of the vacuum Einstein equations with these properties.  
In 1963 
Roy Kerr \cite{kerr:1963PhRvL..11..237K} discovered a new,
explicit family of asymptotically flat solutions of the vacuum Einstein equations which are 
stationary, axisymmetric, and rotating. Shortly after this, a charged, rotating black hole solution to the Einstein-Maxwell equations, known as the Kerr-Newman solution, was found, cf. \cite{1965JMP.....6..918N,1965JMP.....6..915N}. Recall that a vector field $\GenVec^a$ is Killing if $\nabla_{(a} \GenVec_{b)} = 0$. A Kerr spacetime admits two Killing fields, the stationary Killing field $(\partial_t)^a$ which is timelike at infinity, and the axial Killing field $(\partial_\phi)^a$. 
The Kerr family of solutions is parametrized by the mass $M$, and the azimuthal angular momentum per unit mass $a$. In the limit $a=0$, the Kerr solution reduces to the spherically symmetric 
Schwarzschild solution.
 
If $|a| \leq M$, the Kerr spacetime contains a black hole, while if $|a| > M$, there is a ringlike singularity which is naked, in the sense that it fails to be hidden from observers at infinity. This situation would violate the weak cosmic censorship conjecture, and one therefore expects that an overextreme Kerr spacetime is unstable, and in particular, that it cannot arise through a dynamical process from regular Cauchy data. 

For a geodesic $\gamma^a(\lambda)$ with velocity $\dot \gamma^a = d\gamma^a/d\lambda$, in a stationary axisymmetric spacetime\footnote{We use signature $+---$, in particular timelike vectors have positive norm.}, there are three conserved quantities, the mass $\GeodesicMass^2 =  \dot \gamma^a \dot \gamma_b$, energy $\GeodesicEnergy = \dot \gamma^a (\partial_t)_a$, and angular momentum $\GeodesicLz = \dot \gamma^a (\partial_\phi)_a$. In a general axisymmetric spacetime, geodesic motion is chaotic. However, as was discovered by Brandon Carter in 1968, there is a fourth conserved quantity for geodesics in the Kerr spacetime, the Carter constant $\GeodesicKCarter$, see section \ref{sec:kerrspacetime} for details. 
By Liouville's theorem, this allows one to integrate the geodesic equations by quadratures, and thus geodesics in the Kerr spacetime do not exhibit a chaotic behavior. 

The Carter constant is a manifestation of the fact that the Kerr spacetime is algebraically special, of Petrov type $\{2,2\}$, also known as type D. In particular, there are two repeated principal null directions for the Weyl tensor. As shown by Walker and Penrose \cite{walker:penrose:1970CMaPh..18..265W} 
a vacuum spacetime of Petrov type $\{2,2\}$ admits an object satisfying a generalization of Killing's equation, namely a Killing spinor $\kappa_{AB}$,  satisfying $\nabla_{A'(A} \kappa_{BC)} = 0$. As shown in the just cited paper, this leads to the presence of four conserved quantities for null geodesics.

Assuming some technical conditions, any stationary asymptotically flat, stationary black hole spacetime is expected to belong to the Kerr family, a fact which is known to hold in the real-analytic case. Further, the Kerr black hole is expected to be stable in the sense that a small perturbation of the Kerr space time settles down asymptotically to a member of the Kerr family. 

There is much observational evidence pointing to the fact that black holes exist in large numbers in the universe, and that they play a role in many astrophysically significant processes. For example, most galaxies, including our own galaxy, are believed to contain a supermassive black hole at their center. Further, dynamical processes involving black holes, such as mergers, are expected to be important sources of gravitational wave radiation, which could be observed by existing and planned gravitational wave observatories\footnote{At the time of writing, the first such observation has just been announced \cite{PhysRevLett.116.061102}.}
Thus, black holes play a central role in astrophysics. 

Due to its conjectured uniqueness and stability properties, these black holes are expected to be modelled by the Kerr, or Kerr-Newman solutions. However, in order to establish the astrophysical relevance of the Kerr solution, it is vital to find rigorous proofs of both of these conjectures, which can be referred to as the black hole uniqueness and stability problems, respectively. A great deal of work has been devoted to these and related problems,  and although progress has been made, both remain open at present. For a solution of the stability problem, it is important to have an effective characterization of the Kerr spacetime. This is an important aspect of the uniqueness problem.

\subsection*{Overview} Section \ref{sec:background} introduces a range of background material on general relativity, including a discussion of the Cauchy problem for the Einstein equations. The discussion of black hole spacetimes is started in section \ref{sec:blackholes} with a detailed discussion of the global geometry of the extended Schwarzschild spacetime, followed by some background on marginally outer trapped surfaces and dynamical black holes. Section \ref{sec:prel} introduced some concepts from spin geometry and the related GHP formalism. The Petrov classification is introduced and some properties of algebraically special spacetimes are of its consequences are presented.   In section \ref{sec:kerrspacetime} the geometry of the Kerr black hole spacetimes is introduced. 

Section \ref{sec:monotonicity} contains a discussion of null geodesics in the Kerr spacetime. A construction of monotone quantities for null geodesics based on vector fields with coefficients depending on conserved quantities, is introduced. In section \ref{sec:symop}, symmetry operators for fields on the Kerr spacetime are discussed. Dispersive estimates for fields are the analog of monotone quantities for null geodesics, and in constructing these, symmetry operators play a role analogous to the conserved quantities for the case of geodesics.

\section{Background} \label{sec:background}
\subsection{Minkowski space} \label{sec:minkowski}
Minkowski space $\Mink$ is $\Re^4$ with metric which in a Cartesian coordinate system $(x^a) = (t, x^i)$ takes the form\footnote{Here and below we shall use line elements, eg. $d\tau_{\Mink}^2 = (g_{\Mink})_{ab} dx^a dx^b$ and metrics, eg. $(g_{\Mink})_{ab}$ interchangeably.} 
$$
d\tau_{\Mink}^2 =  dt^2 - (dx^1)^2- (dx^2)^2 - (dx^3)^2 . 
$$
Introducing the spherical coordinates $r,\theta, \phi$ we can write the metric in the form $-dt^2 + dr^2 + r^2 d\Omega_{S^2}^2$, where $d\Omega_{S^2}^2$ is the line element on the standard $S^2$, 
\begin{equation}\label{eq:gS2}
d\Omega_{S^2}^2 = (g_{S^2})_{ab} dx^a dx^b = d\theta^2 + \sin^2\theta d\phi^2 .
\end{equation} 
A tangent vector $\GenVec^a$ is timelike, null, or spacelike when  $g_{ab} \GenVec^a \GenVec^b > 0$, $=0$, or $< 0$, respectively.
\begin{wrapfigure}{r}{0.15\textwidth}
\begin{center}
\raisebox{-0.5\height}{\includegraphics{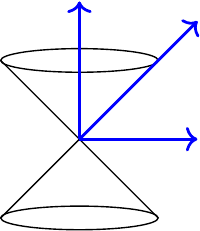}}
\end{center} 
\end{wrapfigure} 
Vectors with $g_{ab} \GenVec^a \GenVec^b \geq 0$ are called causal. 
Let $p,q \in \Mink$. We say that $p$ is in the causal (timelike) future of $q$ if $p-q$ is causal (timelike). The causal and timelike futures $J^+(p)$ and $I^+(p)$ of $p \in \Mink$ are the sets of points which are in the causal and timelike futures of $p$, respectively. The corresponding past notions are defined analogously.  

Let $u,v$ be given by 
$$
u = t-r, \quad v = t+r
$$
In terms of these coordinates the line element takes the form 
\begin{equation}\label{eq:uvMink}
d\tau_\Mink^2 =  du dv - r^2 d\Omega_{S^2}^2
\end{equation} 
We see that there are no terms $du^2$, $dv^2$, which corresponds to the fact that both $u,v$ are null coordinates. In particular, the vectors $(\partial_u)^a$, $(\partial_v)^a$ are null. A complex null tetrad is given by 
\begin{subequations}\label{eq:Minktetrad}
\begin{align}
\NPl^a ={}& \sqrt{2} (\partial_u)^a = \frac{1}{\sqrt{2}} \left (\partial_t)^a + (\partial_r)^a\right ), \\
\NPn^a ={}& \sqrt{2} (\partial_v)^a = \frac{1}{\sqrt{2}} \left ((\partial_t)^a - (\partial_r)^a\right ), \\
\NPm^a ={}& \frac{1}{\sqrt{2}r}\left ((\partial_\theta)^a + \frac{i}{\sin\theta} (\partial_\phi)^a\right )
\end{align}
\end{subequations}
normalized so that $\NPn^a \NPl_a = 1 = - \NPm^a \NPmbar_a$, with all other inner products of tetrad legs zero. Complex null tetrads with this normalization play a central role in the Newman-Penrose and Geroch-Held-Penrose formalisms, see section \ref{sec:prel}. In these notes we will use such tetrads unless otherwise stated. 

In terms of a null tetrad, we have  
\begin{equation}\label{eq:gNP}
g_{ab} = 2(\NPl_{(a} \NPn_{b)} - \NPm_{(a}\NPmbar_{b)}) .
\end{equation}
Introduce compactified null coordinates $\Ucal, \Vcal$, given by 
$$
\Ucal = \arctan u, \quad \Vcal = \arctan v .
$$
These take values in $\{(-\pi/2,\pi/2) \times (-\pi/2, \pi/2)\} \cap \{ \Vcal \geq \Ucal \}$, and we can thus present Minkowski space in a \emph{causal diagram}, see figure \ref{fig:Mink-causal}. 
\begin{figure}[!h]
\centering
\raisebox{-0.5\height}{\includegraphics{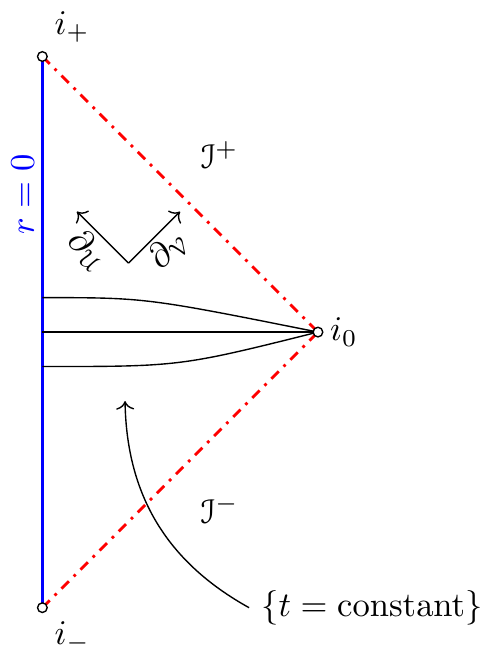}}
\caption{\ }
\label{fig:Mink-causal}
\end{figure} 
Here each point represents an $S^2$ and we have drawn null vectors at 45\degree angles. 
A compactification of Minkowski space is now given by adding the null boundaries\footnote{Here $\Scri$ is pronounced ``Scri'' for ``script I''.} $\Scri^\pm$, spatial infinity $i_0$ and timelike infinity $i^{\pm}$ as indicated in the figure. Explicitely, 
\begin{align*} 
\Scri^+ ={}&  \{ \Vcal = \pi/2 \} \\
\Scri^- ={}& \{\Ucal = - \pi/2 \} \\ 
i_0 ={}& \{\Vcal = \pi/2, \Ucal = - \pi/2 \} \\ 
i_{\pm} ={}& \{(\Vcal, \Ucal) = \pm (\pi/2, \pi/2) \}
\end{align*} 
In figure \ref{fig:Mink-causal}, we have also indicated schematically the $t$-level sets which approach spatial infinity $i_0$. 
Causal diagrams are a useful tool which, if applied with proper care, can be used to understand the structure of quite general spacetimes. Such diagrams are often referred to as Penrose, or Carter-Penrose diagrams. 

In particular, as can be seen from figure \ref{fig:Mink-causal}, we have that 
$\Mink = I^-(\Scri^+) \cap I^+(\Scri^-)$, i.e. any point in $\Mink$ is in the past of $\Scri^+$ and in the future of $\Scri^-$. This fact is related to the fact that $\Mink$ is \emph{asymptotically simple}, in the sense that it admits a conformal compactification with regular null boundary, and has the property that any inextendible null geodesic hits the null boundary. 
For massless fields on Minkowski space, this means that it makes sense to formulate a scattering map which takes data on $\Scri^-$ to data on $\Scri^+$, see \cite{MR0175590}.

Let 
\begin{equation}\label{eq:TcalRcal}
\Tcal = \Vcal + \Ucal, \quad \Rcal = \Vcal - \Ucal. 
\end{equation}
Then, 
with $\Phi^2 = 2 \cos \Ucal \cos\Vcal$, the conformally transformed metric $\tilde g_{ab} = \Phi^2 g_{ab}$ takes the form  
\begin{align*} 
\tilde g^\Mink_{ab} ={}& d\Tcal^2 - d\Rcal^2 - \sin^2 \Rcal d\Omega^2_{S^2} \\
={}& d\Tcal^2 - d\Omega^2_{S^3}
\end{align*} 
which we recognize as the metric on the cylinder $\Re \times S^3$.
This spacetime is known as the Einstein cylinder, and can be viewed as a static solution of the Einstein equations with dust matter and positive cosmological constant \cite{1917SPAW.......142E}.  
 \mymnote{add sketch of compactified Minkowski in Einstein cylinder} 

\subsection{Lorentzian geometry and causality} 
We now consider a smooth Lorentzian 4-manifold $(\Mcal, \met_{ab})$ with signature $+---$.  
Each tangent space in a 4-dimensional spacetime is isometric to Minkowski space $\Mink$, and we can carry intuitive notions of causality over from $\Mink$ to $\Mcal$. We say that a smooth curve $\gamma^a(\lambda)$ is causal if the velocity vector $\dot \gamma^a = d\gamma^a/d\lambda$ is causal. Two points in $\Mcal$ are causally related if they can be connected by a piecewise smooth causal curve. The concept of causal curves is most naturally defined for $C^0$ curves. A $C^0$ curve $\gamma^a$ is said to be causal if each pair of points on $\gamma^a$ are causally related. We may define timelike curve and timelike related points in the analogous manner. 

We now assume that $\Mcal$ is time oriented, i.e. that there is a globally defined time-like vector field on $\Mcal$. This allows us to distinguish between future and past directed causal curves, and to introduce a notion of the causal and timelike future of a spacetime point. The corresponding past notions are defined analogously. If $q$ is in the causal future of $p$, we write $p \preccurlyeq q$. This introduces a partial order on $\Mcal$. 
The causal future $J^+(p)$ of $p$ is defined as $J^+(p) = \{ q : p \preccurlyeq q \}$ while the timelike future $I^+(p)$ is defined in the analogous manner, with timelike replacing causal. 
A subset $\Sigma \subset \Mcal$ is achronal 
\begin{wrapfigure}[5]{l}{0.15\textwidth}
\centering
\vskip -.2in
\raisebox{-0.5\height}{\includegraphics{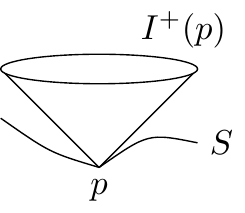}}
\end{wrapfigure} 
if there is no pair $p,q \in \Mcal$ such that $q \in I^+(p)$, i.e. $\Sigma$ does not intersect its timelike future or past. 
The domain of dependence $D(S)$ of $S \subset \Mcal$
\begin{wrapfigure}[4]{r}{0.25\textwidth}
\centering
\vskip -.3in
\raisebox{-0.5\height}{\includegraphics{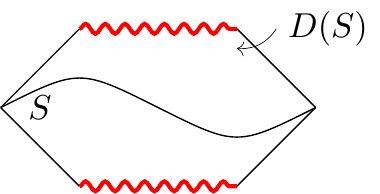}}
\end{wrapfigure}
is the set of points $p$ such that any inextendible causal curve starting at $p$ must intersect $S$.  

\begin{definition}
A spacetime $\Mcal$ is globally hyperbolic if there is a closed, achronal $\Sigma \subset \Mcal$ such that $\Mcal = D(\Sigma)$. In this case, $\Sigma$ is called a Cauchy surface. 
\end{definition} 

Due to results of Bernal and Sanchez \cite{2006LMaPh..77..183B}, global hyperbolicity is characterized by the existence of a smooth, Cauchy time function $\tau: \Mcal \to \Re$. A function $\tau$ on $\Mcal$ is a timefunction if $\nabla^a \tau$ is timelike everywhere, and it is Cauchy if the level sets $\Sigma_t = \tau^{-1}(t)$ are Cauchy surfaces. If $\tau$ is smooth, its levelsets are then smooth and spacelike. It follows that a globally hyperbolic spacetime $\Mcal$ is globally foliated by Cauchy surfaces, and in particular is diffeomorpic to a product $\Sigma \times \Re$. In the following, unless otherwise stated, we shall consider only globally hyperbolic spacetimes.

If a globally hyperbolic spacetime $\Mcal$ is a subset of a spacetime $\Mcal'$, then the boundary $\partial \Mcal$ in $\Mcal'$ is called the Cauchy horizon.  
\begin{example} \label{example:cone}
Let $O$ be the origin in Minkowski space, and let $\Mcal = I^+(O) = \{ t > r\}$ be its timelike future. Then $\Mcal$ is globally hyperbolic with Cauchy time function $\tau = \sqrt{t^2 - r^2}$. Further, $\Mcal$ is a subset of Minkowski space $\Mink$, which is a globally hyperbolic space with Cauchy time function $t$. Minkowski space is geodesically complete and hence inextendible. The boundary $\{t=r\}$ is the Cauchy horizon $\partial \Mcal$ of $\Mcal$. Past inextendible causal geodesics (i.e. past causal rays) in $\Mcal$ end on $\partial\Mcal$. In particular, $\Mcal$ is incomplete. However, $\Mcal$ is extendible, as a smooth flat spacetime, with many inequivalent extensions. 
\end{example} 
We remark that for a globally hyperbolic spacetime, which is extendible, the extension is in general non-unique. In the particular case considered in example \ref{example:cone}, $\Mink$ is an extension of $\Mcal$, which is also happens to be maximal and globally hyperbolic. In the vacuum case, there is a unique maximal globally hyperbolic extension, cf. section \ref{sec:cauchyproblem} below. However, a maximal extension is in general non-unique, and may fail to be globally hyperbolic. 

\subsection{Conventions and notation} 
We shall use mostly abstract indices, but will sometimes work with coordinate indices, and unless confusion arises will not be too specific about this. We raise and lower indices with $\met_{ab}$, for example $\xi^a = \met^{ab} \xi_b$, with $\met^{ab} \met_{bc} = \delta^a{}_c$, where $\delta^a{}_c$ is the Kronecker delta, i.e. the tensor with the property that $\delta^a{}_c \xi^c = \xi^a$ for any $\xi^a$. 

Let $\eps_{a\cdots d}$ be the Levi-Civita symbol, i.e. the skew symmetric expression which in any coordinate system has the property that $\eps_{1\cdots n} = 1$. The volume form of $\met_{ab}$ is $(\mu_\met)_{abcd} = \sqrt{|\met|} \eps_{abcd}$. 
Given $(\Mcal, \met_{ab})$ we have the canonically defined Levi-Civita covariant derivative $\nabla_a$. For a vector $\GenVec^a$, this is of the form 
$$
\nabla_a \GenVec^b = \partial_a \GenVec^b + \Gamma_{ac}^b \GenVec^c
$$
where $\Gamma_{ac}^b = \half \met^{bd} (\partial_a \met_{dc} + \partial_c \met_{db} - \partial_d g_{ac})$ is the Christoffel symbol. In order to fix the conventions used here, we recall that the Riemann curvature tensor is defined by 
$$
(\nabla_a \nabla_b - \nabla_b \nabla_a ) \xi_c = R_{abc}{}^d \xi_d
$$
The Riemann tensor $R_{abcd}$ is skew symmetric in the pairs of indices $ab, cd$, $R_{abcd} = R_{[ab]cd} = R_{ab[cd]}$, is pairwise symmetric $R_{abcd} = R_{cdab}$, and satisfies the first Bianchi identity $R_{[abc]d} = 0$. Here square brackets $[\cdots]$ denote antisymmetrization. We shall similarly use round brackets $(\cdots)$ to denote symmetrization.  
Further, we have $\nabla_{[a} R_{bc]de} = 0$, the second Bianchi identity. A contraction gives $\nabla^a R_{abcd} = 0$. 
The Ricci tensor is $R_{ab} = R^c{}_{acb}$ and the scalar curvature $R = R^a{}_a$. We further let $S_{ab} = R_{ab} - \frac{1}{4} R \met_{ab}$ denote the tracefree part of the Ricci tensor. The Riemann tensor can be decomposed as follows, 
\begin{align}
R_{abcd}={}&- \tfrac{1}{12} g_{ad} g_{bc} R
 + \tfrac{1}{12} g_{ac} g_{bd} R
 + \tfrac{1}{2} g_{bd} S_{ac}
 -  \tfrac{1}{2} g_{bc} S_{ad}
 -  \tfrac{1}{2} g_{ad} S_{bc}
 + \tfrac{1}{2} g_{ac} S_{bd}
 + C_{abcd}.
\end{align}
This defines the Weyl tensor $C_{abcd}$ which is a tensor with the symmetries of the Riemann tensor, and vanishing traces, $C^c{}_{acb} = 0$. Recall that $(\Mcal, \met_{ab})$ is locally conformally flat if and only if $C_{abcd} = 0$. It follows from the contracted second Bianchi identity that the Einstein tensor $G_{ab} = R_{ab} - \half R g_{ab}$ is conserved, $\nabla^a G_{ab} = 0$. 

\subsection{Einstein equation}
The Einstein equation in geometrized units with $G=c=1$, where $G, c$ denote Newtons constant and the speed of light, respectively, cf. \cite[Appendix F]{MR757180}, is the system 
\begin{equation} \label{eq:EFE} 
G_{ab} = 8\pi T_{ab}
\end{equation} 
This equation relates geometry, expressed in the Einstein tensor $G_{ab}$ on the left hand side, to matter, expressed via the energy momentum tensor $T_{ab}$ on the right hand side. For example, for a self-gravitating Maxwell field $F_{ab}$, $F_{ab} = F_{[ab]}$, we have 
$$
T_{ab} = \frac{1}{4\pi} (F_{ac}F_{bc} - \frac{1}{4} F_{cd} F^{cd} g_{ab} ).
$$
The source-free Maxwell field equations 
$$
\nabla^a F_{ab} = 0, \quad \nabla_{[a} F_{bc]} = 0
$$
imply that $T_{ab}$ is conserved, $\nabla^a T_{ab} = 0$. The contracted second Bianchi identity implies that $\nabla^a G_{ab} = 0$, and hence the conservation property of $T_{ab}$ is implied by the coupling of the Maxwell field to gravity.  These facts can be seen to follow from the variational formulation of Einstein gravity, given by the action 
$$
\Action = \int_{\Mcal} \frac{R}{16\pi} d\mu_\met - \int_{\Mcal} L_{\text{\rm matter}} d\mu_\met 
$$
where $L_{\text{\rm matter}}$ is the Lagrangian describing the matter content in the spacetime. In the case of Maxwell theory, this is given by 
$$
L_{\text{\rm Maxwell}} = \frac{1}{4\pi} F_{cd} F^{cd}
$$
Recall that in order to derive the Maxwell field equation, as an Euler-Lagrange equation, from this action, it is necessary to introduce a vector potential for $F_{ab}$, by setting $F_{ab} = 2 \nabla_{[a} A_{b]}$, and carrying out the variation with respect to $A_a$. It is a general fact that for generally covariant (i.e. diffeomorphism invariant) Lagrangian field theories which depend on the spacetime location only via the metric and its derivatives, the symmetric energy momentum tensor 
$$
T_{ab} = \frac{1}{\sqrt{g}} \frac{\partial L_{\text{\rm matter}}}{\partial g^{ab}} 
$$
is conserved when evaluated on solutions of the Euler-Lagrange equations.

As a further example of a matter field, we consider the scalar field, with action 
$$
L_{\text{\rm scalar}} = \half \nabla^c \psi \nabla_c \psi 
$$
where $\psi$ is a function on $\Mcal$. 
The corresponding energy-momentum tensor is  
$$
T_{ab} = \nabla_a \psi \nabla_b \psi - \half \nabla^c \psi \nabla_c \psi g_{ab}
$$
and the Euler-Lagrange equation is the free scalar wave equation 
\begin{equation} \label{eq:freewave} 
\nabla^a \nabla_a \psi = 0
\end{equation} 
As \eqref{eq:freewave} is another example of a field equation derived from a covariant action which depends on the spacetime location only via the metric $g_{ab}$ or its derivatives, the symmetric energy-momentum tensor is conserved for solutions of the field equation.

In both of the just mentioned cases, the energy momentum tensor satisfies the dominant energy condition, $T_{ab} \GenVec^a \zeta^b \geq 0$ for future directed causal vectors $\GenVec^a$, $\zeta^a$. This implies the null energy condition 
\begin{equation}\label{eq:NEC} 
R_{ab} \GenVec^a\GenVec^b \geq 0 \quad \text{ if $\GenVec_a\GenVec^a = 0$}. 
\end{equation}
These energy conditions hold for  most classical matter. 

There are many interesting matter systems which are worthy of consideration, such as fluids, elasticity, kinetic matter models including Vlasov, as well as fundamental fields such as Yang-Mills, to name just a few. We consider only spacetimes which satisfy the null energy condition, and for the most part we shall in these notes be concerned with the vacuum Einstein equations, 
\begin{equation} \label{eq:EVE} 
R_{ab} = 0
\end{equation} 

\subsection{The Cauchy problem}\label{sec:cauchyproblem}
Given a space-like hypersurface\footnote{If there is no room for confusion, we shall denote abstract indices for objects on $\Sigma$ by $a,b,c,\dots$.} $\Sigma$ in $\Mcal$ with timelike normal $T^a$, induced metric $h_{ab}$ and second fundamental form $k_{ab}$, defined by  $k_{ab}X^a Y^b = \nabla_a T_b X^a Y^b$ for $X^a, Y^b$ tangent to $\Sigma$, the Gauss, and Gauss-Codazzi equations imply the constraint equations 
\begin{subequations}\label{eq:constraints} 
\begin{align} 
R[h] + (k_{ab} h^{ab})^2 - k_{ab} k^{ab} =&{} 16\pi T_{ab} T^a T^b \\
\nabla[h]_a (k_{bc} h^{bc}) - \nabla[h]^b k_{ab} =&{} T_{ab} T^b
\end{align} 
\end{subequations}
A 3-manifold $\Sigma$ together with tensor fields $h_{ab}, k_{ab}$ on $\Sigma$ solving the constraint equations is called a Cauchy data set. 
The constraint equations for general relativity are analogues of the constraint equations in Maxwell and Yang-Mills theory, in that they lead to Hamiltonians which generate gauge transformations. 

Consider a 3+1 split of $\Mcal$, i.e. a 1-parameter family of Cauchy surfaces $\Sigma_t$, with a coordinate system $(x^a) = (t, x^i)$, and let 
$$
(\partial_t)^a  = N T^a + X^a
$$
be the split of $(\partial_t)^a$ into a normal and tangential piece. The fields $(N, X^a)$ are called lapse and shift. The definition of the second fundamental form implies the equation  
$$
\Lie_{\partial_t} h_{ab} = -2N k_{ab} + \Lie_X h_{ab}
$$
In the vacuum case, the Hamiltonian for gravity can be written in the form 
$$
\int N \mathcal H + X^a \mathcal J_a + \text{ boundary terms} 
$$
where $\mathcal H$ and $\mathcal J$ are the densitized left hand sides of \eqref{eq:constraints}. 
If we consider only compactly supported perturbations in deriving the Hamiltonian evolution equation, the boundary terms mentioned above can be ignored. However, for $(N, X^a)$ not tending to zero at infinity, and considering perturbations compatible with asymptotic flatness, the boundary term becomes significant, cf. section \ref{sec:asympflat}. 

The resulting Hamiltonian evolution equations, written in terms of $h_{ab}$ and its canonical conjugate $\pi^{ab} = \sqrt{h}(k^{ab} - (h^{cd} k_{cd} h^{ab}))$ are usually called the ADM evolution equations. 

Let $\Sigma \subset \Mcal$ be a Cauchy surface. Given functions $\phi_0, \phi_1$ on $\Sigma$ and $F$ on $\Mcal$, the Cauchy problem is the problem of finding solutions to the wave equation 
$$
\nabla^a \nabla_a \psi = F, \quad \psi \big{|}_{\Sigma} = \phi_0, \quad \Lie_{\partial_t} \psi \big{|}_{\Sigma} = \phi_1
$$
Assuming suitable regularity conditions, the solution is unique and stable with respect to initial data. This fact extends to a wide class of non-linear hyperbolic PDE's including quasi-linear wave equations, i.e. equations of the form 
$$
A^{ab}[\psi] \partial_a \partial_b \psi + B[\psi, \partial \psi] = 0
$$
with $A^{ab}$ a Lorentzian metric depending on the field $\psi$. 

Given a vacuum Cauchy data set, $(\Sigma, h_{ab}, k_{ab})$, a solution of the Cauchy problem for the Einstein vacuum equations is a spacetime metric $\met_{ab}$ with $R_{ab} = 0$, such that $(h_{ab}, k_{ab})$ coincides with the metric and second fundamental form induced on $\Sigma$ from $\met_{ab}$. Such a solution is called a vacuum extension of $(\Sigma, h_{ab}, k_{ab})$. 

Due to the fact that $R_{ab}$ is covariant, the symbol of $R_{ab}$ is degenerate. In order to get a well-posed Cauchy problem, it is necessary to either impose gauge conditions, or introduce new variables. A standard choice of gauge condition is the harmonic coordinate condition. Let $\hme_{ab}$ be a given metric on $\Mcal$. The identity map $\id: \Mcal \to \Mcal$ is harmonic if and only if the vector field
\newcommand{\hGamma}{\widehat{\Gamma}} 
\newcommand{\hmet}{\widehat{g}}
\newcommand{\hnabla}{\widehat{\nabla}}
$$
V^a = g^{bc} (\Gamma_{bc}^a - \hGamma_{bc}^a)
$$
vanishes. Here $\Gamma^a_{bc}$, $\hGamma^a_{bc}$ are the Christoffel symbols of the metrics $\met_{ab}, \hmet_{ab}$. Then $V^a$ is the tension field of the identity map $\id: (\Mcal, \met_{ab}) \to (\Mcal, \hmet_{ab})$. This is harmonic if and only if 
\begin{equation}\label{eq:Va=0}
V^a = 0 .
\end{equation}
Since harmonic maps with a Lorentzian domain are often called wave maps, the gauge condition \eqref{eq:Va=0} is sometimes called wave map gauge.  A particular case of this construction, which can be carried out if $\Mcal$ admits a global coordinate system $(x^a)$, is given by letting $\hmet_{ab}$ be the Minkowski metric defined with respect to $(x^a)$. Then $\hGamma_{bc}^a = 0$ and \eqref{eq:Va=0} is simply 
\begin{equation}\label{eq:wavecoord} 
\nabla^b \nabla_b x^a = 0 ,
\end{equation} 
which is usually called the wave coordinate gauge condition. 

Going back to the general case, let $\hnabla$ be the Levi-Civita covariant derivative defined with respect to $\hmet_{ab}$. We have the identity 
\begin{equation}\label{eq:Ric-reduct}
R_{ab} = - \half \frac{1}{\sqrt{g}}\hnabla_a \sqrt{g} \met^{ab} \hnabla_b \met_{ab} + S_{ab}[\met, \hnabla \met] + \nabla_{(a} V_{b)} 
\end{equation} 
where $S_{ab}$ is an expression which is quadratic in first derivatives $\hnabla_a \met_{cd}$. 
Setting $V^a = 0$ in \eqref{eq:Ric-reduct} yields $R^{\text{\rm harm}}_{ab}$, and  \eqref{eq:EVE} becomes a quasilinear wave equation 
\begin{equation}\label{eq:EVE-red}
R^{\text{\rm harm}}_{ab} = 0 .
\end{equation}
By standard results,  the equation \eqref{eq:EVE-red} has a locally well-posed Cauchy problem in Sobolev spaces $H^s$ for $s > 5/2$. Using more sophisticated techniques, well-posedness can shown to hold for any $s > 2$ \cite{MR2180400}. Recently a local existence has been proved under the assumption of curvature bounded in $L^2$ \cite{MR3402797}. Given a Cauchy data set $(\Sigma, h_{ab}, k_{ab})$, together with initial values for lapse and shift $N, X^a$ on $\Sigma$, it is possible to find $\Lie_t N, \Lie_t X^a$ on $\Sigma$ such that the $V^a$ are zero on $\Sigma$. A calculation now shows that due to the constraint equations, $\Lie_{\partial_t} V^a$ is zero on $\Sigma$. Given a solution to the reduced Einstein vacuum equation \eqref{eq:EVE-red}, one finds that $V^a$ solves a wave equation. This follows from $\nabla^a G_{ab} = 0$, due to the Bianchi identity. Hence, due to the fact that the Cauchy data for $V^a$ is trivial, it holds that $V^a = 0$ on the domain of the solution. Thus, in fact the solution to \eqref{eq:EVE-red} is a solution to the full vacuum Einstein equation \eqref{eq:EVE}. This proves local well-posedness for the Cauchy problem for the Einstein vacuum equation. This fact was first proved by Yvonne Choquet-Bruhat \cite{1952AcM....88..141F}, see \cite{2015CQGra..32l4003R} for background and history. 

Global uniqueness for the Einstein vacuum equtions was proved by Choquet-Bruhat and Geroch \cite{1969CMaPh..14..329C}. 
The proof relies on the local existence theorem sketched above, patching together local solutions. A partial order is defined on the collection of vacuum extensions, making use of the notion of common domain. The common domain $U$ of two extensions $\Mcal$, $\Mcal'$ is the maximal subset in $\Mcal$ which is isometric to a subset in $\Mcal'$. We can then define a partial order by saying that $\Mcal \leq \Mcal'$ if the maximal common domain is $\Mcal$. Given a partially ordered set, a maximal element exists by Zorn's lemma. This is proven to be unique by an application of the local well-posedness theorem for the Cauchy problem sketched above. For a contradiction, let $\Mcal$, $\Mcal'$ be two inequivalent extensions, and let $U$ be the maximal common domain. Due to the Haussdorff property of spacetimes, this leads to a contradiction. By finding a partial Cauchy surface which touches the boundary of $U$, see figure \ref{fig:extend}
and making use of local uniqueness, one finds a contradiction to the maximality of $U$. 
\begin{figure}[!h]
\begin{center}
\raisebox{-0.5\height}{\includegraphics{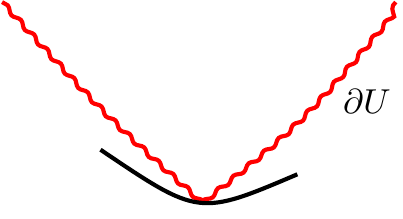}}
\caption{\ }
\label{fig:extend}

\end{center} 
\end{figure} 
It should be noted that here, uniqueness holds up to isometry, in keeping with the general covariance of the Einstein vacuum equations. These facts extend to the Einstein equations coupled to hyperbolic matter equations. See \cite{MR3447847} for a construction of the maximal globally hyperbolic extension which does not rely on Zorn's lemma, see also \cite{2013JMP....54k3511W}.
The global uniqueness result can be generalized to Einstein-matter systems, provided the matter field equation is hyperbolic and that its solutions do not break down. General results on this topic are lacking, see however \cite{2011arXiv1109.0644P} and references therein. The minimal regularity needed for global uniqueness is a subtle issue, which has not been fully addressed. In particular, results on local well-posedness are known, see eg. \cite{2012arXiv1204.1772K} and references therein, which require less regularity than the best results on global uniqueness.  

\subsection{Remarks} \label{sec:remarks} 
We shall now make several remarks relating to the above discussion. 

\subsubsection{Bianchi identities as a hyperbolic system} 
The vacuum Einstein equation $R_{ab} = 0$ implies that the Weyl tensor $C_{abcd}$ satisfies the Bianchi identity $\nabla^a C_{abcd} = 0$. This is the massless spin-2 equation. In particular, this is a first order hyperbolic system for the Weyl tensor. 

The spin-2 equation (i.e. the equation $\nabla^a W_{abcd}$ for a Weyl test field $W_{abcd}$ (i.e. a tensor field with the symmetries and trace properties of the Weyl tensor) implies algebraic conditions relating the field and the curvature. In particular, in a sufficentily general background a Weyl test field must be proportional to the Weyl tensor $C_{abcd}$ of the spacetime. This holds in particular for spacetimes of Petrov type D (cf. section \ref{sec:algspec} below for the definition of Petrov type), see \cite[\S 2.3]{2015arXiv150402069A} and references therein. 

One may view the Bianchi identity for the Weyl tensor as the main gravitational field equation, and the vacuum Einstein equation as type of ``constraint'' equation, which allows one to relate the Weyl tensor to the Riemann curvature of the spacetime. The first order system for the Weyl tensor can be extended to a first order system including the first and second Cartan structure equations. A hyperbolic system can be extracted by introducing suitable gauge conditions, see section \ref{sec:GSF}. 

\subsubsection{Null condition} 
Consider the Cauchy problem for the semilinear wave equation on Minkowski space, 
$$
\nabla^a \nabla_a \psi = Q^{ab} \nabla_a \psi \nabla_b \psi 
$$
with data $\psi\big{|}_{t=0} = \epsilon \psi_0$, $\partial_t \psi \big{|}_{t=0} = \epsilon \psi_1$, where $\epsilon > 0$ and $\psi_1, \psi_2$ are suitably regular functions. Solutions exist globally for small data (i.e. for sufficiently small $\epsilon > 0$) if and only if $Q^{ab}$ satisfies the null condition, $Q^{ab} \xi_a \xi_b = 0$ for any null vector $\xi^a$. 

An example due to Fritz John shows that the equation $\nabla^a \nabla_a \psi = |\partial_t \psi|^2$ for which the null condition fails, can have blowup for small data, cf. \cite{MR2455195}. 

Similar results hold also for quasilinear equations, in particular for quasilinear wave equations satisfying a suitable null condition, one has stability of the trivial solution. For the vacuum Einstein equation in harmonic coordinates, we have 
$$
R_{ab}^{\text{\rm harm}} = -\half g^{cd} \partial_c \partial_d g_{ab} + S_{ab}(g, \partial g)
$$
where the lower order term $S_{ab}$ contains terms of the form 
$\partial_a g_{cd} \partial_b g_{ef}g^{ce} g^{df}$, and hence the null condition fails to hold for the Einstein vacuum equation in harmonic coordinates. For this reason the problem of stability of Minkowski space in Einstein gravity is subtle. The stability of Minkowski space was first proved by Christodoulou and Klainerman \cite{MR1316662}. Later a proof using harmonic coordinates was given by Lindblad and Rodnianski \cite{2005CMaPh.256...43L}. This exploits the fact that the equation $R_{ab}^{\text{\rm harm}} = 0$ satisfies a weak form of the null condition. Consider the system 
\begin{subequations}\label{eq:weaknull} 
\begin{align} 
\nabla^a \nabla_a \psi =&{} |\partial_t \phi|^2 \\ 
\nabla^a \nabla_a \phi =&{} Q^{ab} \nabla_a \phi \nabla_b \phi
\end{align} 
\end{subequations} 
on Minkowski space, 
where $Q^{ab}$ has null structure. For this system, the 
null condition fails to hold. However, $\phi$ satisfies an equation with null structure and therefore has good dispersion. The equation for $\psi$ has a source defined in terms of $\phi$ but no bad self-interaction. One finds therefore that the solution to \eqref{eq:weaknull} exists globally for small data, but with slightly slower falloff than a solution of an equation satisfying the null condition. 

\subsubsection{Gauge source functions} \label{sec:GSF}
As has been pointed out by Helmut Friedrich, see \cite{1996CQGra..13.1451F} for discussion, one may introduce gauge source functions $V^a = F^a(x^b, \met^{cd})$ without affecting the reduction procedure. The gauge source functions can be designed to yield damping effects, or to control the evolution of the lapse and shift. This has frequently been used in numerical relativity. A related strategy is to add terms involving factors of the constraints $C^a$. Such terms vanish for a solution of the field equations, but may provide improved behavior for the reduced system. 

It is often convenient to introduce a suitably normalized tetrad $e_\ua{}^a$. Important examples are orthonormal tetrads, satisfying $e_\ua{}^a e_{\ub}{}^b \met_{ab} = \diag(+1,-1,-1,-1)$, and the null tetrads $(\NPl^a, \NPn^a, \NPm^a, \NPmbar^a)$ with $\NPl^a \NPn_a = 1$, $\NPm^a \NPmbar_a = -1$, all other inner products being zero. Such tetrads appear naturally when working with spinors, see section \ref{sec:prel}. 

The field equations can be written as a system of equations for tetrad components, connection coefficients and curvature. Introducing tetrad gauge source functions $V_{\ua \ub} = (\nabla^c \nabla^c e_\ua^a )e_{\ub}^b \met_{ab}$ it is possible to extract a first order symmetric hyperbolic system  with $V^a, V_{\ua\ub}$ taking values involving tetrad, connection coefficients and  curvature. This opens up a lot of interesting possibilities, but has not been widely used. The phantom gauge introduced by Chandrasekhar \cite[p. 240]{chandrasekhar:MR1210321} was shown in \cite{aksteiner:andersson:2011CQGra..28f5001A} to correspond to a tetrad gauge condition of the above type, and is therefore compatible with a well-posed Cauchy problem. 

Let $(\Mcal, \met_{ab})$ be a vacuum spacetime. Let $g(s)_{ab}$ be a one-parameter familiy of vacuum metrics and let 
$$
h_{ab} = \frac{d}{ds} \met(s)_{ab} \bigg{|}_{s=0} 
$$
Then $h_{ab}$ solves the linearized Einstein equation $DR_{ab}= 0$, where $DR_{ab}$ is the Frechet derivative of the Ricci tensor at $\met_{ab}$ in the direction $h_{ab}$. A calculation, cf.  \cite{aksteiner:andersson:2011CQGra..28f5001A}, shows that if we impose the linearized wave map gauge condition, then $h_{ab}$ satisfies the Lichnerowicz wave equation 
$$
\nabla^c \nabla_c h_{ab} + 2 R_{acbd} h^{cd} = 0
$$

\subsubsection{Asymptotically flat data} \label{sec:asympflat} 
The Kerr black hole represents an isolated system, and the appropriate data for the black hole stability problem should therefore be asymptotically flat. To make this precise we suppose there is a compact set $K$ in $\Mcal$ and a map $\Phi: \Mcal \setminus K \to \Re^3\setminus B(R,0)$, where $B(R,0)$ is a Euclidean ball. This defines a Cartesian coordinate system on the end $\Mcal \setminus K$ so that $h_{ab} - \delta_{ab}$ falls off to zero at infinity, at a suitable rate. Here $\delta_{ab}$ is the Euclidean metric in the Cartesian coordinate system constructed above. Similarly, we require that $k_{ab}$ falls off to zero. 

Let $x^a$ be the chosen Euclidean coordinate system and let $r$ be the Euclidean radius $r = (\delta_{ab} x^a x^b)^{1/2}$. Following Regge and Teitelboim \cite{1974AnPhy..88..286R}, see also \cite{beig:omurchadha:1987AnPhy.174..463B},
we assume that $\met_{ab} = \delta_{ab} + h_{ab}$ with 
\begin{align*} 
h_{ab} =&{} O(1/r), \quad \partial_a h_{bc} = O(1/r^2), \\
k_{ab} =&{} O(1/r^2) .
\end{align*}
Further, we impose the parity conditions 
\begin{equation}\label{eq:parity}
h_{ab}(x) = h_{ab}(-x), \quad  k_{ab} (x)= - k_{ab}(-x) .
\end{equation}
These falloff and parity conditions guarantee that the ADM 4-momentum and angular momentum are well defined. It was shown in \cite{2009CQGra..26a5012H} that data satisfying the parity condition conditions \eqref{eq:parity} are  dense among data which satisfy an asymptotic flatness condition in terms of weighted Sobolev spaces. 

Let $\xi^a$ be an element of the the Poincare Lie algebra and assume that $NT^a + X^a$ tends in a suitable sense to $\xi^a$ at infinity. 
Then the action for Einstein gravity can be written in the form 
$$
\int_\Mcal R d\mu_\met =  P_a \xi^a + \int \pi^{ij} \dot h_{ij}  - 
\int N \mathcal H + X^i \mathcal J_i
$$
Here we may view $P_a$ as a map to the dual of the Poincare Lie algebra, i.e. a momentum map. Evaluating $P_a \xi^a$ on a particular element of the Poincare Lie algebra gives the corresponding momentum. These can also be viewed as charges at infinity. 
We have 
\begin{subequations}\label{eq:ADMmom}
\begin{align} 
P^0  =&{} \frac{1}{16\pi} \lim_{r \to \infty} \int_{S_r}  (\partial_i g_{ji} - \partial_j g_{ii}) d\sigma^i  \\
P^i =&{}  \frac{1}{8\pi} \lim_{r\to \infty} \int_{S_r} \pi_{ij}  d\sigma^j
\end{align}
\end{subequations}  
where $d\sigma^i$ denotes the hypersurface area element of a family of spheres (which can be taken to be coordinate spheres) $S_r$ foliating a neighborhood of infinity.  See \cite{2011JMP....52e2504M} and references therein for a recent discussion of the conditions under which these expressions are well-defined. 
  
The energy and linear momentum $(P^0, P^i)$ provide the components of a 4-vector $P^a$, the ADM 4-momentum. 
Assuming the dominant energy condition, then under the above asymptotic conditions, $P^a$ is future causal, and timelike unless the maximal development $(\Mcal, \met_{ab})$ is isometric to Minkowski space.  
Further, $P^a$ transforms as a Minkowski 4-vector, and the ADM mass is given by $M = \sqrt{P^a P_a}$. The boost theorem \cite{1981CMaPh..80..271C} implies, given an asymptotically flat Cauchy data set, that one may find in a boosted slice $\Sigma'$ in its development such that the data is in the rest frame, i.e. $P^a  = M(\partial_t)^a$.

Since the constraint quantities $\mathcal H, \mathcal J_i$ vanish for solutions of the Einstein equations, the gravitational Hamiltonian takes the value $P_a \xi^a$, and hence the ADM mass and momenta defined by  \eqref{eq:ADMmom} are conserved for an evolution with lapse and shift $(N, X^i) \to (1,0)$ at infinity.  
If we consider the analog of the above definitions for a hyperboloidal slice which meets $\Scri$, then the ADM mass and momentum are replaced by the Bondi mass and momentum. An example of a hyperbolidal slice in Minkowski space is given by a level set of the time function $\Tcal$, cf. \eqref{eq:TcalRcal}, in the compactification of Minkowski space. For the Bondi 4-momentum, one has the important feature that gravitational energy is radiated through $\Scri$, which means that it is not conserved.  See \cite{MR1903925} and references therein for further details. 

\subsubsection{Killing initial data} \label{sec:KID}
A Killing initial data set, is a Cauchy data set $(\Sigma, h_{ab}, k_{ab})$ such that the development $(\Mcal, \met_{ab})$ is a spacetime with a Killing field $\GenVec^a$, i.e. 
$$
\Lie_\GenVec \met_{ab} = 2 \nabla_{(a} \GenVec_{b)} = 0
$$ 
Let now $\GenVec^a$ be a solution to the wave equation
$\nabla^a \nabla_a \GenVec_b = 0$, but not necessarily a Killing field. In a vacuum spacetime, we then have 
$$
\nabla^d \nabla_d ( \nabla_{(a} \GenVec_{b)} ) = 2 R^c{}_{(ab)}{}^d \nabla_{(c} \GenVec_{d)} .
$$
This implies that the tensor $\Lie_\GenVec \met_{ab}$ satisfies a wave equation, so if it has trivial Cauchy data on $\Sigma$, then $\GenVec^a$ is a Killing field in the domain of dependence of $\Sigma$. This allows us to  characterize Lie symmetries of a development $(\Mcal, \met_{ab})$ purely in terms of the Cauchy data. Another way to formulate this statement is that Lie symmetries propagate. This fact, which is closely related to the global uniqueness for the Cauchy problem, allows one to study symmetry restrictions of the Einstein equations. Much work has been done to study consistent subsystems of the Einstein equation, implied by imposing symmetries on the initial data. Examples include Bianchi, $T^2$, $U^1$. Note however, there are also the so-called surface symmetric spacetimes, which arise in a somewhat different manner. In addition, there are consistent subsystems which are not given by symmetry restrictions. Examples are the polarized Gowdy and half-polarized $T^2$. See \cite{MR2098914} and references therein for further details. 

The analog of the principle that symmetries propagate is also valid for spinors.  This leads to the notion of Killing spinor initial data, which is relevant for the problem of Kerr characterization, see \cite{backdahl:valiente-kroon:2010:MR2753388} for further details. 

\subsubsection{Komar integrals} \label{sec:Komar}
Assume that $\GenVec^a$ is a Killing vector field. Then we have $\nabla_a \GenVec_b = \nabla_{[a} \GenVec_{b]}$. A calculation shows 
$$
\nabla^a (\nabla_a \xi_b - \nabla_b \xi_a) = -2 R_{bc} \xi^c 
$$
Hence, in vacuum, 
$$
\int_Se_{abcd} \nabla^c \xi^d
$$ 
depends only on the homology class of the two-surface $S$. The analogous fact for the source free Maxwell equation, were we have $\nabla^a F_{ab} = 0$, $\nabla_{[a} F_{bc]} = 0$, is the conservation of the charge integrals 
$\int_S F_{ab}$, $\int_S \eps_{abcd} F^{cd}$, which again depend only on the homology class of $S$. These statements are immediate consequences of Stokes theorem. 

If we consider asymptotically flat spacetimes, we have in the stationary case, with $\xi^a = (\partial_t)^a$,   
$$
P^a \xi_a = - \frac{1}{8\pi} \int_{S} \eps_{abcd} \nabla^c \xi^d ,
$$ 
where on the left hand side we have the ADM 4-momentum evaluated at infinity. Similarly, in the axially symmetric case, with 
$\eta^a = (\partial_\phi)^a$, 
$$
J = - \frac{1}{16\pi} \int_{S} \eps_{abcd} \nabla^c \eta^d 
$$
These integrals again depend only on the homology class of $S$. 
See \cite[\S 6]{1994PhRvD..50..846I} for background to these facts. For a non-symmetric, but asymptotically flat spacetime, letting $S$ tend to infinity through a sequence of suitably round spheres   
yields the linkage integrals, which again reproduce the ADM momenta \cite{1965PhRvL..15..601W}.

\section{Black holes} \label{sec:blackholes}

\subsection{The Schwarzschild solution} \label{sec:schwarzschild}
Before introducing the Kerr solution, we will discuss the spherically symmetric, static Schwarzschild black hole spacetime. This exhibits some of the features of the Kerr solution and has the advantage that the algebraic form of the line element is much simpler. However, it must be noted that due to the fact that Schwarzschild is static, and spherically symmetric, the essential difficulties in analyzing field on the Kerr background stemming from the complicated trapping and superradiance are not seen in the Schwarzschild case. Therefore, one should be careful in generalizing notions from Schwarzschild to Kerr.  

In Schwarzschild coordinates $(t,r,\theta, \phi)$, the Schwarzschild metric takes the form 
\begin{equation}\label{eq:g_Schw}
g_{ab} dx^a dx^b =  f dt^2 - f^{-1} dr^2 - r^2 d\Omega^2_{S^2}
\end{equation}
with $f = 1-2M/r$. 
Here $d\Omega^2_{S^2} = d\theta^2 + \sin^2\theta d\phi^2$ is the line element on the unit 2-sphere. The coordinate $r$ is the area radius, defined by $4\pi r^2 = A(S(r,t))$, where $S(r,t)$ is the 2-sphere with constant $t,r$.  The line element given in equation \eqref{eq:g_Schw} is valid for $r > 0$, but has a coordinate singularity at $r=2M$, which is also the location of the event horizon. Historically, this fact caused some confusion, and was only fully cleared up in the 1950's due to the work of Kruskal and Szekeres, see eg. \cite[Chapter 31]{MTW} and references therein. The metric is in fact regular, and the line element given above is valid also for $0< r < 2M$. At $r=0$, there is a curvature singularity, where spacetime curvature diverges as $1/r^3$. The Schwarzschild metric is asymptotically flat and the parameter $M$ coincides with the ADM mass. 

We remark that by setting $f = 1- 2M/r + Q^2/r^2$, the Schwarzschild line element becomes that of Reissner-Nordstr\"om, a spherically symmetric solution to the Einstein-Maxwell equations, with field strength of the form $F^{tr} = Q/r^2$. Here $Q = \frac{1}{8\pi} \int_S F^{ab} d\sigma_{ab} $.

In order to get a better understanding of the Schwarzschild spacetime, it is instructive to consider its maximal extension. In order to do this, we first introduce the tortoise coordinate $\rs$, 
\begin{equation}\label{eq:rs(r)}
\rs = r+2M\log(\frac{r}{2M} - 1) .
\end{equation}
\begin{figure}[!h]
\raisebox{-0.5\height}{\includegraphics{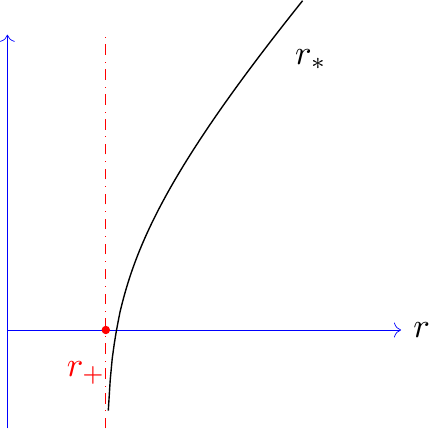}}
\end{figure} 

This solves 
$d\rs = f^{-1} dr$, $\rs(4M) = 4M$. As $r \searrow 2M$, $\rs$ diverges logarithmically to $-\infty$, and for large $r$, $\rs \sim r$. Inverting \eqref{eq:rs(r)} yields 
\newcommand{\LambertW}{\mathrm{W}}
\begin{equation}\label{eq:r(rs)}
r = 2M \LambertW \left (e^{\frac{\rs}{2M}-1}\right ) + 2M
\end{equation}
where $\LambertW$ is the principal branch of the Lambert W function\footnote{The Lambert W function, or product logarithm, is defined as the solution of $W(x)e^{W(x)} = x$ for $x > 0$. It satisfies $W'(x) = W(x)/((W(x)+1)x)$. The principal branch is analytic at $x=0$ and is real valued in the range $(-e^{-1},\infty)$ with values in $(-1,\infty)$. In particular, $W(0) = 0$. See \cite{LWref}.}. 
We can now introduce null coordinates 
$$
u = t - \rs, \quad v = t+\rs
$$
A null tetrad is given by 
\begin{align*} 
\NPl^a ={}& \sqrt{\frac{2}{f}} \partial_v^a , \\ 
\NPn^a ={}&  \sqrt{\frac{2}{f}}\partial_u^a , \\ 
\NPm^a ={}&  \frac{1}{\sqrt{2}r} (\partial_\theta^a + \frac{i}{\sin\theta} \partial_\phi^a)
\end{align*} 
On the exterior region in Schwarzschild, $(u,v)$ take values in the range $(-\infty, \infty)\times (-\infty,\infty)$. Let $\Ucal, \Vcal$ be a pair of coordinates taking values in $(-\pi/2, \pi/2)$, and related to $u,v$ by  
\begin{align*}  
u =&{} - 4M \log(-\tan \Ucal), \quad \Ucal \in (-\pi/2,0) \\ 
v =&{} 4M \log(\tan \Vcal), \quad \Vcal \in (0,\pi/2)
\end{align*} 

We have 
\begin{align*} 
t =&{} \half (v+u) \quad =  4M \log\left (-\tan \Vcal \tan \Ucal \right ) \\
\rs = &{} \half (v-u) \quad = 4M \log \left(- \frac{\tan \Vcal}{\tan\Ucal} \right ) 
\end{align*} 
In terms of $\Ucal, \Vcal$ we have 
\begin{equation}\label{eq:rUcalVcal}
r = 2M \LambertW(-e^{-1} \tan \Ucal \tan \Vcal) + 2M 
\end{equation} 
and $r > 0$ thus corresponds to $\tan\Ucal\tan \Vcal < 1$. 
The line element now takes the form 
\begin{equation}\label{eq:Krusk}
g_{ab} dx^a dx^b = \frac{d\Ucal d\Vcal}{\cos^2\Ucal \cos^2 \Vcal} \frac{32M^3}{r} e^{-\frac{r}{2M}} - r^2 d\Omega^2_{S^2}
\end{equation} 
The form \eqref{eq:Krusk} of the Schwarzschild line element is non-degenerate in the range 
\begin{equation} \label{eq:UVregion} 
(\Ucal, \Vcal) \in (-\pi/2,\pi/2) \times (-\pi/2, \pi/2) \cap \{-\pi/2 < \Ucal + \Vcal < \pi/2 \}.
\end{equation} 
In particular, the location $r=2M$ of the coordinate singularity in the line element \eqref{eq:g_Schw} corresponds to $\Ucal\Vcal = 0$. The line element \eqref{eq:Krusk} has a coordinate singularity, which is also a curvature singularity, at $r=0$ (corresponding to $\tan\Ucal \tan\Vcal = 1$),  and at $\Ucal = \pm \pi/2$, $\Vcal = \pm \pi/2$ (corresponding to $u,v$ taking unbounded values).  
\begin{figure}[!h]
\centering
\raisebox{-0.5\height}{\includegraphics{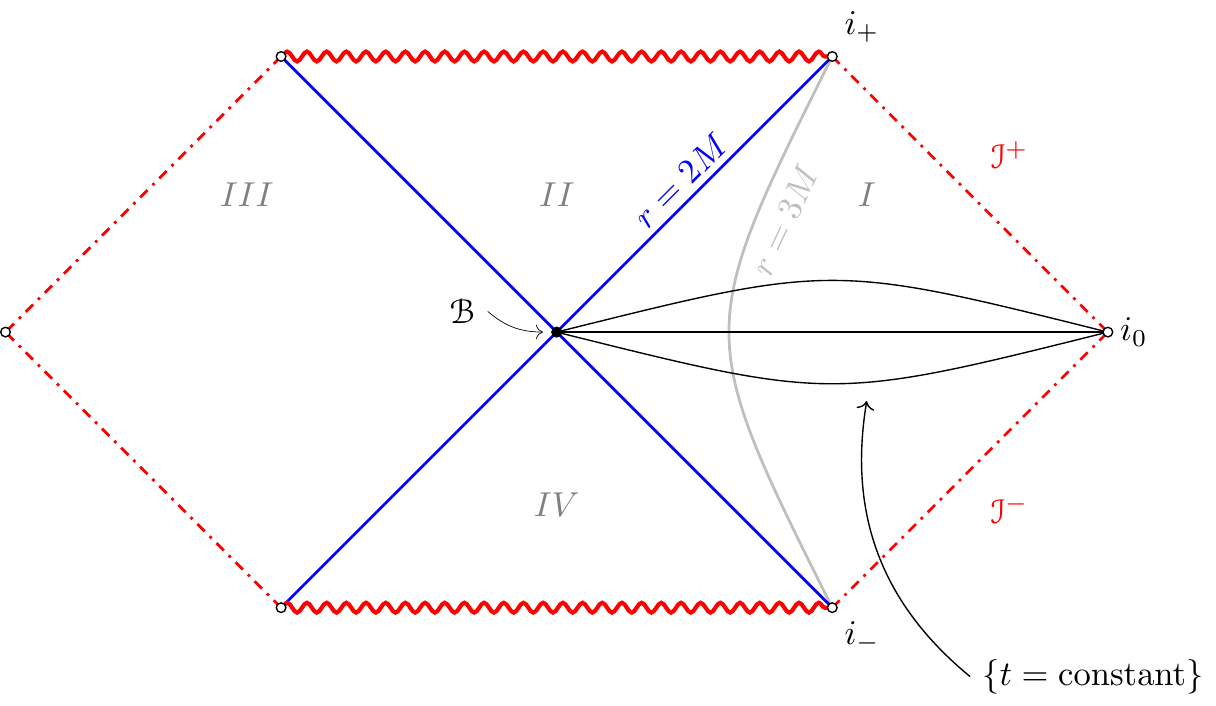}}
\caption{\ }
\label{fig:krusk}

\end{figure} 
Figure \ref{fig:krusk} shows the region given in \eqref{eq:UVregion}, with lines of constant  $t,r$ indicated. Using the causal diagram for the extended Schwarzschild solution, one can easily find the null infinities $\Scri^{\pm}$, spatial infinity $i_0$, timelike infinities $i_{\pm}$, the horizons $\Horizon^{\pm}$ at $r=2M$, which are indicated. Region $I$ is the domain of outer communication, i.e. $I^-(\Scri^+) \cap I^+(\Scri^-)$, while region $II$ is the future trapped (or black hole) region, $\Mcal^{\text{Schw}} \setminus I^-(\Scri^+)$. 

The level sets of $t$ hit the bifurcation sphere $\mathcal{B}$ located at $\Ucal = \Vcal = 0$, where $\partial_t = 0$. In particular, we see that the Schwarzschild coordinates are degenerate, since the level sets of $t$ do not foliate the extended Schwarzschild spacetime. On the other hand, a global Cauchy foliation of the maximally extended Schwarzschild spacetime is given by the level sets of the Kruskal time function $\Tcal = \half (\Vcal + \Ucal)$. 

\begin{wrapfigure}{l}{0.15\textwidth}
\begin{center}
\raisebox{-0.5\height}{\includegraphics{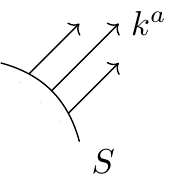}}
\end{center} 
\end{wrapfigure} 
Given a null vector $k^a$, perpendicular to a spacelike 2-surface $S$, 
we may define the null expansion with respect to $k^a$ by 
\begin{equation}\label{eq:nullexp-def}
\Theta_{k^a} = \half \delta_{k^a} \log(A(S))
\end{equation}
where $\delta_{k^a}$ denotes the variation in the direction $k^a$. 

Then $\Theta_{k^a}$ is the expansion of the area element of $S$, along the null geodesic with velocity $k^a$. If we let $k^a = (\partial_\Vcal)^a$, we have 
$$
\Theta_{k^a} \left\{ \begin{array}{ll} > 0 & \text{in region I}, \\
                                                 = 0  & \text{on $\Horizon_+$}, \\
                                                  < 0 & \text{in region II}
                                \end{array} \right.
$$
Thus, the area of a bundle of null rays in region I is expanding with respect to a future, outgoing null vector like $\partial_\Vcal$, while in region II,  they are contracting. Actually, in region II, we find that the expansion with respect to any future null vector is negative.

Although null vectors are conventionally drawn at 45\degree angles, due to the fact that each point in the causal diagram represents a sphere, this does not give a complete description. From the causal diagram it is clear that from each point in the DOC there are null curves which escape through $\Scri^{\pm}$ or fall in through the horizons $\Horizon^{\pm}$. By continuity, it is clear that there must be null curves which neither escape through $\Scri$ nor fall in through the horizons $\Horizon$. We refer to these as orbiting or \emph{trapped} null geodesics. In the Schwarzschild spacetime, the trapped null geodesics are located at $r=3M$, see figure \ref{fig:krusk}. The presence of trapped null geodesics is a robust feature of black hole spacetimes.   

Although the region covered by the null coordinates $\Ucal, \Vcal$ is compact, the line element \eqref{eq:Krusk} is of course isometric to the form given in \eqref{eq:g_Schw}. A conformal factor $\Phi=\cos\Ucal\cos\Vcal$  may now be introduced, which brings $\Scri^{\pm}$ to a finite distance. Letting $\tilde g_{ab} = \Phi^2 g_{ab}$, and adding these boundary pieces to $(\Mcal, \tilde g_{ab})$ provides a \emph{conformal compactification}\footnote{There are subtleties concerning the regularity of the conformal boundary of Schwarzschild, and the naive choice of conformal factor mentioned above does not lead to an analytic compactification. See \cite{2014CQGra..31a5007H} for recent developments. } of the maximally extended Schwarzschild spacetime.  

\subsubsection{Gravitational redshift} A robust fact about black hole spacetimes is that radiation emanating from near the event horizon is strongly red shifted before reaching infinity. In the limit as the source approaches the horizon, the redshift tends to infinity. 
Let $\dot \gamma^a$ be a null geodesic. The observed frequency of a plane fronted wave with wave plane perpendicular to $\dot \gamma^a$ is 
$$
\omega = \frac{\xi^a \dot \gamma_a}{(\xi^a \xi_a)^{1/2}}
$$
where $\xi^a \dot \gamma_a$ is conserved along the null geodesic. In Schwarzschild, $\xi^a \xi_a = f = (1-2M/r)$. If we let $\omega_1, \omega_2$ be the observed frequency at $r_1, r_2$, we find 
$$
\frac{\omega_2}{\omega_1} = \frac{1-2M/r_1}{1-2M/r_2} \quad \text{$\searrow 0$ as $r_1 \searrow 2M$}
$$
\subsubsection{Orbiting null geodesics}  
Consider a null geodesic $\gamma^a$ in the Schwarzschild spacetime. 
Due to the spherical symmetry of the Schwarzschild spacetime, we may assume without loss of generality that $\dot \theta = 0$ and set $\theta = \pi/2$, so that $\gamma^a$ moves in the equatorial plane. 
We have that  
the geodesic energy and azimuthal angular momentum  
$\GeodesicEnergy = -\xi^a \dot \gamma_a $ and $\GeodesicLz = \eta^a \dot \gamma_a$ are conserved. We have 
$$
\GeodesicLz = \eta^a \dot \gamma^b g_{ab} = r^2 \dot\phi
$$
In fact the same is true for the momenta corresponding to each of the three rotational Killing fields. Thus, we may consider the total squared angular momentum $\GeodesicLsquared$ given by 
\begin{equation}\label{eq:Lsquared}
\GeodesicLsquared = 2 r^2 \NPm_{(a} \NPmbar_{b)} \dot \gamma^a \dot \gamma^b = r^4 (g_{S^2})_{ab} \dot \gamma^a \dot \gamma^b. 
\end{equation} 
For geodesics moving in the equatorial plane, we have 
$\GeodesicLsquared = \GeodesicLz^2$. 
Rewriting $g_{ab} \dot \gamma^a \dot \gamma^b = 0$ using \eqref{eq:gNP} and these definitions gives  
\begin{equation}\label{eq:dotr}
\dot r^2 + V = \GeodesicEnergy^2
\end{equation} 
where 
$$
V = \frac{f}{r^2} \GeodesicLsquared .
$$ 
Equation \ref{eq:dotr} 
can be viewed as the equation for a particle moving in a potential $V$.  
\begin{figure}[!hb]
\raisebox{-0.5\height}{\includegraphics{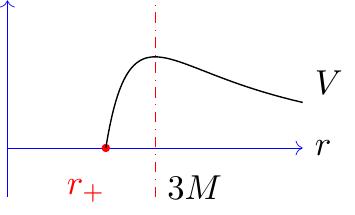}}
\end{figure} 

An analysis  
shows that $V$ has a unique critical point at $r=3M$, and hence a null geodesic with $\dot r = 0$ in the Schwarzschild spacetime must orbit at $r=3M$. We call such null geodesics trapped. The critical point $r=3M$ is a local maximum for $V$ and hence the orbiting null geodesics are unstable. The sphere $r=3M$ is called the \emph{photon sphere}. A similar analysis can be performed for massive particles orbiting the Schwarzschild black hole, see \cite[Chapter 6]{MR757180} for further details. 

The geometric optics correspondence between waves packets and null geodesics indicates that the phenomenon of trapped null geodesics is an obstacle to dispersion, i.e. the tendency for waves to leave every stationary region. For waves of finite energy, the fact that the trapped orbits are unstable can be used to show that such waves in fact disperse. This is a manifestation of the uncertainty principle. 

The close relation between the equation for radial motion of null geodesics and the wave equation $\nabla^a \nabla_a \psi = 0$ can be seen as follows. Equation 
\eqref{eq:dotr} can be written in the form  
\begin{align} 
r^4 \dot r^2 + \curlyR(r,\GeodesicEnergy, L) ={}& 0, \label{eq:Schw:nullgeodr} 
\end{align} 
where
\begin{equation} \label{eq:curlRSchw}
\curlyR = - r^4 \GeodesicEnergy^2 + r^2 f L 
\end{equation} 
On the other hand, the wave equation in the Schwarzschild exterior spacetime takes the form 
$$
r^2 \nabla^a \nabla_a \psi = \partial_r (r^2 f) \partial_r + \frac{\mathcal R}{r^2 f} 
$$
Here $\curlyR = \curlyR(r, \partial_t, \angDelta)$ where $\angDelta$ is the spherical Laplacian. This is the same expression as in the equation for the radial motion of null geodesics, but with $\GeodesicEnergy, L^2$ replaced by symmetry operators $\partial_t, \angDelta$, using the  correspondence $\GeodesicEnergy \leftrightarrow i\partial_t$, $L^2 \leftrightarrow -\angDelta$. If we perform separation of variables, the angular Laplacian $\angDelta$ is replaced by its eigenvalues $-\ell(\ell+1)$. 
This relation between the potential for radial motion of null geodesics and the term $\curlyR$ in the d'Alembertian is a curious and interesting fact, and importantly, this relation holds also in Kerr.

\subsection{Raychaudhouri equation and comparison theory} 
Assume that $k^a$ is a null vector field which generates affinely parametrized geodesics, $k^b \nabla_b k^a = 0$. Let 
\begin{equation}\label{eq:Thetadef}
\Theta = \half \nabla_a k^a
\end{equation} be the divergence, or null expansion\footnote{The definition of $\Theta$ in \eqref{eq:Thetadef} agrees with \eqref{eq:nullexp-def}, we have dropped the subindex on $\Theta$ to avoid clutter. The null expansion is often defined as $\nabla_a k^a$, however we shall here use the normalization as in \eqref{eq:Thetadef}.},
of the null congruence generated by $k^a$. For any $k^a$ as above, we have 
\begin{equation}\label{eq:nullRay}
k^a \nabla_a \Theta + \Theta^2 + \sigma\bar\sigma + \half R_{ab} k^a k^b = 0
\end{equation} 
where $\sigma\bar\sigma = \half(\nabla_{(a}k_{b)} \nabla^{(a} k^{b)} - \half (\nabla_a k^a)^2$ is the squared shear. 
Equation \eqref{eq:nullRay} describes the evolution of the null expansion along null geodesics $\gamma^a(\lambda)$ generated by $k^a$. 
Assuming the null energy condition \eqref{eq:NEC}, we have  
\begin{equation}\label{eq:Raych-ineq}
k^a \nabla_a \Theta + \Theta^2 \leq 0
\end{equation} 
Hence, if $\Theta\big{|}_S < c_0 < 0$, we find that $\Theta \searrow - \infty$ along $\gamma^a$ at some finite affine time $\lambda_0$. 

Recall that a geodesic in a Riemannian manifold ceases to be minimizing at  its first conjugate point. This can be shown by ``rounding off the corner'', which decreases length. In the Lorentzian case, ``rounding off the corner'', see fig \ref{fig:conjugate-rounding}, \emph{increases} Lorentzian length, and one finds that points along null geodesic $\gamma^a$ past $\gamma^a(\lambda_0)$ are timelike related to $\gamma^a (0)$. 
\begin{wrapfigure}{r}{0.15\textwidth}
\begin{center}
\raisebox{-0.5\height}{\includegraphics{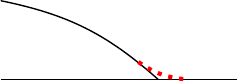}}
\end{center} 
\caption{\ }
\label{fig:conjugate-rounding}
\end{wrapfigure} 
This means that the geodesic in particular leaves the boundary of the causal future of $S$. It is known that any $p \in \partial J^+(S)$ is connected to $S$ by a null geodesic without conjugate points. Combining this argument with the inequality \eqref{eq:Raych-ineq} shows that if $\theta_{k^a} < c_0$ for some $c_0 < 0$, we find that the boundary of the causal future of $S$ can extend only for a finite affine parameter range. 

Now let $\Sigma$ be a spacelike Cauchy surface with future timelike normal $T^a$. 
\begin{wrapfigure}{r}{0.3\textwidth}
\begin{center}
\raisebox{-0.5\height}{\includegraphics{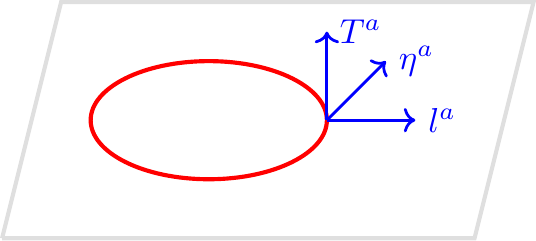}}
\end{center} 
\caption{\ }
\label{fig:notation}
\end{wrapfigure} 
For a 2-sided surface $S \subset \Sigma$, we say that a null normal $k^a$ to $S$ is outward pointing if the projection of $k^a$ to $\Sigma$ points into the exterior of $\Sigma$, i.e. the component of $\Sigma\setminus \Sigma$ connected to $E$. Let $\eta^a$ be the outward pointing normal to $S$ in $\Sigma$. Then $k^a = T^a + \eta^a$ is future directed and outward pointing. Let $H = \nabla_a \eta^a$ be the mean curvature of $S$ in $\Sigma$. Then $\Theta_{k^a} = \half(\tr_S k + H)$, where $\tr_S k = h^{ij} k_{ij} - k_{ij} \eta^i \eta^j$ is the trace of $k_{ij}$ restricted to $S$. See figure \ref{fig:notation}. 
If the outgoing null expansion $\Theta_{k^a}$ satisfies $\Theta_{k^a} = 0$ ($< 0$, $> 0$), we call $S$ is an  marginally outer trapped (trapped, untrapped) surface . 

Consider the Schwarzschild spacetime, see \ref{sec:schwarzschild}. If we designate the null vector $(\partial_{\Vcal})^a$ as outgoing, then the coordinate spheres $S_{t,r}$ are outer untrapped in regions $I, IV$, outer trapped in regions $II, III$, and marginally trapped on $\Horizon$

Due to their importance, we use the acronym MOTS for ``marginally outer trapped surface''. These are analogs of minimal surfaces in Riemannian geometry. In particular, a MOTS is critical with respect to variation of area along the outgoing null directions. For a stationary black hole spacetime, the event horizon is foliated by MOTS.

As an application of the above remarks, we have the following incompleteness result.
\begin{theorem}[\protect{\cite[\S 7]{2009CQGra..26h5018A}}] \label{thm:singular} 
Let $(\Mcal, g_{ab})$ be a globally hyperbolic spacetime satisfying the null energy conditon, and let $(\Sigma, h_{ij}, k_{ij})$ be a Cauchy surface in $(\Mcal, g_{ab})$ with non-compact exterior. Assume that $S$ is outer trapped in the sense that the outgoing null expansion $\theta$ of $S$ satisfies $\theta < c_0 < 0$ for some $c_0 < 0$. Then $(\Mcal, g_{ab})$ is causally geodesically incomplete. 
\end{theorem}
\begin{remark} Results similar to theorem  \ref{thm:singular} are usually referred to as ``singularity theorems'', but actually demonstrate that the spacetime $\Mcal$ has a nontrivial Cauchy horizon $\partial \Mcal$, without giving any information about its properties. Versions of such results were originally proved by Hawking and Penrose, see \cite{MR0424186}. Motivated by the strong cosmic censorship conjecture, one expects that for a generic spacetime, the spacetime metric becomes irregular  as one approaches $\partial \Mcal$, and hence that a regular extension beyond $\partial \Mcal$ is impossible. For example, in the Schwarzschild spacetime, curvature diverges as $1/r^3$ as one approaches the Cauchy horizon at $r=0$. This can be seen by looking at the invariantly defined Kretschmann scalar $R_{abcd} R^{abcd} = 48M/r^6$. 

The detailed behavior of the geometry at the Cauchy horizon in generic situations is subtle and far from understood, see however \cite{2013arXiv1311.4970L} and references therein for recent developments. For cosmological singularities, strong cosmic censorship including curvature blowup for generic data has been established in some symmetric situations, see \cite[\S 5.2]{2015CQGra..32l4003R} and references therein. 

By the weak cosmic censorship conjecture, one expects that in a generic asymptotically flat spacetime, $\partial \Mcal$ is hidden from observers at infinity, and hence that the domain of outer communication has a non-trivial boundary, the event horizon. This motivates the idea that MOTS may be viewed as representing the \emph{apparent horizon} of a black hole, see section \ref{sec:AH} below. Due to the fact that the MOTS can be understood in terms of Cauchy data, this point of view is important in considering dynamical black holes.

\end{remark}

\subsection{The apparent horizon} \label{sec:AH}
Consider the Vaidya line element, cf. \cite[\S 5.1.8]{poisson:toolkit}
\begin{equation}\label{eq:vaidya-metric} 
ds^2 = fdv^2- 2dv dr - r^2 d\Omega^s_{S^2}
\end{equation} 
with
$f = 1-2M(v)/r$, where the \emph{mass aspect function} $M(v)$ is an increasing function of the retarded time coordinate $v$.
The matter in the Vaidya spacetime is infalling null dust. Those regions where $dM/dv = 0$ are empty.  
We see that there is no $dr^2$ term in \eqref{eq:vaidya-metric}, so $r$ is a null coordinate. 
Setting $M(v) \equiv M$, gives the  
Schwarzschild line element in ingoing Eddington-Finkelstein coordinates. A calculation shows that there are MOTS located at $r=2M(v)$. Hence, if 
$M(v)$ varies from $M_1$ to $M_2$ in an interval $(v_1, v_2)$ 
\begin{figure}[!hb]

\raisebox{-0.5\height}{\includegraphics{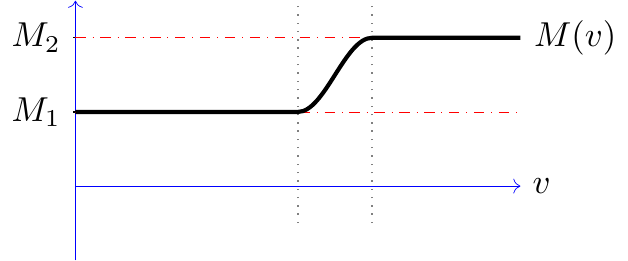}}

\caption{}
\label{fig:vadyam(v)}

\end{figure} 
and is constant elsewhere, we find that the  MOTS move outwards, to the event horizon, located at $r=2M_2$. 

In general, the spacetime tube swept out by the MOTS might, provided it exists, be termed a marginally outer trapped tube (MOTT). By known stability results for MOTS, this exists locally in generic situations, see \cite{2005PhRvL..95k1102A}, see also section \ref{sec:MOTSresults} below. Thus, heuristically the MOTS and MOTT represent the apparent horizon, 
and the fact that the apparent horizon moves outward corresponds to the growth of mass of the black hole due to the stress-energy or gravitational energy crossing the horizon, see figure \ref{fig:vaidya-apparenthorizon}.   
\begin{figure}[!hb]

\raisebox{-0.5\height}{\includegraphics{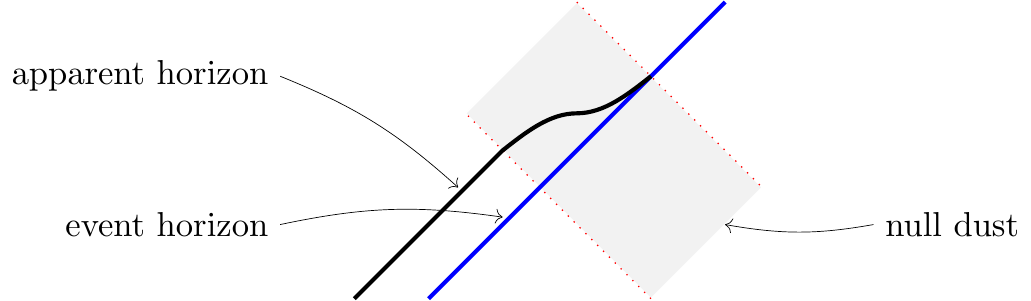}}

\caption{Event and apparent horizons in the Vaidya spacetime.}
\label{fig:vaidya-apparenthorizon}

\end{figure} 

\begin{remark}
\begin{enumerate}
\item 
The event horizon is teleological, in the sense that determining its location requires complete knowledge of spacetime. In particular, it is not possible to compute its location from Cauchy data without constructing the complete spacetime evolution. On the other hand, the notion of MOTS and apparent horizon are quasilocal notions, which can be determined directly from Cauchy data. 
\item The location of MOTS is not a spacetime concept but depends on the choice of Cauchy slicing. See \cite{2011PhRvD..83d4012B} for results on the region of spacetime containing trapped surfaces. It was shown by Wald and Iyer \cite{1991PhRvD..44.3719W} that there are Cauchy surfaces in the extended Schwarzschild spacetime which approach the singularity arbitrarily closly and such that the past of these Cauchy surfaces do not contain any outer trapped surfaces.  
\item The interior of the outermost MOTS is called the trapped region (a notion which depends on the Cauchy slicing). Based on the weak cosmic censorship conjecture, and the above remarks, one expects this to be in the black hole region, which is bounded by the event horizon. See \cite[Theorem 6.1]{2009AnHP...10..893C} for a result in this direction. 
\end{enumerate}
\end{remark}

\subsection{Results on MOTS and the trapped region} \label{sec:MOTSresults}
 Several theorems about MOTS have been proved in the last decade. In particular, if a Cauchy surface $\Sigma$ contains a MOTS, then there is an \emph{outermost} MOTS. 
\begin{figure}[!ht]

\raisebox{-0.5\height}{\includegraphics{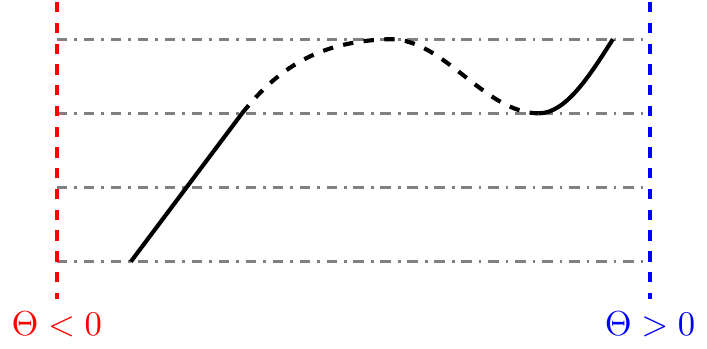}}

\caption{}

\label{fig:mott3}

\end{figure} 

If we conside a Cauchy slicing $(\Sigma_t)$, then if $\Sigma_{t_0}$ contains a MOTS, then for $t > t_0$, $\Sigma_t$ contains a MOTS. However, the location of the MOTS may jump, eg. due to the formation of a MOTS surrounding the previous one, see figure \ref{fig:mott3}. 
This phenomenon is seen in numerical simulations of colliding black holes, cf. \cite{2015CQGra..32w5003M}. There, examples with two merging black holes are considered. When the apparent horizons of the two black holes are sufficiently close together, a new apparent horizon surrounding both is formed, in accordance with the results in \cite{Andersson-Metzger:2009,2009CQGra..26h5018A}.

If the NEC holds, then in a generic situation the MOTT is spacelike \cite{2009CQGra..26h5018A}, and hence from the point of view of the exterior part of $\Mcal$ it represents an outflow boundary. This means that it is not necessary to impose any boundary condition on the MOTT in order to get a well-posed Cauchy problem. This leads to the \emph{exterior Cauchy problem}. As mentioned above, cf. figure \ref{fig:exterior2}, in strong field situations, it can happen that the MOTS jumps out. In this case, one must then restart solving exterior Cauchy problem at the jump time.  This corresponds closely to what one sees in a numerical evolution of strong field situations, eg. of merging black holes, when using horizon trackers to determine the location of MOTS. 
\begin{figure}[!ht]

\raisebox{-0.5\height}{\includegraphics{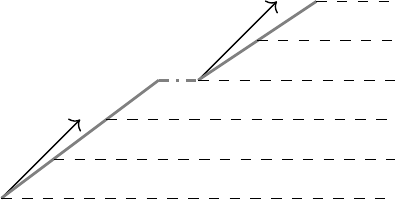}}

\caption{The exterior Cauchy problem}
\label{fig:exterior2}

\end{figure}

\subsection{Formation of black holes} The first example of a dynamically forming black hole through the collapse of a cloud of dust, was constructed by Oppenheimer and Snyder \cite{1939PhRv...56..455O} in 1939. Examples of the formation of a black hole by concentration of gravitational radiation was constructed by Christodoulou \cite{2009fbhg.book.....C}. There has been much recent work refining and extending result, see \cite{2013arXiv1302.5951K} and references therein. 

In order to understand the formation of black holes, it is important to have good conditions for the existence of marginally outer trapped surfaces in a given Cauchy surface. Such results have been proved by Schoen and Yau \cite{Schoen-Yau:1983:cond}, see also \cite{Clarke:1988}. The result in \cite{Schoen-Yau:1983:cond} makes use of Jang's equation to show that MOTS form if a sufficiently dense concentration of matter is present. A related result for the vacuum case is given in \cite{Eardley:1995}, see also 
\cite{Yau:2001}. 

\subsection{Black hole stability} \label{sec:BHstab}
Taking the trapped region as representing a dynamical black hole, the above discussion leads to a picture of the evolution dynamical black holes, as well as their formation. 
Based on these general considerations, we can now give a heuristic formulation of the black hole stability problem, and related conjectures. Recall that the Kerr black hole spacetime, which we shall study in detail below, is conjectured to be the unique rotating vacuum black hole spacetime, and further to be dynamically stable. 

The \emph{black hole stability conjecture} is that Cauchy data sufficiently close, in a suitable sense, to Kerr Cauchy data\footnote{See \cite{2015arXiv150402069A}, see also eg.  \cite{backdahl:valente-kroon:2011RSPSA.467.1701B, 2016arXiv160305839M} for discussions of the problem of characterizing Cauchy data as Kerr data.} have a maximal development which is future asymptotic to a Kerr spacetime, see figure \ref{fig:BHstabScritrap}. In approaching this problem, one may use the results on the evolution of MOTS mentioned above, cf. section 
\ref{sec:MOTSresults} to consider only the exterior Cauchy problem. 
\begin{figure}[!ht]
\raisebox{-0.5\height}{\includegraphics{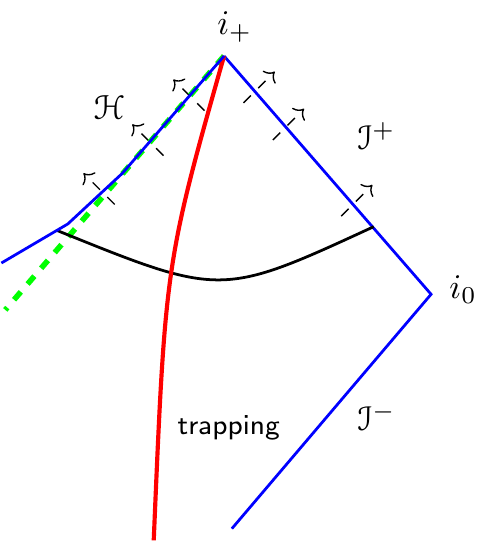}}
\caption{}
\label{fig:BHstabScritrap}
\end{figure} 
It is important to note that the parameters of the ``limiting'' Kerr spacetime cannot be determined in any effective manner from the initial data.  

As discussed above, cf. section \ref{sec:Komar}, if we restrict to axial symmetry, then angular momentum is quasi-locally conserved. This means that if we further restrict to zero angular momentum, the end state of the evolution must be a Schwarzschild black hole. 

Thus, the \emph{black hole stability conjecture for the axially symmetric case} is that the maximal development of sufficiently small (in a suitable sense), axially symmetric, deformations of Schwarzschild Cauchy data is asymptotic to the future to a Schwarzschild spacetime. In this case, due to the loss of energy through $\Scri^+$, the mass of the ``limiting'' Schwarzschild black hole cannot be determined directly from the Cauchy data. 

A conjecture related to the black hole stability conjecture, but which is even more far reaching may be termed  \emph{the end state conjecture}. Here the idea is that the maximal evolution of generic asymptotically flat vacuum initial data is asymptotic in a suitable sense, to a collection black holes moving apart, with the near region of each black hole approaching a Kerr geometry. No smallness condition is implied. 

The heuristic ideas relating to weak cosmic censorship and Kerr as the final state of the evolution of an isolated system, together with Hawking's area theorem was used by Penrose to motivate the Penrose inequality, 
$$
\sqrt{\frac{A_{\text{min}}}{16\pi}} \leq M_{\text{ADM}}
$$
were $A_{\text{min}}$ is the minimal area of any surface surrounding all past and future trapped regions in a given Cauchy surface, and $M_{\text{ADM}}$ is the ADM mass at infinity. The Riemannian version of the Penrose inequality has been proved by Bray \cite{Bray:2001}, and Huisken and Ilmanen \cite{Huisken-Ilmanen:2001}. The spacetime version of the Penrose inequality remains open. It should be stressed that the formulation of the inequality given above may have to be adjusted. Interesting possible approaches to the problem have been developed by Bray and Khuri, see \cite{2014arXiv1409.0067H} and references therein.

\subsection{The Kerr metric} \label{sec:kerrmetric}
In this section we shall discuss the Kerr metric, which is the main object of our considerations. Although many features of the geometry and analysis on black hole spacetimes are seen in the Schwarzschild case, there are many new and fundamental phenomena persent in the Kerr case. Among those are complicated trapping, i.e. the fact that trapped null geodesics fill up an open spacetime region, the fact that the Kerr metric admits only two Killing fields, but a hidden symmetry manifested in the Carter constant, and the fact that the stationary Killing vector field $\xi^a$ fails to be timelike in the whole domain of outer communications, which leads to a lack of a positive conserved energy for waves in the Kerr spacetimes. This fact is the origin of superradiance and the Penrose process. See \cite{2015CQGra..32l4006T} for a recent survey.

The Kerr metric describes a family of stationary, axisymmetric, asymptotically flat vacuum 
spacetimes, parametrized by ADM mass $M$ and angular momentum per unit mass $a$. The expressions for mass and angular momentum introduced in section \ref{sec:remarks} when applied in Kerr geometry yield $M$ and $J = aM$.  
In Boyer-Lindquist coordinates $(t,r,\theta,\phi)$, the Kerr metric takes the form
\begin{align}
\met_{ab}={}&\frac{(\Delta - a^2 \sin^2\theta) dt_{a} dt_{b}}{\Sigma}
 -  \frac{\Sigma dr_{a} dr_{b}}{\Delta}
 -  \Sigma d\theta_{a} d\theta_{b}
 - \frac{\sin^2\theta \bigl((a^2 + r^2)^2 - a^2 \sin^2\theta \Delta\bigr) d\phi_{a} d\phi_{b}}{\Sigma}\nonumber\\
& + \frac{2 a \sin^2\theta (a^2 + r^2 -  \Delta) dt_{(a}d\phi_{b)}}{\Sigma},\label{eq:met}
\end{align}
where
$\Delta = a^2 - 2 M r + r^2$ and $\Sigma = a^2 \cos^2\theta + r^2$.
The volume element is 
\begin{equation}\label{eq:Kerrvol}
\sqrt{|\det g_{ab}|} = \Sigma \sin\theta 
\end{equation} 
There is a ring-shaped singularity at $r=0$, $\theta = \pi/2$. 
For $|a|\leq M$, the Kerr spacetime contains a black hole, with event horizon at $r=r_+ \equiv M + \sqrt{M^2 - a^2}$, while for $|a|>M$, the singularity is naked in the sense that it is causally connected to observers at infinity. The area of the horizon is $A_
{\text{Hor}} = 4 \pi (r_+^2 + a^2)$. This achieves its maximum of $16\pi M^2$ when $a=0$, providing one of the ingredients in the heuristic argument for the Penrose inequality, see section \ref{sec:BHstab}. 
The case $|a| = M$ is called extreme. We shall here be interested only in the subextreme case, $|a| < M$, as this is the only case where we expect black hole stability to hold.

The Boyer-Lindquist coordinates are analogous to the Schwarzschild coordinates section \ref{sec:schwarzschild} and upon setting $a=0$,  \eqref{eq:met} reduces to \eqref{eq:g_Schw}. The line element takes a simple form in Boyer-Lindquist coordinates, but similarly for the Schwarzschild coordinates, the Boyer-Lindquist coordinates have the drawback that they are not regular at the horizon.   

The Kerr metric admits two Killing vector fields $\xi^a = (\partial_t)^a$ (stationary) and 
$
\eta^a = (\partial_\phi)^a$ (axial). Although the stationary Killing field $\xi^a$ is timelike near infinity, since $\met_{ab} \xi^a\xi^b \rightarrow 1$ as $r\rightarrow
\infty$, $\xi^a$ becomes spacelike for $r$ sufficiently small,
when $1-2M/\KSigma<0$. 
In the Schwarzschild case $a=0$, this occurs at
the event horizon $r=2M$. However, for a rotating Kerr black hole with
$0<|a|\leq M$, there is a region, called the ergoregion, outside the
event horizon where $\xi^a$ is spacelike. The ergoregion is bounded by the surface $M+\sqrt{M^2-a^2\cos^2\theta}$ which touches the horizon at the poles $\theta=0,\pi$, see figure \ref{fig:ergoregion}.
\begin{figure}[h!]%
\centering
\raisebox{-0.5\height}{\includegraphics{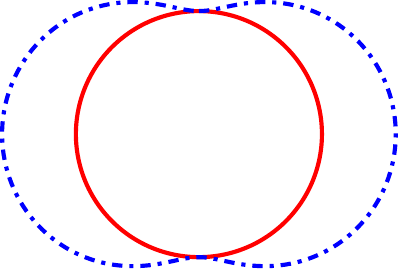}}
\centering
\caption{The ergoregion}\label{fig:ergoregion}
\end{figure} 
In the ergoregion, null and
timelike geodesics can have negative energy with respect to $\xi^a$. The fact that there is no globally timelike vectorfield in the Kerr exterior is the origin of superradiance, i.e. the fact that waves which scatter off the black hole can leave the ergoregion with larger energy (as measured by a stationary observer at infinity) than was sent in. This effect was originally found by an analysis based on separation of variables, but can be demonstrated rigorously, see \cite{2009CMaPh.287..829F}. However, it is a subtle effect and not easy to demonstrate numerically, see \cite{Laszlo:2012zz}. 

If we consider a dynamical spacetime containing a rotating black hole, then the presence of the ergoregion allows for the Penrose process, which extracts rotational energy from the black hole, see \cite{2014PhRvD..89b4041L}, see also \cite{2014PhRvD..89f1503E} for a numerical study of superradiance of graviational waves in a dynamical spacetime. 

Let $\omega_H = a/(r_+^2 + a^2)$ be the rotation speed of the black hole. The Killing field $\chi^a = \xi^a + \omega_H \eta^a$ is null on the event horizon in Kerr, which is therefore a Killing horizon. For $|a| < M$, there is a neighborhood of the horizon in the black hole exterior where $\chi^a$ is timelike. The surface gravity $\kappa$, defined by $\kappa^2 = -\half (\nabla^a \chi^b)(\nabla_a \chi_b)$ takes the value $\kappa = (r_+ - M)/(r_+^2 +a^2)$, and is in the subextreme case $|a| < M$ nonzero. By general results, a Killing horizon with non-vanishing surface gravity is bifurcate, i.e. there is a cross-section where the null generator vanishes. In the Schwarzschild case, this is the 2-sphere $\Ucal=\Vcal=0$. 
See \cite{ONeill,poisson:toolkit} for background on the geometry of the Kerr spacetime, see also \cite{FrolovNovikov}.  

\begin{TOOLONG} 
\red{

\begin{itemize} 

Penrose scenario: sketch expanding horizon: Hawking area theorem $A(S)$ increases to future 

Hawking area theorem proved by Chrusciel, Delay, Galloway, Howard

area increases along horizon, Bondi mass decreases along Scri, this leads to Penrose inequality 

replace $A_H$ by minimizing hull of any MOTS/MITS, proved in Riemannian case, open in general. Rigidity aspect in terms of Bondi mass important for BH stability -- understanding how spacetime approaches Kerr BH as you move to the future (take into account angular momentum)

\end{itemize} 

} 
\end{TOOLONG}

\section{Spin geometry} \label{sec:prel} 
The 2-spinor formalism, and the closely related GHP formalism, are important tools in Lorentzian geometry and the analysis of black hole spacetimes, and we will introduce them here. A detailed of this material is given by Penrose and Rindler \cite{Penrose:1986fk}.  Following the conventions there, we use the abstract index notation with lower case latin letters $a,b,c,\dots$ for tensor indices, and unprimed and primed upper-case latin letters $A,B,C, \dots, A', B', C', \dots$ for spinor indices. Tetrad and dyad indices are boldface latin letters following the same scheme,  $\ba,\bb,\bc,\dots, \bA,\bB,\bC,\dots,\bA',\bB',\bC',\dots$. For coordinate indices we use greek letters $\alpha, \beta, \gamma, \dots$.  

\subsection{Spinors on Minkowski space} 
Consider Minkowski space $\Mink$, i.e. $\Re^4$ with coordinates $(x^\alpha) = (t, x, y, z)$ and metric 
$$
\met_{\alpha\beta} dx^\alpha dx^\beta = dt^2 - dx^2 - dy^2 - dz^2.
$$
Define a complex null tetrad (i.e. frame) $(\met_{\ba}{}^a)_{\ba = 0,\cdots, 3} = (\NPl^a, \NPn^a, \NPm^a, \NPmbar^a) $, as in \eqref{eq:Minktetrad} above, 
normalized so that $\NPl^a \NPn_a = 1$, $\NPm^a \NPmbar_a = -1$, 
so that 
\begin{equation}\label{eq:metNP} 
\met_{ab} = 2 (\NPl_{(a} \NPn_{b)} - \NPm_{(a} \NPmbar_{b)}). 
\end{equation} 
Similarly, let $\eps_\bA{}^A$ be a dyad (i.e. frame) in $\Co^2$, with dual frame $\eps_A{}^\bA$.
The complex conjugates will be denoted $\bar\eps_{\bA'}{}^{A'}, \bar\eps_{A'}{}^{\bA'}$ and again form a basis in another 2-dimensional complex space denoted $\bar \Co^2$, and its dual.  We can identify  the space of complex $2\times2$ matrices with $\Co^2 \otimes \bar \Co^2$. By construction, the tensor products $\eps_\bA{}^A \bar\eps_{\bA'}{}^{A'}$ and $\eps_A{}^\bA \bar\eps_{A'}{}^{\bA'}$forms a basis in $\Co^2 \otimes \bar \Co^2$ and its dual.  

Now, with $x^\ba = x^a \met_a{}^\ba $,  writing 
\begin{equation}\label{eq:IvW} 
x^\ba \met_\ba{}^{\bA \bA'} \equiv \begin{pmatrix} x^0 & x^2 \\ x^3 & x^1 \end{pmatrix} 
\end{equation} 
defines the soldering forms, also known as Infeld-van der Waerden symbols $\met_a{}^{AA'}$, (and analogously $\met_{AA'}{}^a$). 
By a slight abuse of notation we may write $x^{AA'} = x^a$ instead of $x^{\bA\bA'} = x^\ba \met_\ba{}^{\bA\bA'}$ or, dropping reference to the tetrad, $x^{AA'} = x^a \met_a{}^{AA'}$. 
In particular, we have that $x^a \in \Mink$ corresponds to a $2\times 2$ complex Hermitian 
matrix $x^{\bA\bA'} \in \Co^2 \otimes \bar \Co^2$.
Taking the complex conjugate of both sides of \eqref{eq:IvW} gives 
$$
\bar x^a = \bar x^{A'A} = (x^{AA'})^* .
$$
where $*$ denotes Hermitian conjugation. This extends to a correspondence $\Co^4 \leftrightarrow \Co^2 \otimes \bar \Co^2$ with complex conjugation corresponding to Hermitian conjugation. 

Note that 
\begin{equation}\label{eq:detmet}
\det(x^{\bA\bA'}) = x^0 x^1 - x^2 x^3 = x^a x_a /2.
\end{equation} 

We see from the above that the group 
$$
\SL(2,\Co) = \Bigl\{ A = \begin{pmatrix} a & b \\ c & d \end{pmatrix}, \quad a,b,c,d \in \Co, \quad ad-bc = 1 \Bigr\}
$$
acts on $X \in \Co^2 \otimes \bar \Co^2$ by 
$$
X \mapsto A X A^* .
$$
In view of \eqref{eq:detmet} this exhibits $\SL(2,\Co)$ as a double cover of the identity component of the Lorentz group $\SO_0(1,3)$, the group of linear isometries of $\Mink$. 
In particular, $\SL(2,\Co)$ is the spin group of $\Mink$. The canonical action
$$
(A, v) \in \SL(2,\Co) \times \Co^2 \mapsto A v \in \Co^2
$$
of $\SL(2,\Co)$ on $\Co^2$ is the spinor representation. Elements of $\Co^2$ are called (Weyl) spinors. The conjugate representation given by 
$$
(A, v) \in \SL(2,\Co) \times \Co^2 \mapsto \bar A v \in \Co^2
$$
is denoted $\bar\Co^2$.

Spinors\footnote{It is conventional to refer to spin-tensors eg. of the form $x^{AA'}$ or $\psi_{ABA'}$ simply as spinors.} 
of the form $x^{AA'} = \alpha^A \beta^{A'}$ correspond to matrices of rank one, and hence to complex null vectors. Denoting $o^A = \eps_{\mathbf 0}{}^A, \iota^A = \eps_{\mathbf 1}{}^A$, we have from the above that 
\begin{equation}\label{eq:tetrad-dyad} 
\NPl^a = o^A o^{A'}, \quad \NPn^a = \iota^A \iota^{A'}, \quad \NPm^a = o^A \iota^{A'}, \quad \NPmbar^a = \iota^A o^{A'}
\end{equation} 
This gives a correspondence between a null frame in $\Mink$ and a dyad in $\Co^2$. 

The action of $\SL(2,\Co)$ on $\Co^2$ leaves invariant a complex area element, a skew-symmetric bispinor. A unique such spinor $\eps_{AB}$ is determined by the normalization 
$$
\met_{ab} = \eps_{AB} \bar\eps_{A'B'}.
$$
The inverse $\eps^{AB}$ of $\eps_{AB}$ is defined by $\eps_{AB}\eps^{CB} = \delta_A{}^C$, $\eps^{AB} \eps_{AC} = \delta_C{}^B$.
As with $\met_{ab}$ and its inverse $\met^{ab}$, the spin-metric $\eps_{AB}$ and its inverse $\eps^{AB}$ 
is used to lower and raise spinor indices, 
$$
\lambda_{B} = \lambda^{A}\eps_{AB} , \quad \lambda^{A} = \eps^{AB} \lambda_{B}.
$$
We have
$$
\eps_{AB} = o_A \iota_B - \iota_A o_B.
$$
In particular,
\begin{equation} \label{eq:dyad-normalization} 
o_A \iota^A = 1.
\end{equation}

An element $\phi_{A\cdots D A' \cdots D'}$ of $\bigotimes^k \Co^2 \bigotimes^l \bar \Co^2$ is called a spinor of valence $(k,l)$. The space of totally symmetric\footnote{The ordering between primed and unprimed indices is irrelevant.} spinors $\phi_{A \cdots D A' \cdots D'} = \phi_{(A \cdots D) (A' \cdots D')}$ is denoted $\SymSpin_{k,l}$. The spaces $\SymSpin_{k,l}$ for $k,l$ non-negative integers yield all irreducible representations of $\SL(2,\Co)$. In fact, one can decompose any spinor into ``irreducible pieces'', i.e. as a linear combination of totally symmetric spinors in $\SymSpin_{k,l}$ with factors of $\eps_{AB}$. The above mentioned correspondence between vectors and spinors extends to tensors of any type, and hence the just mentioned decomposition of spinors into irreducible pieces carries over to tensors as well. 
Examples are given by $\mathcal F_{ab} = \phi_{AB} \eps_{A'B'}$, a complex anti-self-dual 2-form, and ${}^-C_{abcd} = \Psi_{ABCD} \eps_{A'B'} \eps_{C'D'}$, a complex anti-self-dual tensor with the symmetries of the Weyl tensor. Here, $ \phi_{AB}$ and $\Psi_{ABCD}$ are symmetric.

\subsection{Spinors on spacetime} 
Let now $(\Mcal, \met_{ab})$ be a Lorentzian 3+1 dimensional spin manifold with metric of signature $+---$. The spacetimes we are interested in here are spin, in particular any orientable, globally hyperbolic 3+1 dimensional spacetime is spin, cf. \cite[page 346]{Ger70spinstructII}.
If $\Mcal$ is spin, then the orthonormal frame bundle $\SO(\Mcal)$ admits a lift to $\Spin(\Mcal)$, a principal $\SL(2,\Co)$-bundle. The associated bundle construction now gives vector bundles over $\Mcal$ corresponding to the  representations of $\SL(2,\Co)$, in particular we have bundles of valence $(k,l)$ spinors with sections $\phi_{A\cdots D A' \cdots D'}$.
The Levi-Civita connection lifts to act on sections of the spinor bundles, 
\begin{equation}\label{eq:nablavarphi}
\nabla_{AA'} : \varphi_{B \cdots D B' \cdots D' } \to \nabla_{AA'} \varphi_{B \cdots D B' \cdots D'} 
\end{equation} 
where we have used the tensor-spinor correspondence to replace the index $a$ by $AA'$. We shall denote the totally symmetric spinor bundles by $\SymSpin_{k,l}$ and their spaces of sections by $\SymSpinSec_{k,l}$. 

The above mentioned correspondence between spinors and tensors, and the decomposition into irreducible pieces, can be applied to the Riemann curvature tensor. In this case, the irreducible pieces correspond to the scalar curvature, traceless Ricci tensor, and the Weyl tensor, denoted by  
$R$, $S_{ab}$, and $C_{abcd}$, respectively. The Riemann tensor then takes the form  
\begin{align}
R_{abcd}={}&- \tfrac{1}{12} g_{ad} g_{bc} R
 + \tfrac{1}{12} g_{ac} g_{bd} R
 + \tfrac{1}{2} g_{bd} S_{ac}
 -  \tfrac{1}{2} g_{bc} S_{ad}
 -  \tfrac{1}{2} g_{ad} S_{bc}
 + \tfrac{1}{2} g_{ac} S_{bd}
 + C_{abcd}.
\end{align}
The spinor equivalents of these tensors are
\begin{subequations}
\begin{align}
C_{abcd}={}&\Psi_{ABCD} \bar\epsilon_{A'B'} \bar\epsilon_{C'D'}+\bar\Psi_{A'B'C'D'} \epsilon_{AB} \epsilon_{CD},\\
S_{ab} ={}& -2 \Phi_{ABA'B'},\\
R={}&24 \Lambda.
\end{align}
\end{subequations}

\subsection{Fundamental operators} 
Projecting \eqref{eq:nablavarphi} on its irreducible pieces gives the following four \emph{fundamental operators}, introduced in \cite{ABB:symop:2014CQGra..31m5015A}. 
\begin{definition}
The differential operators
$$
\sDiv_{k,l}:\mathcal{S}_{k,l}\rightarrow \mathcal{S}_{k-1,l-1}, \quad 
\sCurl_{k,l}:\mathcal{S}_{k,l}\rightarrow \mathcal{S}_{k+1,l-1}, \quad 
\sCurlDagger_{k,l}:\mathcal{S}_{k,l}\rightarrow \mathcal{S}_{k-1,l+1}, \quad 
\sTwist_{k,l}:\mathcal{S}_{k,l}\rightarrow \mathcal{S}_{k+1,l+1}
$$
are defined as
\begin{subequations}
\begin{align}
(\sDiv_{k,l}\varphi)_{A_1\dots A_{k-1}}{}^{A_1'\dots A_{l-1}'}\equiv{}&
\nabla^{BB'}\varphi_{A_1\dots A_{k-1}B}{}^{A_1'\dots A_{l-1}'}{}_{B'},\\
(\sCurl_{k,l}\varphi)_{A_1\dots A_{k+1}}{}^{A_1'\dots A_{l-1}'}\equiv{}&
\nabla_{(A_1}{}^{B'}\varphi_{A_2\dots A_{k+1})}{}^{A_1'\dots A_{l-1}'}{}_{B'},\\
(\sCurlDagger_{k,l}\varphi)_{A_1\dots A_{k-1}}{}^{A_1'\dots A_{l+1}'}\equiv{}&
\nabla^{B(A_1'}\varphi_{A_1\dots A_{k-1}B}{}^{A_2'\dots A_{l+1}')},\\
(\sTwist_{k,l}\varphi)_{A_1\dots A_{k+1}}{}^{A_1'\dots A_{l+1}'}\equiv{}&
\nabla_{(A_1}{}^{(A_1'}\varphi_{A_2\dots A_{k+1})}{}^{A_2'\dots A_{l+1}')}.
\end{align}
\end{subequations}
The operators are called respectively the divergence, curl, curl-dagger, and twistor operators. 
\end{definition}
With respect to complex conjugation, the operators $\sDiv, \sTwist$ satisfy $\overline{\sDiv_{k,l}} = \sDiv_{l,k}$, $\overline{\sTwist_{k,l}} = \sTwist_{l,k}$, while $\overline{\sCurl_{k,l}} = \sCurlDagger_{l,k}$, $\overline{\sCurlDagger_{k,l}} = \sCurl_{l,k}$. 

Denoting the adjoint of an operator by $\mathcal A$ with respect to the bilinear pairing 
$$
(\phi_{A_1 \cdots A_k A'_1 \cdots A'_l}, \psi_{A_1 \cdots A_k A'_1 \cdots A'_l})=\int \phi_{A_1 \cdots A_k A'_1 \cdots A'_l} \psi^{A_1 \cdots A_k A'_1 \cdots A'_l}d\mu 
$$
by $\mathcal A^\dagger$, 
and the adjoint with respect to the sesquilinear pairing  
$$
\langle \phi_{A_1 \cdots A_k A'_1 \cdots A'_l}, \psi_{A_1 \cdots A_l A'_1 \cdots A'_k}\rangle =\int \phi_{A_1 \cdots A_k A'_1 \cdots A'_l} \bar\psi^{A_1 \cdots A_k A'_1 \cdots A'_l}d\mu 
$$
by $\mathcal A^\star$ , 
we have 
\begin{align*}
(\sDiv_{k,l})^\dagger &= - \sTwist_{k-1,l-1}, & (\sTwist_{k,l})^\dagger &= - \sDiv_{k+1,l+1},&
(\sCurl_{k,l})^\dagger &= \sCurlDagger_{k+1,l-1}, & (\sCurlDagger_{k,l})^\dagger &= \sCurl_{k-1,l+1},
\end{align*}
and
\begin{align*}
(\sDiv_{k,l})^\star &= - \sTwist_{l-1,k-1} , & (\sTwist_{k,l})^\star &= - \sDiv_{l+1,k+1},&
(\sCurl_{k,l})^\star &= \sCurl_{l-1,k+1}, & (\sCurlDagger_{k,l})^\star &= \sCurlDagger_{l+1,k-1}.
\end{align*}

As we will see in section~\ref{sec:masslessspins}, the kernels of $\sCurlDagger_{2s,0}$ and $\sCurl_{0,2s}$ are the massless spin-s fields. The kernels of $\sTwist_{k,l}$, are the valence $(k,l)$ Killing spinors, which we will discuss further in section~\ref{sec:KillingSpinors} and section~\ref{sec:KillSpinSpacetime}.
A complete set of commutator properties of these operators can be found in \cite{ABB:symop:2014CQGra..31m5015A}.

\subsection{Massless spin-$s$ fields} \label{sec:masslessspins}
For $s \in \half \NatNum$,  $\varphi_{A\cdots D} \in \ker  \sCurlDagger_{2s,0}$ is a totally symmetric spinor $\varphi_{A\cdots D} = \varphi_{(A \cdots D)}$ of valence $(2s,0)$ which solves the massless spin-s equation 
$$
(\sCurlDagger_{2s,0} \varphi)_{A\cdots BD'} = 0.
$$
For $s=1/2$, this is the Dirac-Weyl equation $\nabla_{A'}{}^A \varphi_A = 0$, for $s=1$, we have the left and right Maxwell equation $\nabla_{A'}{}^B \phi_{AB} = 0$ and $\nabla_A{}^{B'} \varphi_{A'B'} = 0$, i.e. $(\sCurlDagger_{2,0} \phi)_{AA'} =0$, $(\sCurl_{0,2} \varphi)_{AA'} = 0$. 

An important example is the Coulomb Maxwell field on Kerr,   
\begin{equation}\label{eq:coulomb} 
\phi_{AB} = - \frac{2}{(r-ia\cos\theta)^2} o_{(A} \iota_{B)}
\end{equation} 
This is a non-trivial sourceless solution of the Maxwell equation on the Kerr background.
We note that the scalars components, see section \ref{sec:GHP} below, of the Coulomb field $\phi_1 = (r - i a \cos\theta)^{-2}$ while $\phi_0 = \phi_2 = 0$. 

For $s > 1$, the existence of a non-trivial solution to the spin-s equation implies  curvature conditions, a fact known as the Buchdahl constraint \cite{Buchdahl58},
\begin{equation}
0=\Psi_{(A}{}^{DEF}\phi_{B \dots C)DEF}.
\end{equation}
This is easily obtained by commuting the operators in
\begin{equation}
0 = (\sDiv_{2s-1,1} \sCurlDagger_{2s,0} \phi)_{A \dots C}. 
\end{equation}
For the case $s=2$, the equation 
$\nabla_{A'}{}^D \Psi_{ABCD} = 0$
is the Bianchi equation, which holds for the Weyl spinor in any vacuum spacetime. Due to the Buchdahl constraint, it holds that in any sufficiently general spacetime, a solution of the spin-2 equation is proportional to the Weyl spinor of the spacetime.  

\subsection{Killing spinors}\label{sec:KillingSpinors} 
Spinors $\varkappa_{A_1\cdots A_k}{}^{A_1' \cdots A_l'} \in \SymSpinSec_{k,l}$ satisfying 
$$
(\sTwist_{k,l} \varkappa)_{A_1 \cdots A_{k+1}}{}^{A_1' \cdots A_{l+1}'} = 0,
$$
are called Killing spinors of valence $(k,l)$. 
We denote the space of Killing spinors of valence $(k,l)$ by $\KillSpin_{k,l}$. 
The Killing spinor equation is an over-determined system. 
The space of Killing spinors is a finite dimensional space, and the existence of Killing spinors imposes strong restrictions on $\Mcal$, see section \ref{sec:KillSpinSpacetime} below. Killing spinors $\GenVec_{AA'} \in \KillSpin_{1,1}$ are simply conformal Killing vector fields, satisfying 
$\nabla_{(a} \GenVec_{b)} - \half \nabla^c \GenVec_c g_{ab}$. A Killing spinor $\kappa_{AB} \in \KillSpin_{2,0}$ corresponds to a complex anti-selfdual conformal Killing-Yano 2-form 
$
\mathcal{Y}_{ABA'B'} = \kappa_{AB} \eps_{A'B'} 
$
satisfying the equation 
\begin{equation}\label{eq:CKY}
\nabla_{(a} \mathcal{Y}_{b)c} -2 \zeta_c g_{ab} + \zeta_{(a} g_{b)c} = 0 ,
\end{equation} 
where in the 4-dimensional case, $\zeta_a = \tfrac{1}{3} \nabla_b \mathcal{Y}^b{}_a$. 

In the mathematics literature, Killing spinors of valence $(1,0)$ are known as twistor spinors. The terms conformal Killing-Yano form or twistor form is used also for the real 2-forms corresponding to Killing spinors of valence $(2,0)$, as well as for forms of higher degree and in higher dimension, in the kernel of an analogous Stein-Weiss operator.
Further, we mention that  Killing spinors $L_{ABA'B'} \in \KillSpin_{2,2}$ are traceless symmetric conformal Killing tensors $L_{ab}$, satisfying the equation 
\begin{equation}\label{eq:confkill}
\nabla_{(a}L_{bc)} -  \tfrac{1}{3} g_{(ab}\nabla^{d}L_{c)d}  = 0 .
\end{equation}
In particular, any tensor of the form $\zeta g_{ab}$ for some scalar field $\zeta$ is a conformal Killing tensor. If $\gamma^a$ is a null geodesic and $L_{ab}$ is a conformal Killing tensor, then $L_{ab} \dot \gamma^a \dot \gamma^b$ is conserved along $\gamma^a$. 
For any $\kappa_{AB} \in \KillSpin_{2,0}$ we have that 
$L_{ABA'B'} = \kappa_{AB} \bar \kappa_{A'B'} \in \KillSpin_{2,2}$.  
See section \ref{sec:KillSpinSpacetime} below for further details.

\subsection{Algebraically special spacetimes} \label{sec:algspec}
Let $\varphi_{A\cdots D} \in \SymSpinSec_{k,0}$. A spinor $\alpha_A$ is a \emph{principal spinor} of $\varphi_{A\cdots D}$ if  
$$
\varphi_{A \cdots D} \alpha^A \cdots \alpha^D= 0.
$$
An application of the fundamental theorem of algebra shows that any $\varphi_{A \cdots D} \in \SymSpinSec_{k,0}$ has exactly $k$ principal spinors $\alpha_A, \dots, \delta_A$, and hence is of the form 
$$
\varphi_{A \cdots D} = \alpha_{(A} \cdots \delta_{D)}.
$$
If $\varphi_{A \cdots D} \in \SymSpinSec_{k,0}$ has $n$ distinct principal spinors $\alpha^{(i)}_A$, repeated $m_i$ times, then $\varphi_{A \cdots D} $ is said to have algebraic type $\{m_1, \dots, m_n\}$. 
Applying this to the Weyl tensor leads to the Petrov classification, see table 1. We have the following list of algebraic, or Petrov, types\footnote{The Petrov classification is exclusive, so a spacetime belongs at each point to exactly one Petrov class.}. 

\begin{table}[!h]
\centering
\begin{minipage}{0.5\textwidth}
\centering
\vskip .1in
\begin{tabular}{l|l|l} 
I & $\{1,1,1,1\}$ & $\Psi_{ABCD} = \alpha_{(A} \beta_B \gamma_C \delta_{D)}$ \\ 
II & $\{2,1,1\}$ & $\Psi_{ABCD} = \alpha_{(A} \alpha_B \gamma_C \delta_{D)}$ \\ 
D & $\{2,2\}$ & $\Psi_{ABCD} = \alpha_{(A} \alpha_B \beta_C \beta_{D)}$ \\ 
III & $\{3,1\}$ & $\Psi_{ABCD} = \alpha_{(A} \alpha_B \alpha_C \beta_{D)}$ \\ 
N & $\{4\}$ & $\Psi_{ABCD} = \alpha_A \alpha_B \alpha_C \alpha_D$ \\ 
O & $\{ - \}$ & $\Psi_{ABCD} = 0$
\end{tabular} 
\bigskip
\caption{The Petrov classification}
\end{minipage}%
\begin{minipage}{0.5\textwidth}
\centering
\vskip -.1in

\centering

\raisebox{-0.5\height}{\includegraphics{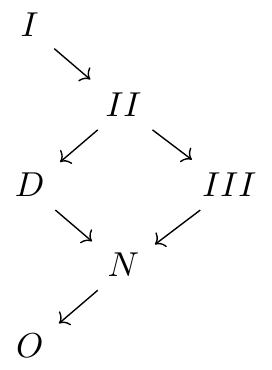}}

\end{minipage}
\end{table} 

A principal spinor $o_A$ determines a principal null direction $l_a = o_A \bar o_{A'}$. 
The Goldberg-Sachs theorem states that in a vacuum spacetime, the congruence generated by a null field $l^a$ is geodetic and shear free\footnote{If $l^a$ is geodetic and shear then the spin coefficients $\sigma,\kappa$, cf. \eqref{eq:spincoeff-def} below, satisfy $\sigma=\kappa=0$.} if and only if $l_a$ is a repeated principal null direction of the Weyl tensor $C_{abcd}$ (or equivalently $o_A$ is a repeated principal spinor of the Weyl spinor $\Psi_{ABCD}$). 

\subsubsection{Petrov type $\PetrovD$} The Kerr metric is of Petrov type D, and many of its important properties follows from this fact. 
The vacuum type $\PetrovD$ spacetimes have been classified by Kinnersley \cite{kinnersley:1969JMP....10.1195K},  see also Edgar et al \cite{edgar:etal:2009CQGra..26j5022E}. The family of Petrov type D spacetimes includes the Kerr-NUT family and the boost-rotation symmetric C-metrics. The only Petrov type D vacuum spacetime which is asymptotically flat and has positive mass is the Kerr metric, see theorem \ref{thm:kerrchar} below. 

A Petrov type $\PetrovD$ spacetime has two repeated principal spinors $o_A, \iota_A$, and correspondingly there are two repeated principal null directions $\NPl^a, \NPn^a$, for the Weyl tensor. We can without loss of generality assume that $\NPl^a \NPn_a = 1$, and define a null tetrad by adding  complex null vectors $\NPm^a, \NPmbar^a$ normalized such that $\NPm^a \NPmbar_a = -1$.
By the Goldberg-Sachs theorem both $\NPl^a, \NPn^a$ are geodetic and shear free, and only one of the 5 independent complex Weyl scalars is non-zero, namely 
\begin{align}
\Psi_2={}&- l^{a} m^{b} \bar{m}^{d} n^{c} C_{abcd}
\end{align}
In this case, the Weyl spinor takes the form  
$$
\Psi_{ABCD} = \frac{1}{6} \Psi_2 o_{(A} o_B \iota_C \iota_{D)}.
$$
See \eqref{eq:Psi2-Kerr} below for the explicit form of $\Psi_2$ in the Kerr spacetime. 

The following result is a consequence of the Bianchi identity. 
\begin{theorem}[\cite{walker:penrose:1970CMaPh..18..265W}]
Assume $(\Mcal, g_{ab})$ is a vacuum spacetime of Petrov type D. Then $(\Mcal, g_{ab})$ admits a one-dimensional space of Killing spinors $\kappa_{AB}$ of the form 
\begin{equation}\label{eq:kappa11} 
\kappa_{AB} = -2 \kappa_1 o_{(A} \iota_{B)}
\end{equation}
where $o_A, \iota_A$ are the principal spinors of $\Psi_{ABCD}$ and $\kappa_1 \propto \Psi_2^{-1/3}$. 
\end{theorem} 
\begin{remark}
Since the Petrov classes are exclusive, we have that $\Psi_2 \ne 0$ for a Petrov type D space. 
\end{remark} 

\subsection{Spacetimes admitting a Killing spinor} \label{sec:KillSpinSpacetime} 
Differentiating the Killing spinor equation $(\sTwist_{k,l} \phi)_{A \cdots D A' \cdots D'} = 0$,  and commuting derivatives yields
an algebraic relation between the curvature, Killing spinor, and their covariant derivatives which restrict the curvature spinor, see \cite[\S 2.3]{ABB:symop:2014CQGra..31m5015A}, see also \cite[\S 3.2]{2015arXiv150402069A}. 
In particular, for a Killing spinor $\kappa_{A\cdots D}$ of valence $(k,0)$, $k \geq 1$, the condition 
\begin{equation}\label{eq:killspin-cond-general}
\Psi_{(ABC}{}^F \kappa_{D\cdots E ) F} = 0
\end{equation}
must hold, which restricts the algebraic type of  the Weyl spinor. 
For a valence $(2,0)$ Killing spinor $\kappa_{AB}$, the condition takes the form
\begin{align} 
\label{eq:curv-restrict-valence-2}
\Psi_{(ABC}{}^E \kappa_{D)E} ={}& 0
\end{align}
	It follows from \eqref{eq:curv-restrict-valence-2} that a spacetime admitting a valence $(2,0)$ Killing spinor is of type $\PetrovD, \PetrovN$, or $\PetrovO$. The space of Killing spinors of valence $(2,0)$ on Minkowski space (or any space of Petrov type $\PetrovO$) has complex dimension 10. The explicit form  in Cartesian coordinates $x^{AA'}$ is
$$
\kappa^{AB} = U^{AB} + 2x^{A'(A} V^{B)}{}_{A'} + x^{AA'} x^{BB'} W_{A'B'},
$$
where $U^{AB}, V^B{}_{A'}, W^{A'B'}$ are arbitrary constant symmetric spinors, see\cite[Eq. (4.5)]{aksteiner:thesis}. One of these corresponds to the spinor in \eqref{eq:kappa11}, in spheroidal coordinates it takes the form given in \eqref{eq:kappaKerrNewman} below. 

A further application of the commutation properties of the fundamental operators yields that the 1-form  
\begin{equation}\label{eq:xikappadef}
\xi_{AA'} = (\sCurlDagger_{2,0} \kappa)_{AA'},
\end{equation} 
is a Killing field, $\nabla_{(a} \xi_{b)} = 0$, provided $\Mcal$ is vacuum. 
Clearly the real and imaginary parts of $\xi_a$ are also Killing fields. 
If $\xi_a$ is proportional to a real Killing field\footnote{We say that such spacetimes are of the generalized Kerr-NUT class, see \cite{backdahl:valiente-kroon:2010PhRvL.104w1102B} and references therein.}, we can without loss of generality assume that $\xi_a$ is real.
In this case, the 2-form  
\begin{equation}\label{eq:Yab-form}
Y_{ab} = \tfrac{3}{2}  i (\kappa_{AB}\bar\epsilon_{A'B'}  -  \bar{\kappa}_{A'B'}\epsilon_{AB})
\end{equation} 
is a Killing-Yano tensor, $\nabla_{(a} Y_{b)c} = 0$, and the symmetric 2-tensor 
\begin{equation}\label{eq:K=YY} 
K_{ab} =  Y_a{}^c Y_{cb}
\end{equation}  
is a Killing tensor, 
\begin{equation}\label{eq:Ktenseq}
\nabla_{(a} K_{bc)} = 0. 
\end{equation}
Further, in this case,
\begin{equation}\label{eq:zeta=xiK}
\zeta_a = \xi^b K_{ab} 
\end{equation}
is a Killing field, see \cite{1973CMaPh..33..129H,1977CMaPh..56..277C}. 
Recall that the quantity $L_{ab} \dot \gamma^a \dot \gamma^b$ is conserved along null geodesics if $L_{ab}$ is a conformal Killing tensor. For Killing tensors, this fact extends to all geodesics, so that if $K_{ab}$ is a Killing tensor, then $K_{ab} \dot \gamma^a \dot \gamma^b$ is conserved along a geodesic $\gamma^a$. See \cite{2015arXiv150402069A} for further details and references.

\subsection{GHP formalism} \label{sec:GHP}
Taking the point of view that the null tetrad components of tensors are sections of complex line bundles with action of the non-vanishing complex scalars corresponding to the rescalings of the tetrad, respecting the normalization, leads to the GHP formalism \cite{GHP}. 

Given  
a null tetrad $\NPl^a, \NPn^a, \NPm^a, \NPmbar^a$ we have a spin dyad $o_A, \iota_A$ as discussed above. For a spinor $\varphi_{A\cdots D} \in \SymSpinSec_{k,0}$, it is convenient to introduce the Newman-Penrose  scalars 
\begin{equation}\label{eq:varphi_i-def}
\varphi_i = \varphi_{A_1 \cdots A_i A_{i+1} \cdots A_k} \iota^{A_1} \cdots \iota^{A_i} o^{A_{i+1}} \cdots o^{A_k}.
\end{equation} 
In particular, $\Psi_{ABCD}$ corresponds to the five complex Weyl scalars 
$\Psi_i, i = 0, \dots 4$.  The definition $\varphi_i$ extends in a natural way to the scalar components of spinors of valence $(k,l)$.

The normalization \eqref{eq:dyad-normalization} is left invariant under rescalings $o_A \to \lambda o_A$, $\iota_A \to \lambda^{-1} \iota_A$ where $\lambda$ is a non-vanishing complex scalar field on $\Mcal$. Under such rescalings, the scalars defined by projecting on the dyad, such as $\varphi_i$ given by \eqref{eq:varphi_i-def} transform as sections of complex line bundles. A scalar $\varphi$ is said to have type $\{p,q\}$ if $\varphi \to \lambda^p \bar\lambda^q \varphi$ under such a rescaling. Such fields are called properly weighted. The lift of the Levi-Civita connection $\nabla_{AA'}$ to these bundles gives a covariant derivative denoted $\Theta_a$. Projecting on the null tetrad $\NPl^a, \NPn^a, \NPm^a, \NPmbar^a$ gives the GHP operators 
$$
\tho = \NPl^a \Theta_a, \quad \tho' = \NPn^a\Theta_a, \quad \edt = \NPm^a \Theta_a , \quad 
\edt' = \NPmbar^a \Theta_a.
$$
The GHP operators are properly weighted, in the sense that they take properly weighted fields to properly weighted fields, for example if $\varphi$ has type $\{p,q\}$, then $\tho \varphi$ has type $\{p+1, q+1\}$. This can be seen from the fact that $\NPl^a = o^A \bar{o}^{A'}$ has type $\{1,1\}$. 
There are 12 connection coefficients in a null frame, up to complex conjugation. Of these, 8 are properly weighted, the GHP spin coefficients. The other connection coefficients enter in the connection 1-form for the connection $\Theta_a$. 

The following formal operations take weighted quantities to weighted
quantities,
\begin{equation}\label{eq:ghpsym}
\begin{aligned}
^-(\text{bar})&: \; \NPl^a \to \NPl^a, \;   \NPn^a \to \NPn^a, \;  \NPm^a
\to \NPmbar^a,	\; \NPmbar^a \to \NPm^a, &\{p,q\}\to\{q,p\} ,\\
'(\text{prime})&: \; \NPl^a \to \NPn^a, \;   \NPn^a \to \NPl^a, \;
\NPm^a \to \NPmbar^a,  \; \NPmbar^a \to \NPm^a, &\{p,q\}\to\{-p,-q\} ,\\
^*(\text{star})&: \; \NPl^a \to \NPm^a , \;  \NPn^a \to -\NPmbar^a, \;
\NPm^a \to -\NPl^a, \; \NPmbar^a \to \NPn^a, &\{p,q\}\to\{p,-q\} .
\end{aligned}
\end{equation}
The properly weighted spin coefficients can be represented as 
\begin{align}
\kappa &=  \NPm^b \NPl^a \nabla_a \NPl_b , \quad
\sigma =  \NPm^b \NPm^a \nabla_a \NPl_b, \quad
\rho   =  \NPm^b \NPmbar^a \nabla_a \NPl_b , \quad
\tau   = \NPm^b \NPn^a \nabla_a \NPl_b ,
\label{eq:spincoeff-def}
\end{align}
together with their primes $\kappa', \sigma', \rho', \tau'$.

A systematic application of the above formalism allows one to write the tetrad projection of the geometric field equations in a compact form. For example, the Maxwell equation corresponds to the four scalar equations given by 
\begin{align}
(\tho -2\rho)\phi_1 -({\edt}'-\tau')\phi_0 = -\kappa \phi_2 , \label{eq:ghpmaxwell}
\end{align}
with its primed and starred versions.

Working in a spacetime of Petrov type $\PetrovD$ gives drastic simplifications, in view of the fact that 
choosing the null tedrad so that $\NPl^a$, $\NPn^a$ are aligned with principal null directions of the Weyl tensor (or equivalently choosing the spin dyad so that $o_A, \iota_A$ are principal spinors of the Weyl spinor), as has already been mentioned, the Weyl scalars are zero with the exception of $\Psi_2$, and the only non-zero spin coefficients are $\rho, \tau$ and their primed versions.

\section{The Kerr spacetime} \label{sec:kerrspacetime} 
Taking into account the background material given in section \ref{sec:prel}, we can now state some further properties of the Kerr spacetime. 
As mentioned above, the Kerr metric is algebraically special, of Petrov type $\PetrovD$. 
An explicit principal null tetrad $(\NPl^a, \NPn^a, \NPm^a, \NPmbar^a)$ is given by the Carter tetrad \cite{znajek:1977MNRAS.179..457Z}  
\begin{subequations}\label{eq:KerrTetrad}
\begin{align}
\NPl^{a}={}&\frac{a (\partial_\phi)^{a}}{\sqrt{2} \KDelta^{1/2} \KSigma^{1/2}}
 + \frac{(a^2 + r^2)(\partial_t)^{a} }{\sqrt{2} \KDelta^{1/2} \KSigma^{1/2}}
 + \frac{\KDelta^{1/2}(\partial_r)^{a} }{\sqrt{2} \KSigma^{1/2}},\\
\NPn^{a}={}&\frac{a (\partial_\phi)^{a}}{\sqrt{2} \KDelta^{1/2} \KSigma^{1/2}}
 + \frac{(a^2 + r^2)(\partial_t)^{a} }{\sqrt{2} \KDelta^{1/2} \KSigma^{1/2}}
 -  \frac{\KDelta^{1/2}(\partial_r)^{a} }{\sqrt{2} \KSigma^{1/2}},\\
\NPm^{a}={}&\frac{(\partial_\theta)^{a}}{\sqrt{2} \KSigma^{1/2}}
 + \frac{i \csc\theta(\partial_\phi)^{a} }{\sqrt{2} \KSigma^{1/2}}
 + \frac{i a \sin\theta (\partial_t)^{a} }{\sqrt{2} \KSigma^{1/2}}.
\end{align}
\end{subequations}
In view of the normalization of the tetrad, the metric takes the form $\met_{ab} = 2 (\NPl_{(a} \NPn_{b)} - \NPm_{(a} \NPmbar_{b)})$. We remark that the choice of $\NPl^a$, $\NPn^a$ to be aligned with the principal null directions of the Weyl tensor, together with the normalization of the tetrad fixes the tetrad up to rescalings. 

We have 
\begin{align}
\Psi_2={}& - \frac{M}{(r -i a \cos\theta)^3}. \label{eq:Psi2-Kerr} \\ 
\kappa_{AB} ={}& \tfrac{2}{3} (r - i a \cos\theta) o_{(A}\iota_{B)},\label{eq:kappaKerrNewman} 
\end{align}
With $\kappa_{AB}$ as in \eqref{eq:kappaKerrNewman}, equation \eqref{eq:xikappadef} yields
\begin{align}\label{eq:xiKerr}
\xi^a={}&(\partial_t)^a,
\end{align}
and from \eqref{eq:Yab-form} we get 
\begin{align}
Y_{ab} ={}&  a \cos\theta \NPl_{[a} \NPn_{b]} -i r \NPm_{[a}  \NPmbar_{b]}  
\end{align}
With the normalizations above, the Killing tensor \eqref{eq:K=YY} takes the form 
\begin{align}
K_{ab} ={}&  \tfrac{1}{4} ( 2\KSigma l_{(a} n_{b)} - r^2 g_{ab} )    \label{eq:KDefinition}
\end{align}
and \eqref{eq:zeta=xiK} gives 
\begin{align}
\zeta^a={}&a^2 (\partial_t)^a
 + a (\partial_\phi)^a  . \label{eq:zetaKerr}
\end{align}
Recall that for a geodesic $\gamma$, the quantity $\GeodesicKCarter = 4 \TensorKCarter_{ab} \dot \gamma^a \dot \gamma^b$, known as Carter's constant, is conserved. Explicitely, 
\begin{align}
\GeodesicKCarter ={}& \dot \gamma_\theta^2 + a^2 \sin^2\theta \GeodesicEnergy^2 + 2a \GeodesicEnergy \GeodesicLz + a^2 \cos^2\theta \GeodesicMass^2 
\end{align}
where $\dot \gamma_\theta = \dot \gamma^a (\partial_\theta)_a$. For $a \ne 0$, the tensor $K_{ab}$ cannot be expressed as a tensor product of Killing fields \cite{walker:penrose:1970CMaPh..18..265W}, and similarly Carter's constant $\GeodesicKCarter$ cannot be expressed in terms of the constants of motion associated to Killing fields.  
In this sense $\TensorKCarter_{ab}$ and $\GeodesicKCarter$ manifest a \emph{hidden symmetry} of the Kerr spacetime. As we shall see in section \ref{sec:symop}, these structures are also related to symmetry operators and separability properties, as well as conservation laws, for field equations on Kerr, and more generally in spacetimes admitting Killing spinors satisfying certain auxiliary conditions.

\newcommand{\geodesic}{\gamma}
\newcommand{\geodesictangent}{\dot\gamma}
\newcommand{\CurlyR}{\mathcal{R}}
\newcommand{\geodq}{\mathbf{q}}

\newcommand{\Sphere}{\mathbb{S}}

\subsection{Characterizations of Kerr} \label{sec:characterization} 

Consider a vacuum Cauchy data set $(\SpaceSlice, \Spaceh_{ij}, \SpaceK_{ij})$. We say that $(\SpaceSlice, \Spaceh_{ij}, \SpaceK_{ij})$ is asymptotically flat if $\SpaceSlice$ has an end $\Re^3 \setminus B(0, R)$ with a coordinate system $(x^i)$ such that 
\begin{equation}\label{eq:asymptflat} 
\Spaceh_{ij} = \delta_{ij} + O_\infty(r^{\alpha}) , \quad \SpaceK_{ij} = O_\infty(r^{\alpha -1})
\end{equation} 
for some $\alpha < - 1/2$.  The Cauchy data set $(\SpaceSlice, \Spaceh_{ij}, \SpaceK_{ij})$ is asymptotically Schwarzschildean if 
\begin{subequations}\label{eq:boosteddecay}
\begin{eqnarray}
&& h_{ij} = -\left(1+\frac{2A}{r}\right)\delta_{ij} - \frac{\alpha}{r}\left( \frac{2x_ix_j}{r^2}-\delta_{ij}\right)+o_\infty(r^{-3/2}), \label{BoostedDecay1} \\
&& k_{ij} = \frac{\beta}{r^2}\left( \frac{2x_ix_j}{r^2}-\delta_{ij} \right) + o_\infty(r^{-5/2}),  \label{BoostedDecay2}
\end{eqnarray}
\end{subequations}
where $A$ is a constant, and $\alpha,\beta$ are functions on $S^2$, see \cite[\S 6.5]{backdahl:valiente-kroon:2010:MR2753388} for details. 
Here, the symbols 
$\ordo_\infty(r^{\alpha})$ are defined in terms of weighted Sobolev spaces, see \cite[\S 6.2]{backdahl:valiente-kroon:2010:MR2753388} for details.

If $(\Mcal, \met_{ab})$ is vacuum and contains a Cauchy surface $(\SpaceSlice, \Spaceh_{ij}, \SpaceK_{ij})$ satisfying \eqref{eq:asymptflat} or \eqref{eq:boosteddecay}, then $(\Mcal, \met_{ab})$ is asymptotically flat, respectively asymptotically Schwarzschildean, at spatial infinity.
In this case there is a spacetime coordinate system $(x^\alpha)$ such that $\met_{\alpha\beta}$ is asymptotic to the Minkowski line element with asymptotic conditions compatible with \eqref{eq:boosteddecay}. %
For such spacetimes, the ADM 4-momentum $P^\mu$ is well defined. The positive mass theorem states that $P^\mu$ is future directed causal $P^\mu P_\mu \geq 0$ (where the contraction is in the asymptotic Minkowski line element), $P^0 \geq 0$, and gives conditions under which $P^\mu$ is strictly timelike. This holds in particular if $\SpaceSlice$ contains an apparent horizon. 

Mars \cite{mars:2000CQGra..17.3353M}  
has given a characterization of the Kerr spacetime as an asymptotically flat vacuum spacetime with a Killing field $\xi^a$ asymptotic to a time translation, positive mass, and an additional condition on the Killing form $F_{AB} = (\sCurl_{1,1} \xi)_{AB}$, 
$$
\Psi_{ABCD} F^{CD} \propto F_{AB}
$$
A characterization in terms of algebraic invariants of the Weyl tensor has been given by Ferrando and Saez \cite{ferrando:saez:2009CQGra..26g5013F}. 
The just mentioned characterizations are in terms of spacetime quantities. As mentioned in section \ref{sec:KID} Killing spinor initial data propagates,  which
can be used to formulate a characterization of Kerr in terms of Cauchy data, see \cite{backdahl:valiente-kroon:2010PhRvL.104w1102B, backdahl:valiente-kroon:2010:MR2753388, backdahl:valente-kroon:2011RSPSA.467.1701B, backdahl:valiente-kroon:2012JMP....53d2503B}.

We here give a characterization in terms spacetimes admitting a Killing spinor of valence $(2,0)$. 
\begin{theorem} \label{thm:kerrchar} 
Assume that $(\Mcal, \met_{ab})$ is vacuum, asymptotically Schwarzschildean at spacelike infinity, and contains a Cauchy slice bounded by an apparent horizon. Assume further $(\Mcal, \met_{ab})$ admits a non-vanishing Killing spinor $\kappa_{AB}$ of valence $(2,0)$. Then $(\Mcal, \met_{ab})$ is locally isometric to the Kerr spacetime. 
\end{theorem} 
\begin{proof} 
Let $P^\mu$ be the ADM 4-momentum vector for $\Mcal$. 
By the positive mass theorem, $P^\mu P_\mu \geq 0$. In the case where $\Mcal$ contains a Cauchy surface bounded by an apparent horizon, then $P^\mu P_\mu > 0$ by  
\cite[Remark 11.5]{chrusciel:bartnik:2003math......7278C}\footnote{Section 11 appears only in the ArXiv version of \cite{chrusciel:bartnik:2003math......7278C}.}. 

Recall that a spacetime with a Killing spinor of valence $(2,0)$ is of Petrov type $\PetrovD, \PetrovN$, or $\PetrovO$. 
From asymptotic flatness and the positive mass theorem, we have $C_{abcd} C^{abcd} = O(1/r^6)$, and hence there is a neighbourhood of spatial infinity where $\Mcal$ is Petrov type $\PetrovD$. It follows that near spatial infinity,  
$\kappa_{AB} = -2 \kappa_1 o_{(A} \iota_{B)}$, with $\kappa_1 \propto \Psi_2^{-1/3} = O(r)$. 
It follows from our asymptotic conditions that the Killing field 
$\xi_{AA'} = (\sCurlDagger_{2,0}\kappa)_{AB}$ is $O(1)$ and hence asymptotic to a translation, $\xi^\mu \to A^\mu$ as $r \to \infty$, for some constant vector $A^\mu$. It follows from the discussion in \cite[\S 4]{aksteiner:andersson:2013CQGra..30o5016A} that $A^\mu$ is non-vanishing. 
Now, by \cite[\S III]{beig:chrusciel:1996JMP....37.1939B}, it follows that in the case $P^\mu P_\mu > 0$, then $A^\mu$ is proportional to $P^\mu$, see also \cite{beig:omurchadha:1987AnPhy.174..463B}.
We are now in the situation considered in the work by 
B\"ackdahl and Valiente-Kroon, see \cite[Theorem B.3]{backdahl:valente-kroon:2011RSPSA.467.1701B}, and hence we can conclude that $(\Mcal, \met_{ab})$ is locally isometric to the Kerr spacetime. 
\end{proof}  

\begin{remark}
\begin{enumerate} 
\item 
This result can be turned into a characterization in terms of Cauchy data along the lines in \cite{backdahl:valiente-kroon:2010:MR2753388}.
\item Theorem \ref{thm:kerrchar} can be viewed as a variation on the Kerr characterization given in 
\cite[Theorem B.3]{backdahl:valente-kroon:2011RSPSA.467.1701B}. In the version given here, the asymptotic conditions on the Killing spinor have been removed.  
\end{enumerate} 
\end{remark}

\section{Monotonicity and dispersion} \label{sec:monotonicity}
The dispersive properties of fields, i.e. the tendency of the energy density contained within any stationary region to decrease asymptotically to the future is a crucial property for solutions of field equations on spacetimes, and any proof of stability must exploit this phenomenon. In view of the geometric optics approximation, the dispersive property of fields can be seen in an analogous dispersive property of null geodesics, i.e. the fact that null geodesics in the Kerr spacetime which do not orbit the black hole at a fixed radius must leave any stationary region in at least one of the past or future directions. In section \ref{sec:geodesics} we give an explanation for this fact using tools which can readily be adapted to the case of field equations, while in section \ref{sec:fields} we outline sketch how these ideas apply to fields.

We begin by a discussion of conservation laws. For a null geodesic $\Geodesic^a$, we define the energy associated with  a vector field $\vecX$ and evaluated on a Cauchy hypersurface $\HypersurfaceGeneral$ to be
\begin{align*}
\GenEnergyGeodesic{\vecX}[\Geodesic](\HypersurfaceGeneral)
&= \gMetric_{ab}\vecX^a \dot\Geodesic^b |_{\HypersurfaceGeneral}. 
\end{align*}
Since $\dot\Geodesic^b \nabla_b \dot\Geodesic^a=0$ for a geodesic, integrating the derivative of the energy gives  
\begin{align}
\GenEnergyGeodesic{\vecX}[\Geodesic](\HypersurfaceGeneral_2)-\GenEnergyGeodesic{\vecX}[\Geodesic](\HypersurfaceGeneral_1)
=& \int_{\lambda_1}^{\lambda_2} (\dot\Geodesic_a \dot\Geodesic_b) \nabla^{(a}\vecX^{b)} \di\lambda ,
\label{eq:deformForGeodesics}
\end{align}
where $\lambda_i$ is the unique value of $\lambda$ such that $\Geodesic(\lambda)$ is the intersection of $\Geodesic$ with $\HypersurfaceGeneral_i$. 
Formula \eqref{eq:deformForGeodesics} is particularly easy to work with, if one recalls that 
\begin{align*}
\nabla^{(a}\vecX^{b)}&= -\frac12 \Lie_{\vecX}\gMetric^{ab} .
\end{align*}
The tensor $\nabla^{(a}\vecX^{b)}$ is commonly called the ``deformation tensor''. In the following, unless there is room for confusion, we will drop reference to $\Geodesic$ and $\HypersurfaceGeneral$ in referring to $\GenEnergyGeodesic{\vecX}$.

Conserved quantities play a crucial role in understanding the
behaviour of geodesics as well as fields. By \eqref{eq:deformForGeodesics}, the energy $\GenEnergyGeodesic{\vecX}$ is conserved if $\vecX^a$ is a Killing field. In the Kerr spacetime we have the Killing fields $\xi^a = (\partial_t)^a$, $\eta^a = (\partial_\phi)^a$ with the corresponding conserved quantities energy $\GeodesicEnergy=(\partial_t)^a \geodesictangent_a$ and azimuthal angular momentum  
$\GeodesicLz =(\partial_\phi)^a \geodesictangent_a$. In addition, the  squared particle mass $\GeodesicMass=\met_{ab}\geodesictangent^a\geodesictangent^b$, and the Carter constant 
$\GeodesicKCarter=\TensorKCarter_{ab} \geodesictangent^a\geodesictangent^b$ are conserved along any geodesic $\geodesic^a$ in the Kerr spacetime. The presence of the extra conserved
quantity allows one to integrate
the equations of geodesic motion\footnote{In general, the geodesic equation in a 4-dimensional stationary and axi-symmetric spacetime cannot be integrated, and the dynamics of particles may in fact be chaotic, see \cite{2008PhRvD..77b4035G, 2010PhRvD..81l4005L} and references therein. Note however that the geodesic equations are not \emph{separable} in the Boyer-Lindquist coordinates. On the other hand, the Darboux coordinates have this property, cf. \cite{frolov2011introduction}.}.

For a covariant field equation derived from an action principle which depends on the background geometry only via the metric and its derivatives, 
the symmetric stress-energy tensor $T_{ab}$ is conserved. As an example, we consider the wave equation 
\begin{equation}\label{eq:wave}
\nabla^a \nabla_a \psi = 0
\end{equation} 
which has stress-energy tensor 
\begin{equation}\label{eq:Tab-wave}
T_{ab} = \nabla_{(a} \psi \nabla_{b)} \bar \psi - \half \nabla^c \psi \nabla_c \bar \psi g_{ab}
\end{equation} 
Let $\psi$ be a solution to \eqref{eq:wave}. Then $T_{ab}$ is conserved, $\nabla^a T_{ab} = 0$. For a vector field $\vecX^a$ we have that $\nabla^a (T_{ab} \vecX^b)$ is given in terms of the deformation tensor,
$$
\nabla^a (T_{ab} \vecX^b) = T_{ab} \nabla^{(a} \vecX^{b)}
$$
Let $(J_{\vecX})_a = T_{ab} \vecX^b$ be the current corresponding to $\vecX^a$. 
By the above, we have conserved currents $J_\xi$ and $J_\eta$ corresponding to the Killing fields $\xi^a$, $\eta^a$. 

An application of Gauss' law gives the analog of \eqref{eq:deformForGeodesics}, 
$$
\int_{\Sigma_2} (J_{\vecX})_a d\sigma^a - \int_{\Sigma_1} (J_{\vecX})_a d\sigma^a = \int_{\Omega} T_{ab} \nabla^{(a} \vecX^{b)}
$$
where $\Omega$ is a spacetime region bounded by $\Sigma_1$, $\Sigma_2$.

\subsection{Monotonicity for null geodesics} \label{sec:geodesics} 
We shall consider only null geodesics, i.e.  $\GeodesicMass=0$. In this case we have 
\begin{align} 
\GeodesicKCarter ={}& \TensorKCarter_{ab} \dot \gamma^a \dot \gamma^b \nonumber \\
={}& 2 \Sigma \NPl_{(a} \NPn_{b)} \dot \gamma^a \dot \gamma^b  \nonumber \\ 
={}& 2 \Sigma \NPm_{(a} \NPmbar_{b)} \dot \gamma^a \dot \gamma^b \label{eq:mmbar} 
\end{align} 
We note that the tensors $2\Sigma \NPl_{(a} \NPn_{b)}$ and $2\Sigma \NPm_{(a} \NPmbar_{b)}$ are conformal Killing tensors, see section \ref{sec:KillingSpinors}. From \eqref{eq:mmbar} it is clear that $\GeodesicKCarter$ is non-negative. 
A calculation using \eqref{eq:KerrTetrad} gives 
\begin{align*}
2\Sigma \NPl^{(a} \NPn^{b)} \partial_a \partial_b ={}& \frac{1}{\Delta} [ (r^2 + a^2) \partial_t + a \partial_\phi]^2 - \Delta \partial_r^2 \\ 
2 \Sigma \NPm^{(a} \NPmbar^{b)} \partial_a \partial_b ={}& 
\partial_\theta^2 + \frac{1}{\sin^2\theta} \partial_\phi^2 + a^2 \sin^2 \theta \partial_t^2 + 2 a \partial_t \partial_\phi
\end{align*}  
Let $Z = (r^2 + a^2) \GeodesicEnergy + a \GeodesicLz$. Recall that $\dot r = \dot \gamma^r = g^{rr} \dot \gamma_r$ where $g^{rr} = - \KDelta/\KSigma$. Now we can write $0 = g_{ab} \dot \gamma^a \dot\gamma^b$ in the form 
\begin{align} 
\Sigma^2 \dot r^2 + \CurlyR(r;\GeodesicEnergy,\GeodesicLz,\GeodesicKCarter)={}& 0 \label{eq:Kerrnullgeod} \\
\intertext{where} 
\CurlyR =- Z^2 + \Delta \GeodesicKCarter \label{eq:CurlyR=Z2+DelK}
\end{align}
Equation \eqref{eq:Kerrnullgeod} is the exact analog of \eqref{eq:Schw:nullgeodr} for the Schwarzschild case. It is clear from \eqref{eq:mmbar} that for null geodesics, $\GeodesicKCarter$ corresponds, in the Schwarzschild case with $a=0$, to $\GeodesicLsquared$, the squared total angular momentum. It is possible to derive equations similar to \eqref{eq:Kerrnullgeod} for the other coordinates $t,\theta,\phi$, which allows the solution of the geodesic equations by quadratures, see eg. \cite{Teo} for details.

Equation \eqref{eq:Kerrnullgeod} allows one to make a qualitative analysis of the motion of null geodesics in the Kerr spacetime. In particular, we find that the location of orbiting null geodesics is determined by $\CurlyR=0$, $\partial_r \CurlyR = 0$.  
Due to the form of $\CurlyR$, the location of orbiting null geodesics depends only on the ratios $\GeodesicKCarter/\GeodesicLz^2, \GeodesicEnergy/\GeodesicLz$. One finds that orbiting null geodesics exist for a range of radii 
$r_1 \leq r \leq r_2$, with $r_+ < r_1 < 3M < r_2$.  Here $r_1, r_2$ depend on $a,M$ and as  $|a| \nearrow M$, $r_1 \searrow r_+$, and $r_2 \nearrow 4M$.
The orbits at $r_1, r_2$ are restricted to the equatorial plane, those at $r_1$ are  corotating, while those at $r_2$ are counterrotating. 
For $r_1 < r < r_2$, the range of $\theta$ depends on $r$. 
There is $r_3 = r_3(a,M)$, $r_1 < r_3 < r_2$ such that the orbits at $r_3$ reach the poles, i.e. $\theta = 0$, $\theta =\pi$, see figure \ref{fig:photonsphere}. For such geodesics, it holds that $\GeodesicLz = 0$.

\begin{figure}
    \centering
    \begin{subfigure}[b]{0.4\textwidth}
         \includegraphics[width=.7\textwidth]{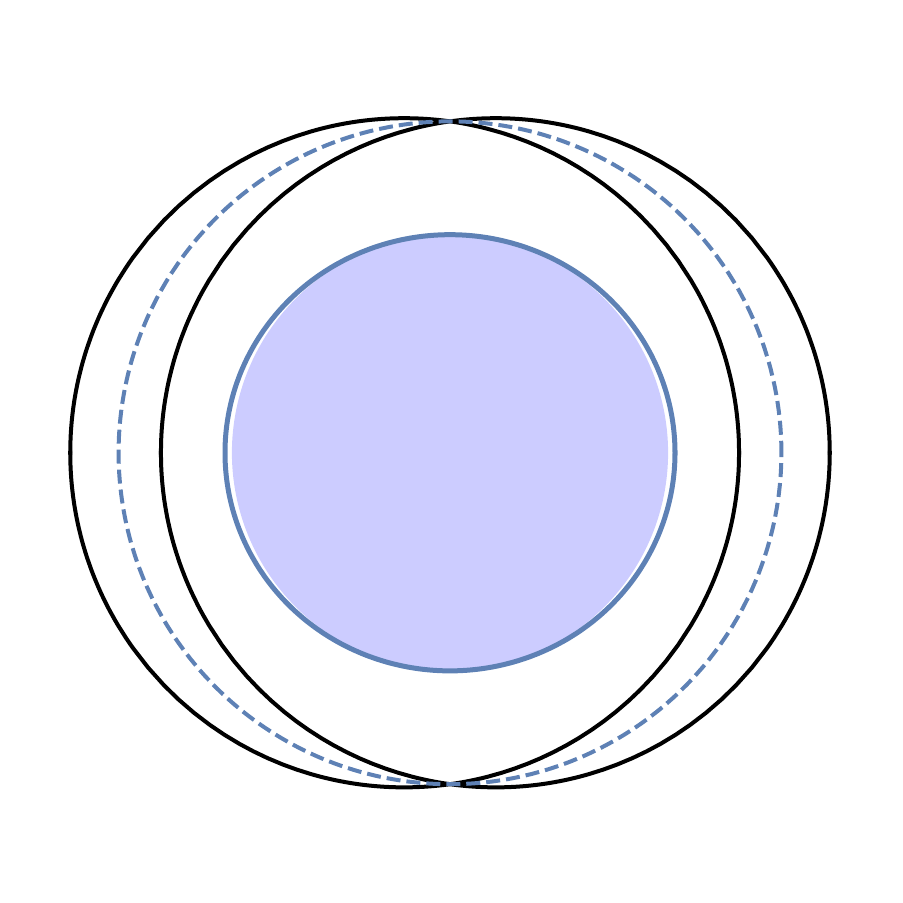}
\caption{}
                \label{fig:Kerr-trap-low}
    \end{subfigure}
    ~ %
 \begin{subfigure}[b]{0.4\textwidth}
         \includegraphics[width=.7\textwidth]{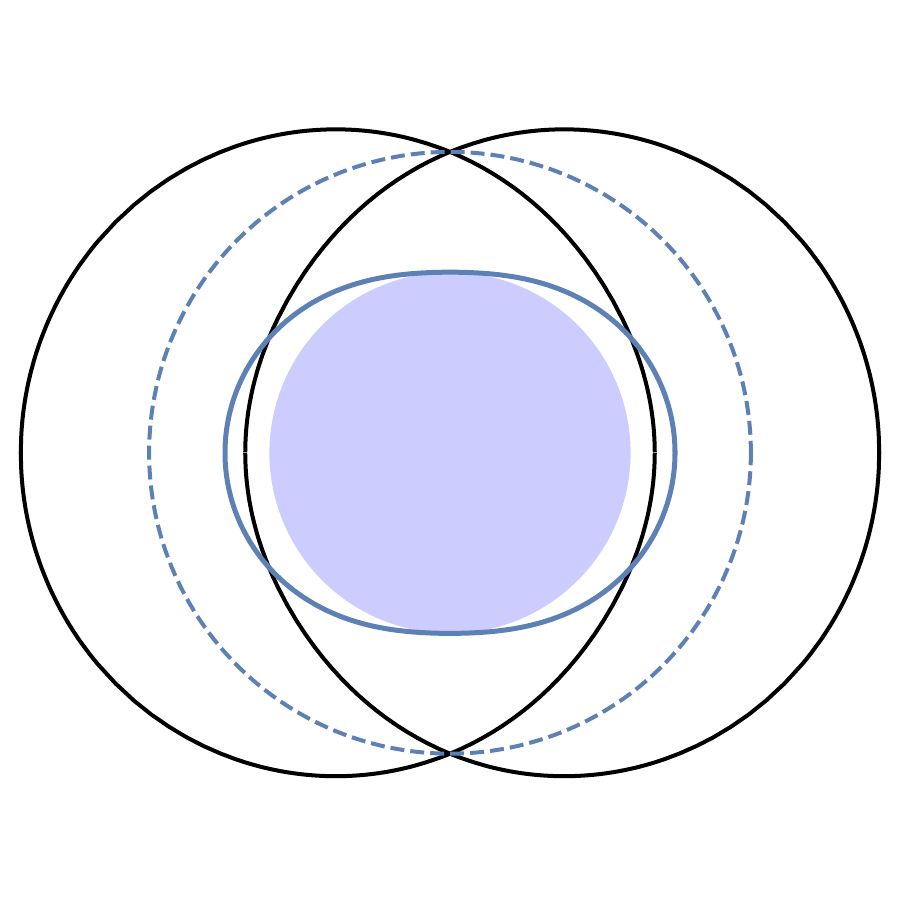}
\caption{}
 \label{fig:Kerr-trap-high}
    \end{subfigure}   
\caption{The Kerr photon region. In subfigure (\subref{fig:Kerr-trap-low}), $|a| \ll M$ and the ergoregion, see section \ref{sec:kerrmetric}, is well separated from the photon region (bordered in black). The radius $r_3$ where geodesics reach the poles is indicated by a grey, dashed line. In subfigure (\subref{fig:Kerr-trap-high}), $|a|$ is close to $M$ and the ergoregion overlaps the photon region.}
 \label{fig:photonsphere}
\end{figure}

\begin{figure}
    \centering
    \begin{subfigure}[b]{0.4\textwidth}
         \includegraphics[width=.7\textwidth]{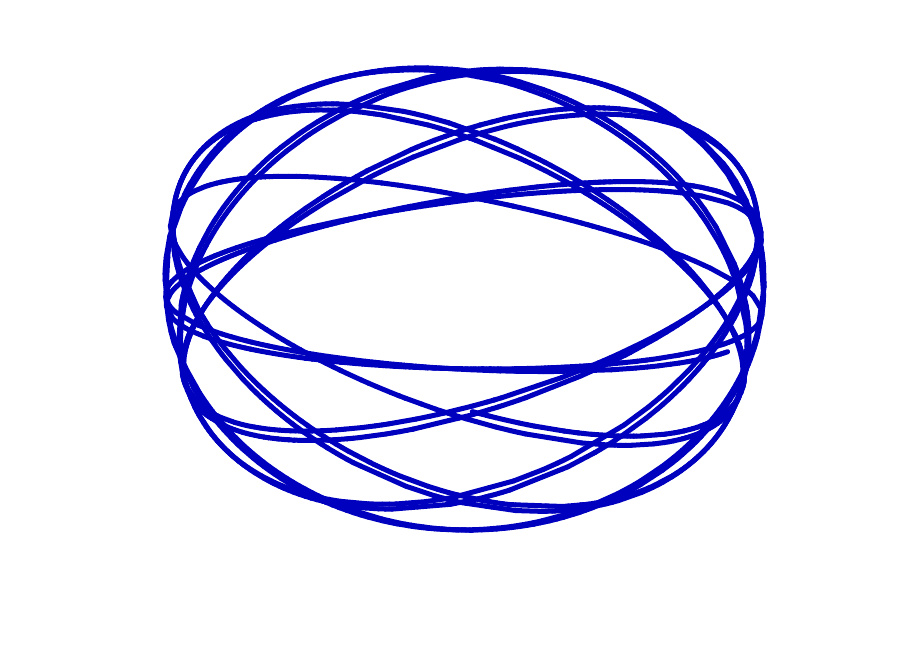}
\caption{}
                \label{fig:Kerrgeod1}
    \end{subfigure}
    ~ %
 \begin{subfigure}[b]{0.33\textwidth}
         \includegraphics[width=.7\textwidth]{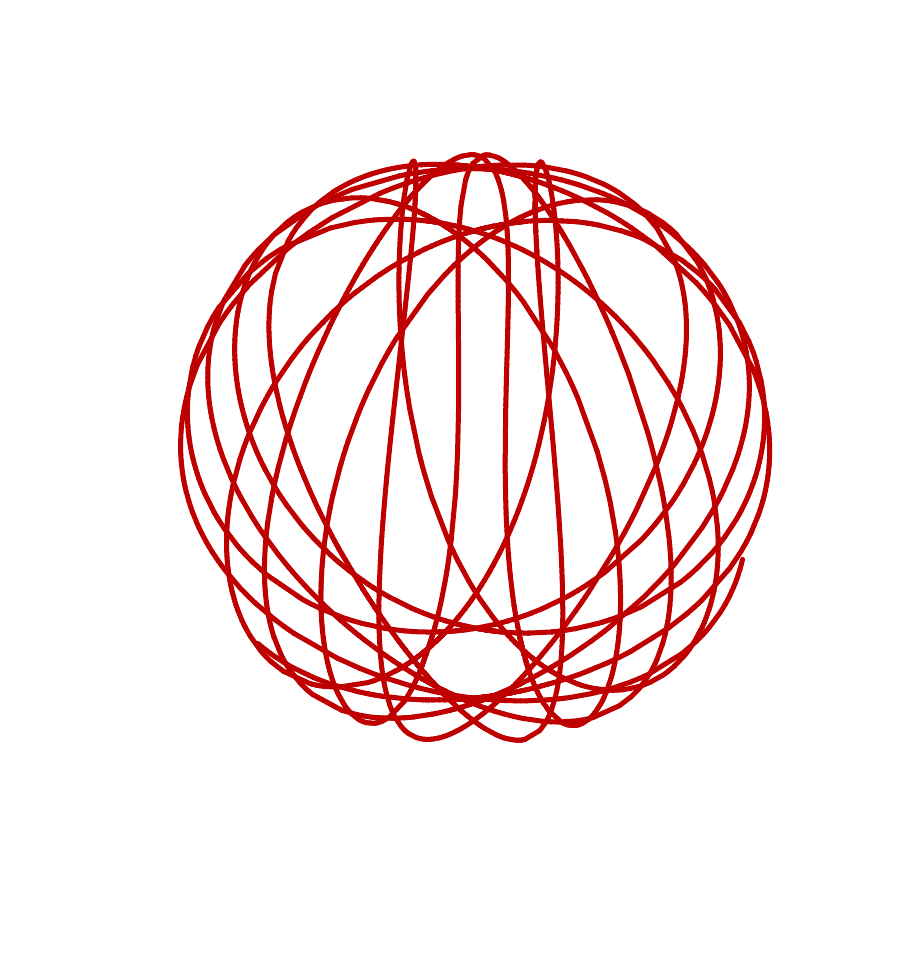}
\caption{}
 \label{fig:Kerrgeod2}
    \end{subfigure}   
\caption{Examples of orbiting null geodesics in Kerr with $a=M/2$. In subfigure (\subref{fig:Kerrgeod1}), the $\GeodesicKCarter/\GeodesicLz^2$ is small, while in subfigure (\subref{fig:Kerrgeod2}), this constant is larger.}
 \label{fig:Kerrgeodesics}
\end{figure}

For the following discussion, it is convenient to
introduce 
\begin{align*}
\GeodesicQ
={}& \GeodesicKCarter -2a\GeodesicEnergy\GeodesicLz-\GeodesicLz^2 
= \TensorQ^{ab} \dot \Geodesic_a \dot \Geodesic_b , 
\end{align*} 
where 
\begin{align} 
\TensorQ^{ab} ={}& (\partial_\theta)^a(\partial_\theta)^b+\frac{\cos^2\theta}{\sin^2\theta}(\partial_\phi)^a(\partial_\phi)^b+a^2\sin^2\theta(\partial_t)^a(\partial_t)^b . 
\label{eq:TensorQdef}
\end{align}
By construction, $\GeodesicQ$ is a sum
of conserved quantities, and is therefore conserved. Further, it is non-negative, since it is a sum of
non-negative terms. In the following we use $(\GeodesicEnergy, \GeodesicLz, \GeodesicQ)$ as parameters for null geodesics. Since we are considering only null geodesics, there is no loss of generality compared to using $(\GeodesicEnergy, \GeodesicLz, \GeodesicKCarter)$ as parameters. 

For a null geodesic with given parameters $(\GeodesicEnergy,\GeodesicLz,\GeodesicQ)$, a simple turning point analysis shows that there is a number $\rorbit \in (\rp, \infty)$ so that the quantity $(r-\rorbit)\dot\Geodesic^r$ increases overall. This quantity corresponds to the energy $\GenEnergyGeodesic{\vecMGeodesic}$ for the vector field $\vecMGeodesic=-(r-\rorbit)\dr$.   
Following this idea, we may now look for a function $\fnMrGeodesic$ which will play the role of $-(r - \rorbit)$, so that for $\vecMGeodesic=\fnMrGeodesic\dr$, the energy $\GenEnergyGeodesic{\vecMGeodesic}$ is non-decreasing for all $\lambda$ and not merely non-decreasing overall. For $a\not=0$, both $\rorbit$ and $\fnMrGeodesic$ will necessarily depend on both the Kerr parameters $(M,a)$ and the constants of motion $(\GeodesicEnergy,\GeodesicLz,\GeodesicQ)$; the function $\fnMrGeodesic$ will also depend on $r$, but no other variables. 

We define $\vecMGeodesic^a=\fnMrGeodesic (\dr)^a$ with 
$$
\fnMrGeodesic = \fnMrGeodesic(r;M,a,\GeodesicEnergy,\GeodesicLz,\GeodesicQ)
$$
It is important to note that this is a map from the tangent bundle to the tangent bundle, and hence $\vecMGeodesic^a = \fnMrGeodesic (\partial_r)^a$ cannot be viewed as a standard vector field, which is a map from the manifold to the tangent bundle. 

To derive a monotonicity formula, we wish to choose $\fnMrGeodesic$ so that $\GenEnergyGeodesic{\vecMGeodesic}$ has a non-negative derivative. We define the covariant derivative of $\vecMGeodesic$ by holding the values of $(\GeodesicEnergy,\GeodesicLz,\GeodesicQ)$ fixed and computing the covariant derivative as if $\vecMGeodesic$ were a regular vector field. Similarly, we define $\Lie_{\vecMGeodesic}\gMetric^{ab}$ by fixing the values of the constants of geodesic motion. Since the constants of motion have zero derivative along null geodesics, equation \eqref{eq:deformForGeodesics} remains valid. 

Recall that null geodesics are conformally invariant up to reparameterization. Hence, it is sufficient to work with the conformally rescaled metric $\KSigma\gMetric^{ab}$. Furthermore, since $\gamma$ is a null geodesic, for any function $\fnMpGeodesic$, we may subtract $\fnMpGeodesic\KSigma\gMetric^{ab}\dot\gamma_a\dot\gamma_a$ wherever it is convenient. Thus, the change in $\GenEnergyGeodesic{\vecMGeodesic}$ is given as the integral of 
\begin{align*}
\KSigma\dot\gamma_a\dot\gamma_b\nabla^{(a}\vecMGeodesic^{b)}
&=\left(-\frac12\Lie_{\vecMGeodesic}(\KSigma\gMetric^{ab}) - \fnMpGeodesic\KSigma\gMetric^{ab}\right)\dot\gamma_a\dot\gamma_b
\end{align*}

The Kerr metric can be written as 
\begin{align}
\KSigma\gMetric^{ab}&= - \KDelta(\dr)^a(\dr)^b - \frac{1}{\KDelta}\CurlyR^{ab} , \label{eq:Kerrg-CurlyR}
\end{align}
where the tensorial form of $\CurlyR^{ab}$ can be read off from the earlier definitions. 
We now calculate $-\Lie_{\vecMGeodesic}\gMetric^{ab}\dot\gamma_a\dot\gamma_b$ using \eqref{eq:Kerrg-CurlyR}. Ignoring distracting factors of $\KSigma$, $\KDelta$, the most important terms are 
\begin{align*}
-2(\dr\fnMrGeodesic)\dot\gamma_r\dot\gamma_r +\fnMrGeodesic(\dr\CurlyR^{ab})\dot\gamma_a\dot\gamma_b =-2(\dr\fnMrGeodesic)\dot\gamma_r\dot\gamma_r +\fnMrGeodesic(\dr\CurlyR).
\end{align*} 
The second term in this sum will be non-negative if $\fnMrGeodesic=\dr\CurlyR(r;M,a;\GeodesicEnergy,\GeodesicLz,\GeodesicQ)$. Recall that the vanishing of $\dr\CurlyR(r;M,a;\GeodesicEnergy,\GeodesicLz,\GeodesicQ)$ is one of the two conditions for orbiting null geodesics. With this choice of $\fnMrGeodesic$, the instability of the null geodesic orbits ensures that, for these null geodesics, the coefficient in the first term, $-2(\dr\fnMrGeodesic)$, will be positive. These observations motivate the form of $\fnMrGeodesic$ which yields non-negativity for all null geodesics. 

It remains to make explicit choices of $\fnMrGeodesic$ and $\fnMpGeodesic$. Once these choices are made, the necessary calculations are straight-forward but rather lengthy.  Let $\fnMna$ and $\fnMnb$ be smooth functions of $r$ and the Kerr parameters $(M,a)$. Let $\DiffCurlyRTilde$ denote $\dr(\frac{\fnMna}{\KDelta}\CurlyR(r;M,a;\GeodesicEnergy,\GeodesicLz,\GeodesicQ))$ and choose $\fnMrGeodesic=\fnMna\fnMnb\DiffCurlyRTilde$ and $\fnMpGeodesic=(1/2)(\dr\fnMna)\fnMnb\DiffCurlyRTilde$. In terms of these functions,
\begin{align}
\KSigma\dot\gamma_a\dot\gamma_b \nabla^{(a}\vecMGeodesic^{b)}
&=\half \fnMnb (\DiffCurlyRTilde)^2 -\fnMna^{1/2}\KDelta^{3/2} \left(\dr\left(\fnMnb\frac{\fnMna^{1/2}}{\KDelta^{1/2}}\DiffCurlyRTilde\right)\right) \dot\gamma_r^2. 
\label{eq:GeodesicBulk}
\end{align}
If $\fnMna$ and $\fnMnb$ are chosen to be positive, then the first term on the right hand side of \eqref{eq:GeodesicBulk} which contains a square $(\DiffCurlyRTilde)^2$ is non-negative. If we now take $\fnMna=\fnMca=\KDelta(r^2+a^2)^{-2}$ and 
$\fnMnb=\fnMcb=(r^2+a^2)^4/(3r^2-a^2)$, then\footnote{Equation \eqref{eq:GeodesicsDDiffCurlyRTTilde}corrects a misprint in \cite[Eq. (1.15b)]{AnderssonBlue:KerrWave}. }

\begin{align}
- \dr\left(\fnMnb\frac{\fnMna^{1/2}}{\KDelta^{1/2}}\DiffCurlyRTilde\right)
&=2\frac{3r^4+a^4}{(3r^2-a^2)^2}\GeodesicLz^2 + 
  2\frac{3r^4-6a^2r^2-a^4}{(3r^2-a^2)^2}\GeodesicQ .\label{eq:GeodesicsDDiffCurlyRTTilde}
\end{align}
The coefficient of $\GeodesicQ$ is positive for $r > r_+$  when $|a|< 3^{1/4}2^{-1/2} M \cong 0.93 M$. Since  
$\GeodesicQ$ is non-negative, the right-hand side of \eqref{eq:GeodesicsDDiffCurlyRTTilde} is non-negative, and hence also the right-hand side of equation \eqref{eq:GeodesicBulk} is non-negative, for this range of $a$. 
Since equation \eqref{eq:GeodesicBulk} gives the rate of change, the energy $\GenEnergyGeodesic{\vecMGeodesic}$ is monotone. 

These calculations reveal useful information about the geodesic motion. The positivity of the term on the right-hand side of \eqref{eq:GeodesicsDDiffCurlyRTTilde} shows that $\DiffCurlyRTilde$ can have at most one root, which must be simple. In turn, this shows that $\CurlyR$ can have at most two roots. For orbiting null geodesics $\CurlyR$ must have a double root, which must coincide with the root of $\DiffCurlyRTilde$. It is convenient to think of the corresponding value of $r$ as being $\rorbit$.

The first term in \eqref{eq:GeodesicBulk} vanishes at the root of $\DiffCurlyRTilde$, as it must so that $\GenEnergyGeodesic{\vecMGeodesic}$ can be constantly zero on the orbiting null geodesics. When $a=0$, the quantity $\DiffCurlyRTilde$ reduces to $-2(r-3M)r^{-4}(\GeodesicLz^2+\GeodesicQ)$, so that the orbits occur at $r=3M$. The continuity in $a$ of $\DiffCurlyRTilde$ guarantees that its root converges to $3M$ as $a\rightarrow 0$ for fixed $(\GeodesicEnergy,\GeodesicLz,\GeodesicQ)$. 

From the geometrics optics approximation, it is natural to imagine that the monotone quantity constructed in this section for null geodesics might imply the existence of monotone quantities for fields, which would imply some form of dispersion. 
For the wave equation, this is true. In fact, the above discussion, when carried over to the case of the wave equation, closely parallels the proof of the Morawetz estimate for the wave equation given in \cite{AnderssonBlue:KerrWave}, see section \ref{sec:fields} below. The quantity 
$(\dot\Geodesic_\alpha\dot\Geodesic_\beta)(\nabla^{(\alpha}\vecX^{\beta)}) $ corresponds to the Morawetz density, i.e. the divergence of the momentum corresponding to the Morawetz vector field. The role of the conserved quantities $(\GeodesicEnergy,\GeodesicLz,\GeodesicQ)$ for geodesics is played, in the case of fields, by the energy fluxes defined via second order symmetry operators corresponding to these conserved quantities. 
The fact that the quantity $\CurlyR$ vanishes quadratically on the trapped orbits is reflected in the Morawetz estimate for fields, by a quadratic degeneracy of the Morawetz density at the trapped orbits.  

\subsection{Dispersive estimates for fields} \label{sec:fields}

As discussed in section \ref{sec:geodesics}, one may construct a suitable function of the conserved quantities for null geodesics in the Kerr spacetime which is monotone along the geodesic flow. This function may be viewed as arising from a generalized vector field on phase space. The monotonicity property implies, as discussed there, that non-trapped null geodesics disperse, in the sense that they leave any stationary region in the Kerr space time.  As mentioned in section \ref{sec:geodesics}, in view of the geometric optics approximation for the wave equation, such a monotonicity property for null geodesics reflects the tendency for waves in the Kerr spacetime to disperse. 

At the level of the wave equation, the analogue of the just mentioned monotonicity estimate is called the Morawetz estimate. For the wave equation $\nabla^a \nabla_a \psi = 0$, a Morawetz estimate provides a current $J_a$ defined in terms of $\psi$ and some of its derivatives, with the property that $\nabla^a J_a$ has suitable positivity properties, and that the flux of $J_a$ can be controlled by a suitable energy defined in terms of the field. 

Let $\psi$ be a solution of the wave equation $\nabla^a \nabla_a \psi = 0$. 
Define the current $J_a$ by 
$$
J_a = T_{ab} \vecMprimary^b +  \half \scalMprimary(\bar \psi \nabla_a \psi  + \psi \nabla_a \bar\psi)  - \half (\nabla_a \scalMprimary ) \psi \bar \psi.
$$
where $T_{ab}$ is the stress-energy tensor given by \eqref{eq:Tab-wave}. 
We have 
\begin{equation}\label{eq:JJJJ}
\nabla^a J_a = T_{ab} \nabla^{(a} \vecMprimary^{b)} + \scalMprimary \nabla^c \psi \nabla_c \bar \psi - \half (\nabla^c \nabla_c \scalMprimary) \psi \bar \psi .
\end{equation}
We now specialize to Minkowski space, with the line element 
$g_{ab}dx^a dx^b = dt^2 - dr^2 - d\theta^2 - r^2 \sin^2\theta d\phi^2$. 
Let 
$$
E(\tau) = \int_{\{ t= \tau\}} T_{tt} d^3 x
$$
be the energy of the field at time $\tau$, where $T_{tt}$ is the energy density. The energy is conserved, so that $E(t)$ is independent of $t$. 

Setting $\vecMprimary^a = r(\partial_r)^a$, we have 
\begin{equation} \label{eq:deform} 
\nabla^{(a} \vecMprimary^{b)} = g^{ab} - (\partial_t)^a (\partial_t)^b . 
\end{equation} 
With $q = 1$,  we get
$$
\nabla^a J_a = - T_{tt} .
$$
With the above choices, the bulk term $\nabla^a  J_a$ has a sign. This method can be used to prove dispersion for solutions of the wave equation. In particular, by introducing suitable cutoffs, one finds that for any $R_0 > 0$, there is a constant $C$, so that  
\begin{align}
\int_{t_0}^{t_1} \int_{|r| \leq R_0}  T_{tt} d^3 x dt 
\leq C(E(t_0) +E(t_1)) 
\leq 2C E(t_0) , \label{eq:basicmor} 
\end{align}
see \cite{MR0234136}. 
The local energy, $\int_{|r| \leq R_0} T_{tt} d^3 x$, is a function of
time. By \eqref{eq:basicmor} it is integrable in $t$, and hence it must decay to zero
as $t\rightarrow\infty$, at least sequentially. This shows that the field disperses. 
Estimates of this type are called
Morawetz or integrated local energy decay estimates. 

For a solution $\phi_{AB}$ of the Maxwell equation $(\sCurlDagger_{2,0} \phi)_{AA'} = 0$, the stress-energy tensor $T_{ab}$ given by 
$$
T_{ab} = \phi_{AB} \bar \phi_{A'B'}
$$
is conserved, $\nabla^a T_{ab} = 0$. Further, $T_{ab}$ has trace zero, 
 with $T^a{}_a = 0$. 
 
Restricting to Minkowski space and setting $J_a = T_{ab} \vecMprimary^b$, with $\vecMprimary^a = r(\partial_r)^a$ we have
$$
\nabla^a J_a = - T_{tt}
$$
which again gives local energy decay for the Maxwell field on Minkowski space. 

For the wave equation on Schwarzschild we can choose  
\begin{subequations}
\begin{align}
\vecMprimary^{a}={}&\frac{(r - 3 M) (r - 2 M)}{3 r^2} (\partial_r)^{a} ,\\
\scalMprimary={}&\frac{6 M^2 - 7 M r + 2 r^2}{6 r^3}.
\end{align}
\end{subequations}
This gives

\begin{subequations}
\begin{align}\label{eq:Afirst}
-\nabla^{(a}\vecMprimary^{b)}={}&- \frac{M g^{ab} (r - 3 M)}{3 r^3}
 +  \frac{M (r - 2 M)^2 (\partial_r)^{a} (\partial_r)^{b}}{r^4}\nonumber\\
& +  \frac{(r - 3 M)^2 ((\partial_\theta)^{a} (\partial_\theta)^{b} + \csc^2\theta (\partial_\phi)^{a} (\partial_\phi)^{b})}{3 r^5},\\
- \nabla_{a}J^{a}={}&\frac{M |\partial_r \psi|^2 (r - 2 M)^2}{r^4}
 + \frac{\bigl(|\partial_\theta \psi|^2 + |\partial_\phi\psi|^2 \csc^2\theta\bigr) (r - 3 M)^2}{3 r^5}\nonumber\\
& + \frac{M |\psi|^2 (54 M^2 - 46 M r + 9 r^2)}{6 r^6}. \label{eq:divJ}
\end{align}
\end{subequations} 
Here, $\vecMprimary^a$ was chosen so that the last two terms \eqref{eq:Afirst} have good signs. The form of $\scalMprimary$ given here was chosen to eliminate the $|\partial_t \psi|^2$ term in \eqref{eq:divJ}. 
The first terms in \eqref{eq:divJ} are clearly non-negative, while the last is of lower-order and can be estimated using a Hardy estimate \cite{AnderssonBlue:KerrWave}.
The effect of trapping in Schwarzschild at $r=3M$ is manifested in the fact that the angular derivative term 
vanishes at $r=3M$. 

In the case of the wave equation on Kerr, the above argument using a classical vector field cannot work due to the complicated structure of the trapping. However, making use of higher-order currents constructed using second order symmetry operators for the wave equation, and a generalized Morawetz vector field analogous to the vector field $\vecMprimary^a$ as discussed in section \ref{sec:geodesics}. This approach has been carried out in detail in \cite{AnderssonBlue:KerrWave}. 

If we apply the same idea for the Maxwell field on Schwarzschild, there is no reason to expect that local energy decay should hold, in view of the fact that the Coulomb solution is a time-independent solution of the Maxwell equation which does not disperse. In fact, with 
\begin{align}
\vecMprimary^{a}={}& \fnMrMorawetz(r) \Bigl(1
 -  \frac{2 M}{r} \Bigr) (\partial_r)^{a},\\
\intertext{we have}
- T_{ab} \nabla^{(a} \vecMprimary^{b)} ={}& - \phi^{AB} \bar{\phi}^{A'B'} (\sTwist_{1,1} \vecMprimary)_{ABA'B'}\\
={}&\bigl(|\phi_{0}|^2 + |\phi_{2}|^2\bigr) \frac{(r - 2 M)}{2 r} \fnMrMorawetz'(r) \nonumber\\
& 
- \frac{|\phi_{1}|^2 \bigl(r (r-2M) \fnMrMorawetz'(r) -  2 \fnMrMorawetz(r) (r - 3 M) \bigr)}{r^2}.
\label{eq:ssss} 
\end{align}
If $\fnMrMorawetz'$ is chosen to be positive, then the coefficient of the extreme components in \eqref{eq:ssss} is positive. However, at $r=3M$, the coefficient of the middle component is necessarily of the opposite sign. It is possible to show that no choice of $\fnMrMorawetz$ will give positive coefficients for all components in \eqref{eq:ssss}. 

The dominant energy condition, that  $T_{ab} V^a W^b \geq 0$ for all causal vectors  $V^a, W^a$
is a common and important condition on stress energy tensors. 
In Riemannian geometry, a natural condition on a
symmetric $2$-tensor $T_{ab}$ would be non-negativity, i.e. the condition that
for all $X^a$, one has $T_{ab}X^a X^b\geq
0$. 

However, in order to prove dispersive estimates for null geodesics and the wave equation,
the dominant energy condition on its own is not sufficient and non-negativity cannot be expected for stress energy tensors. Instead, a useful condition to consider is non-negativity modulo trace
terms, i.e. the condition that for every $X^a$ there is a $q$ such that
$T_{ab} X^a X^b +q T^a{}_a \geq 0$. For null
geodesics and the wave equation, the tensors $\dot\gamma_a\dot\gamma_b$
and $\nabla_a u\nabla_b u=T_{ab} +T^\gamma{}_\gamma
g_{ab}$ are both non-negative, so
$\dot\gamma_a\dot\gamma_b$ and $T_{ab}$ are
non-negative modulo trace terms. 

From equation \eqref{eq:Afirst}, we see
that $-\nabla^{(a}A^{b)}$ is of the form $f_1 g^{ab}
+f_2\partial_r^a\partial_r^b
+f_3\partial_\theta^a\partial_\theta^b+f_4\partial_\phi^a\partial_\phi^b$ where $f_2$, $f_3$ and $f_4$ are
non-negative functions. That is $-\nabla^{(a}A^{b)}$ is a sum
of a multiple of the metric plus a sum of terms of the form of a
non-negative coefficient times a vector tensored with itself. Thus,
from the non-negativity modulo trace terms, for null geodesics and the
wave equation respectively, there are functions $q$ such that
$\dot\gamma_a\dot\gamma_b \nabla^a
A^b=\dot\gamma_a\dot\gamma_b \nabla^a A^b +q
g^{ab}\dot\gamma_a\dot\gamma_b \leq 0$ and
$T_{ab}\nabla^a A^b +qT^a{}_a\leq 0$. For
null geodesics, since
$g^{ab}\dot\gamma_a\dot\gamma_b=0$, the $q$ term can
be ignored. For the wave equation, one can use the terms involving $\scalMprimary$
in equations \eqref{eq:JJJJ}, to cancel the
$T^a{}_a$
term in $\nabla^a J_a$. For the wave equation, this gives
non-negativity for the first-order terms in $-\nabla^a J_a$,
and one can then hope to use a Hardy estimate to control the zeroth
order terms.

If we now consider the Maxwell equation, we have the fact that the Maxwell stress energy tensor is traceless, $T^a{}_a = 0$ and does not satisfy
the non-negativity condition. Therefore it also does not satisfy the condition of non-negativity modulo trace. This appears to be the
fundamental underlying obstruction to proving a Morawetz estimate
using $T_{ab}$. This can be seen as a manifestation of the fact that the Coulomb solution does not disperse. 
\begin{wrapfigure}{r}{0.2\textwidth}
\raisebox{-0.5\height}{\includegraphics{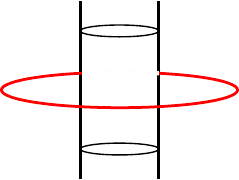}}
\end{wrapfigure}
In fact, it is immediately clear that the Maxwell stress energy cannot be used directly to prove dispersive estimates since it does not vanish for the Coulomb field \eqref{eq:coulomb} on the Kerr spacetime. 
We remark that the existence of the Coulomb solution on the Kerr spacetime is a consequence of the facts exterior of the black hole contains non-trivial 2-spheres, and the existence of two conserved charge integrals $\int_S F_{ab} d\sigma^{ab}$, $\int_S (*F)_{ab} d\sigma^{ab}$. Hence this is valid also for dynamical black hole spacetimes.

\section{Symmetry operators}  \label{sec:symop} 
A symmetry operator for a field equation is an operator which takes solutions to solutions.  
In order to analyze higher spin fields on the Kerr spacetime, it is important to gain an understanding of the symmetry operators for this case.  
In the paper \cite{ABB:symop:2014CQGra..31m5015A} we have given a complete characterization of those spacetimes admitting symmetry  operators of second order for the field equations of spins $0, 1/2, 1$, i.e. the conformal wave equation, the Dirac-Weyl equation and the Maxwell equation, respectively, and given the general form of the symmetry operators, up to equivalence. In order to simplify the presentation here, we shall discuss only the spin-$1$ case, and restrict to spacetimes admitting a valence $(2,0)$ Killing spinor $\kappa_{AB}$. We first give some background on the wave equation.

\subsection{Symmetry operators for the Kerr wave equation} \label{sec:kerrwave}
As shown by Carter \cite{carter:1977PhRvD..16.3395C}, if $K_{ab}$ is a Killing tensor in a Ricci flat spacetime, the operator 
\begin{equation}\label{eq:OpK=K} 
\OperatorKCarter = \nabla_a K^{ab} \nabla_b
\end{equation} 
is a commuting symmetry operator for the d'Alembertian,
$$
[\nabla^a \nabla_a , \OperatorKCarter ] = 0
$$
In particular there is a second order symmetry operator for the wave equation, i.e. an operator which maps solutions to solutions, 
$$
\nabla^a \nabla_a \psi = 0 \quad \Rightarrow \quad \nabla^a \nabla_a \OperatorKCarter \psi = 0 
$$
Due to the form of the Carter Killing tensor, $\TensorKCarter_{ab}$, cf. \eqref{eq:KDefinition}, the operator $\OperatorKCarter$ defined by \eqref{eq:OpK=K} contains derivatives with respect to all coordinates. 

Recall that $\nabla^a \nabla_a = \frac{1}{\mu_g} \partial_a \mu_g g^{ab} \partial_b$, where $\mu_g = \sqrt{\det(g_{ab})}$ is the volume element. For Kerr in Boyer-Lindquist coordinates, we have from \eqref{eq:Kerrvol} that $\mu_g = \KSigma \bmu$, with $\bmu = \sin\theta$. 
After rescaling the d'Alembertian by $\KSigma$, and using the just mentioned facts, one finds 
\begin{align}
\KSigma \nabla^a \nabla_a ={}& - \partial_r \KDelta \partial_r + 
\frac{\CurlyR(r;\partial_t,\partial_\phi,\OperatorQ)}{\KDelta} \label{eq:SigmadAl}
\end{align} 
where  
$$
\OperatorQ = \frac{1}{\bmu} \partial_a \bmu \TensorQ^{ab} \partial_b
$$
In view of the form of $\TensorQ^{ab}$ given in \eqref{eq:TensorQdef}, we see that $\OperatorQ$ contains derivatives only with respect to $\theta,\phi,t$, but not with respect to $r$. Thus, it is clear from \eqref{eq:SigmadAl} that $\OperatorQ$ is a commuting symmetry operator for the rescaled d'Alembertian $\Sigma \nabla^a \nabla_a$, 
$$
[ \Sigma \nabla^a \nabla_a , \OperatorQ ] = 0
$$
In addition to the symmetry operator $\OperatorQ$ related to the Carter constant, we have the second order symmetry operators generated by the Killing fields $\xi^a \nabla_a = \partial_t$, $\eta^a \nabla_a = \partial_\phi$. The operator $\OperatorQ$ can be termed a hidden symmetry, since it cannot be represented in terms of operators generated by the Killing fields.

The above shows that we can write 
$$
\KSigma \nabla^a \nabla_a = \Rop + \Sop
$$
where the operators $\Rop, \Sop$ commute, $[\Rop,\Sop]=0$, and $\Rop$ contains derivatives with respect to the non-symmetry coordinate $r$, and the two symmetry coordinates $t,\phi$, while $\Sop$ contains derivatives with respect to the non-symmetry coordinate $\theta$, and with respect to $t,\phi$.

By making a separated ansatz 
$$
\psi_{\omega,\ell,m}(t,r,\theta,\phi) = e^{-i\omega t} e^{im\phi} \Rsol_{\omega,\ell,m}(r) \Ssol_{\omega,\ell,m}(\theta)
$$
the equation $\nabla^a \nabla_a \psi = 0$ becomes a pair of scalar ordinary differential equations
\begin{subequations} \label{eq:RSeq} 
\begin{align} 
\Rop \Rsol + \lambda \Rsol ={}& 0 \label{eq:Req} \\
\Sop \Ssol ={}& \lambda \Ssol \label{eq:Seq}
\end{align}
\end{subequations} 
where $\lambda = \lambda_{\omega,\ell,m}$. Here it should be noted that equation \eqref{eq:Seq} is to be considered as a boundary value problem on $[0,\pi]$ with boundary conditions determined by the requirement that $\phi$ be smooth. 
In the Schwarzschild case $a = 0$, we can take $\Sop = \angDelta$, the angular Laplacian. The eigenfunctions of $\angDelta$ are the spherical harmonics $Y_{\ell,m}(\theta,\phi) = e^{im\phi}Y_\ell(\theta)$. The eigenvalues of $\angDelta$ are $\lambda_{\ell,m} = - \ell(\ell+1)$.

The solutions to the eigenproblem 
$\Sop \Ssol = \lambda \Ssol$ are the spheroidal harmonics, the eigenvalues in this case are not known in closed form, depend on the time frequency $\omega$, and are indexed by $\ell,m$. For real $\omega$, it is known that the eigensystem is complete, but for general $\omega$ this is not known. 

One may now apply a Fourier transform and represent a typical solution $\psi$ to the wave equation in the form
$$
\psi = \int d\omega \sum_{\ell,m}  e^{-i\omega t} e^{im\phi} \Rsol_{\omega,\ell,m} \Ssol_{\omega,\ell,m} ,
$$
analyze the behavior of the separated modes $\psi_{\omega,\ell,m}$, and recoved estimates for $\psi$ after inverting the Fourier transform. 
In order to do this, one must show a priori that the Fourier transform can be applied. This can be done by applying cutoffs, and removing these after estimates have been proved using Fourier techniques. This approach has been followed in eg. \cite{DafermosRodnianski:KerrEnergy,AnderssonNicolasBlue,andersson:blue:maxwell:2013arXiv1310.2664A}. In recent work by Dafermos, Rodnianski and Shlapentokh-Rothman, see \cite{2014arXiv1402.7034D}, proving boundedness and decay for the wave equation on Kerr for the whole range $|a| < M$, makes use of the technical condition of time integrability, i.e. that the solution to the wave equation and its derivatives to a sufficiently high order is bounded in $L^2$ on time lines,
$$
\int_{-\infty}^{\infty} dt |\partial^{\alpha} \psi(t,r,\theta,\phi)| 
$$
This condition is consistent with integrated local energy decay and is removed at the end of the argument. 

However, by working directly with currents defined in terms of second order symmetry operators, one may prove a Morawetz estimate directly for the wave equation on the Kerr spacetime. This was carried out for the case $|a| \ll M$ in  \cite{AnderssonBlue:KerrWave}. This involves introducing a generalization of the vector field method to allow for currents defined in terms of generalized, operator valued, vector fields. These are operator analogs of the generalized vector field $\vecMGeodesic^a$ introduced in section \ref{sec:geodesics}. 

Fundamental for either of the above mentioned approaches, is that the analysis of the wave equation on the Kerr spacetime is based on the hidden symmetry manifested in the existence of the Carter constant, or the conserved quantity $\GeodesicQ$, and its corresponding symmetry operator $\OperatorQ$. 

\subsection{Symmetry operators for the Maxwell field} \label{eq:symMax}
There are two spin-$1$ equations (left and right) depending on the helicity  of the spinor. These are 
$$ 
(\sCurlDagger_{2,0} \phi)_{AA'} = 0 \quad \text{(left), \quad and } \quad (\sCurl_{0,2} \varphi)_{AA'} = 0 \quad \text{ (right)}
$$
 The real Maxwell equation $\nabla^a F_{ab} = 0$, $\nabla_{[a} F_{bc]} = 0$ for a real two form $F_{ab} = F_{[ab]}$ is equivalent to either the right or the left Maxwell equations. 
Henceforth we will always assume that $\phi_{AB}$ solves the left Maxwell equation.

Given a conformal Killing vector $\GenVec^{AA'}$, we follow \cite[Equations (2) and (15)]{anco:pohjanpelto:2003:CRM:MR2056970}, see also \cite{anco:pohjanpelto:2003:ProcRSoc:MR1997098},
and define a conformally weighted Lie derivative acting on a symmetric valence $(2s,0)$ spinor field as follows 
\begin{definition}
For $\GenVec^{AA'} \in \ker \sTwist_{1,1}$, and $\varphi_{A_1\dots A_{2s}}\in \mathcal{S}_{2s,0}$, we define
\begin{align}
\hat{\mathcal{L}}_{\GenVec}\varphi_{A_1\dots A_{2s}}\equiv{}&\GenVec^{BB'} \nabla_{BB'}\varphi_{A_1\dots A_{2s}}+s \varphi_{B(A_2\dots A_{2s}} \nabla_{A_1)B'}\GenVec^{BB'}
 + \tfrac{1-s}{4} \varphi_{A_1\dots A_{2s}} \nabla^{CC'}\GenVec_{CC'}.
\end{align}
\end{definition}
If $\GenVec^a$ is a conformal Killing field, then 
$(\sCurlDagger_{2,0} \hat{\mathcal{L}}_\GenVec \varphi)_{AA'} = \hat{\mathcal{L}}_\GenVec (\sCurlDagger_{2,0} \varphi)_{AA'}$.
It follows that the first order operator $\varphi \to \hat{\mathcal{L}}_\GenVec \varphi$ 
defines a symmetry operator of first order, which is also of the first kind.  For the equations of spins $0$ and $1$, the only first order symmetry operators are given by conformal Killing fields. For the spin-$1$ equation, we may have symmetry operators of the first kind, taking left fields to left, i.e. $\ker \sCurlDagger \mapsto \ker \sCurlDagger$ and of the second kind, taking left fields to right, $\ker \sCurlDagger \mapsto \ker \sCurl$.  
Observe that symmetry operators of the first kind are linear symmetry operators in the usual sense, while symmetry operators of the second kind followed by complex conjugation gives anti-linear symmetry operators in the usual sense.

Recall that the Kerr spacetime admits a constant of motion for geodesics $\GeodesicQ$ which is not reducible to the conserved quantities defined in terms of Killing fields, but rather is defined in terms of a Killing tensor. Similarly, in a spacetime with Killing spinors, the geometric field equations may admit symmetry operators of order greater than one, not expressible in terms of the symmetry operators defined in terms of (conformal) Killing fields. We refer to such symmetry operators as ``hidden symmetries''.  

In general, the existence of symmetry operators of the second order implies the existence of Killing spinors (of valence $(2,2)$ for the conformal wave equation and for Maxwell symmetry operators of the first kind for Maxwell, or $(4,0)$ for Maxwell symmetry operators for of the second kind) satisfying certain auxiliary conditions. The conditions given in \cite{ABB:symop:2014CQGra..31m5015A} are are valid in arbitrary 4-dimensional spacetimes, with no additional conditions on the curvature. As shown in \cite{ABB:symop:2014CQGra..31m5015A}, the existence of a valence $(2,0)$ Killing spinor is a sufficient condition for the existence of second order symmetry operators for the spin-$s$ equations, for $s=0,1/2,1$. 

\begin{remark} \label{rem:remark5:2}
\begin{enumerate} 
\item If $\kappa_{AB}$ is a Killing spinor of valence $(2,0)$, then $L_{ABA'B'} =  \kappa_{AB} \bar \kappa_{A'B'}$ and $L_{ABCD} = \kappa_{(AB} \kappa_{CD)}$ are Killing spinors of valence $(2,2)$ and $(4,0)$, respectively, satisfying the auxiliary conditions given in \cite{ABB:symop:2014CQGra..31m5015A}. 

\item In the case of aligned matter with respect to $\Psi_{ABCD}$, any valence $(4,0)$ Killing spinor $L_{ABCD}$ factorizes, i.e. $L_{ABCD} = \kappa_{(AB} \kappa_{CD)}$ for some Killing spinor $\kappa_{AB}$ of valence $(2,0)$ \cite[Theorem 8]{ABB:symop:2014CQGra..31m5015A}. An example of a spacetime with aligned matter which admits a valence $(2,2)$ Killing spinor that does not factorize is given in \cite[\S 6.3]{ABB:symop:2014CQGra..31m5015A}, see also \cite{2014SIGMA..10..016M}. 
\end{enumerate} 
\end{remark}

\begin{proposition}[\protect{\cite{ABB:symop:2014CQGra..31m5015A}}]
\label{prop:SymOpPot}
\begin{enumerate} 
\item The general symmetry operator of the first kind for the Maxwell field, of order at most two, is of the form 
\begin{align}
\chi_{AB}={}&Q \phi_{AB}+(\sCurl_{1,1} A)_{AB}, \label{eq:SymFirstPot} 
\end{align} 
where $\phi_{AB}$ is a Maxwell field, and $A_{AA'}$ is a  linear concomitant\footnote{A concomitant is a covariant, local partial differential operator.} of first order, such that $A_{AA'} \in \ker \sCurlDagger_{1,1}$ and $Q \in \ker \sTwist_{0,0}$, i.e. locally constant.
\item
The general symmetry operator of the second kind for the Maxwell field is of the form 
\begin{align}
\omega_{A'B'}={}&(\sCurlDagger_{1,1} B)_{A'B'}, \label{eq:SymSecondPot} 
\end{align} 
where $B_{AA'}$ is a first order linear concomitant of $\phi_{AB}$ such that $B_{AA'} \in \ker\sCurl_{1,1}$.
\end{enumerate} 
\end{proposition}

\begin{remark} 
The operators $\sCurlDagger_{1,1}$ and $\sCurl_{1,1}$ are the adjoints of the left and right Maxwell operators $\sCurlDagger_{2,0}$ and $\sCurl_{0,2}$.  
The conserved currents for the Maxwell field can be characterized in terms of solutions of the adjoint Maxwell equations
\begin{subequations}\label{eq:AdjointMaxwell}
\begin{align} 
(\sCurlDagger_{1,1} A)_{A'B'} &= 0 \label{eq::LeftAdjointMaxwell} \\
(\sCurl_{1,1}B)_{AB} &= 0 \label{eq::RightAdjointMaxwell} 
 \end{align} 
 \end{subequations} 
\end{remark}

\begin{definition}
Given a spinor  
$\kappa_{AB} \in \SymSpinSec_{2,0}$ we define the operators $\sExt_{2,0}: \SymSpinSec_{2,0}\rightarrow \SymSpinSec_{2,0}$ and $\bar{\sExt}_{0,2}: \SymSpinSec_{0,2}\rightarrow \SymSpinSec_{0,2}$ by
\begin{subequations}
\begin{align}
(\sExt_{2,0}\varphi)_{AB}={}&-2 \kappa_{(A}{}^{C}\varphi_{B)C}, \label{eq:Ephidef} \\
(\bar{\sExt}_{0,2}\phi)_{A'B'}={}&-2 \bar{\kappa}_{(A'}{}^{C'}\phi_{B')C'}.
\end{align}
\end{subequations}
\end{definition}
Let $\kappa_i$ be the Newman-Penrose scalars for $\kappa_{AB}$. If $\kappa_{AB}$ is of algebraic type $\{1,1\}$ then $\kappa_0 = \kappa_2 = 0$, in which case $\kappa_{AB} = - 2 \kappa_1 o_{(A} \iota_{B)}$. A direct calculations gives the following result. 
\begin{lemma} \label{lem:Ephi-extreme} 
Let $\kappa_{AB} \in \SymSpinSec_{2,0}$ and assume that $\kappa_{AB}$ is of algebraic type $\{1,1\}$. Then the operators $\sExt_{2,0}, \bar{\sExt}_{2,0}$ remove the middle component and rescale the extreme components as 
\begin{subequations}
\begin{align}
(\sExt_{2,0}\varphi)_{0}={}&-2 \kappa_1 \varphi_0,&
(\sExt_{2,0}\varphi)_{1}={}&0,&
(\sExt_{2,0}\varphi)_{2}={}&2 \kappa_1 \varphi_2,\\
(\bar{\sExt}_{0,2}\phi)_{0'}={}&-2 \bar{\kappa}_{1'} \phi_{0'},&
(\bar{\sExt}_{0,2}\phi)_{1'}={}&0,&
(\bar{\sExt}_{0,2}\phi)_{2'}={}&2 \bar{\kappa}_{1'} \phi_{2'}.
\end{align}
\end{subequations}
\end{lemma}
\begin{remark} 
If $\kappa_{AB}$ is a Killing spinor in a Petrov type $\PetrovD$  spacetime, then $\kappa_{AB}$ is of algebraic type $\{1,1\}$. 
\end{remark}

\begin{definition} 
Define the first order 1-form linear concomitants $A_{AA'}, B_{AA'}$ by 
\begin{subequations}\label{eq:ABdef} 
\begin{align} 
A_{AA'}[\kappa_{AB}, \phi_{AB} ] ={}&- \tfrac{1}{3} (\sExt_{2,0} \phi)_{AB}  (\sCurl_{0,2} \bar{\kappa})^{B}{}_{A'}
 + \bar{\kappa}_{A'B'} (\sCurlDagger_{2,0} \sExt_{2,0} \phi)_{A}{}^{B'},\label{eq:Adef} \\
A_{AA'}[\GenVec_{AA'}, \phi_{AB} ] ={}&  \GenVec_{BA'} \phi_{A}{}^{B} \\
B_{AA'}[\kappa_{AB}, \phi_{AB} ] ={}&\kappa_{AB} (\sCurlDagger_{2,0} \sExt_{2,0} \phi)^{B}{}_{A'}
 + \tfrac{1}{3} (\sExt_{2,0} \phi)_{AB} (\sCurlDagger_{2,0} \kappa)^{B}{}_{A'}, \label{eq:Bdef} 
 \end{align}
\end{subequations} 
\end{definition} 
When there is no room for confusion, we suppress the arguments, and write simply $A_{AA'}, B_{AA'}$. 
The following result shows that $A_{AA'}, B_{AA'}$ solves the adjoint Maxwell equations, 
provided $\phi_{AB}$ solves the Maxwell equation. 
\begin{lemma}[\protect{\cite[\S 7]{ABB:symop:2014CQGra..31m5015A}}] 
\label{lem:ABpot}
Assume that $\kappa_{AB}$ is a Killing spinor of valence $(2,0)$, that $\GenVec_{AA'}$ is a conformal Killing field, and that $\phi_{AB}$ is a Maxwell field. Then, with $A_{AA'}, B_{AA'}$ given by \eqref{eq:ABdef} it holds that $A_{AA'}[\kappa_{AB}, \phi_{AB}]$ and $A_{AA'}[\GenVec_{AA'}, \phi_{AB}]$ satisfy $(\sCurlDagger_{1,1} A)_{A'B'} = 0$, and $B_{AA'}[\kappa_{AB}, \phi_{AB}]$ satisfies $(\sCurl_{1,1} B)_{AB} = 0$. 
\end{lemma} 
\begin{remark} Proposition \ref{prop:SymOpPot} together with Lemma \ref{lem:ABpot} show that the existence of a valence $(2,0)$ Killing spinor implies that there are non-trivial second order symmetry operators of the first and second kind for the Maxwell equation. 
\end{remark}

\section{Conservation laws for the Teukolsky system} \label{sec:teuk} 
Recall that the operators $\sCurl$ and $\sCurlDagger$ are adjoints, and hence their composition yields a wave operator.  We have the identities (valid in a general spacetime)
\begin{subequations}\label{eq:wavegen}
\begin{align}\label{eq:MaxWaveGen}
\square \varphi_{AB} + 8 \Lambda \varphi_{AB} - 2 \Psi_{ABCD} \varphi^{CD}={}&-2 (\sCurl_{1,1} \sCurlDagger_{2,0} \varphi)_{AB}, \\
\label{eq:PsiWaveGen} 
\square \varphi_{ABCD} - 6 \Psi_{(AB}{}^{FH}\varphi_{CD)FH}={}&-2 (\sCurl_{3,1} \sCurlDagger_{4,0} \varphi)_{ABCD}.
\end{align} 
\end{subequations}
Here $\varphi_{AB}$ and $\varphi_{ABCD}$ are elements of $\SymSpinSec_{2,0}$ and $\SymSpinSec_{4,0}$, respectively. This means that the the Maxwell equation 
$
(\sCurlDagger_{2,0} \phi)_{AA'} = 0
$
in a vacuum spacetime implies the wave equation 
\begin{align}\label{eq:MaxWave}
\square \phi_{AB} - 2 \Psi_{ABCD} \phi^{CD}= 0 .
\end{align} 
Similarly, in a vacuum spacetime, the Bianchi system $(\sCurlDagger_{4,0} \Psi)_{A'ABC} = 0$ holds for the Weyl spinor, and we arrive at the Penrose wave equation 
\begin{align}\label{eq:PsiWave} 
\square \Psi_{ABCD} - 6 \Psi_{(AB}{}^{FH}\Psi_{CD)FH}= 0
\end{align} 
Restricting to a vacuum type $\PetrovD$ spacetime, and projecting the Maxwell wave equation \eqref{eq:MaxWave} and the linearized Penrose wave equation \eqref{eq:PsiWave} on the principal spin dyad, one obtains wave equations for the extreme Maxwell scalars $\phi_0, \phi_2$ and the extreme linearized Weyl scalars $\dot \Psi_0, \dot \Psi_4$. 

\newcommand{\sfrak}{\mathfrak{s}}
Letting $\psi^{(\sfrak)}$ denote $\phi_0, \Psi_2^{-2/3} \phi_2$ for $\sfrak = 1, -1$, respectively, and $\dot \Psi_0, \Psi_2^{-4/3} \dot \Psi_4$ for $\sfrak = 2, -2$, respectively, one finds that these fields satisfy the system  
\begin{equation}\label{eq:squareTME} 
[ \squareTME_{2\sfrak} - 4 \sfrak^2 \Psi_2 ] \psi^{(\sfrak)} = 0 ,
\end{equation} 
see \cite[\S 3]{aksteiner:andersson:2011CQGra..28f5001A}, 
where, in GHP notation 
\begin{equation}\label{eq:squaretme-restrict}
\squareTME_{p} =
2(\tho -p\rho -\bar{\rho})({\tho}'-\rho')- 2(\edt-p\tau
-\bar{\tau}')({\edt}'-\tau')  + (3p-2)\Psi_2 .
\end{equation}
The equation \eqref{eq:squareTME} was first derived by Teukolsky \cite{teukolsky:1972PhRvL..29.1114T,teukolsky:1973} for massless spin-$s$ fields and linearized gravity on Kerr, and is referred to as the Teukolsky Master Equation (TME). It was shown by Ryan \cite{ryan:1974PhRvD..10.1736R} that the tetrad projection of the linearized Penrose wave equation yields the TME, see also Bini et al \cite{bini:etal:2002PThPh.107..967B,bini:etal:2003IJMPD..12.1363B}. In the Kerr case, the TME admits a commuting symmetry operator, and hence allows separation of variables. The TME applies to fields of all half-integer spins between $0$ and $2$. 

As discussed above, the TME is a wave equation for the weighted field $\psi^{(\sfrak)}$.  
It is derived from the spin-$s$ field equation by applying a first order operator and hence is valid for the extreme scalar components of the field, rescaled as explained above. It is important to emphasize that there is a loss of information in deriving the TME from the spin-$s$ equation. For example, if we consider two independent solutions of the TME with spin weights $\sfrak=\pm 1$, these will not in general be components of a single Maxwell field. If indeed this is the case, the Teukolsky-Starobinsky identities (TSI) (also referred to as Teukolsky-Press relations), see \cite{kalnins:etal:1989JMP....30.2925K} and references therein, hold. 

The TME admits commuting symmetry operators $\Sop_{\sfrak}, \Rop_{\sfrak}$, so that 
$$
 \squareTME_{2\sfrak} - 4 \sfrak^2 \Psi_2 = \Rop_{\sfrak} + \Sop_{\sfrak}
$$
with $[\Rop_{\sfrak},\Sop_{\sfrak}] = 0$, and such that as in the case for the wave equation discussed in section \ref{sec:kerrwave}, the operators $\Rop_{\sfrak},\Sop_{\sfrak}$ involve derivatives with respect to $r$ and $\theta$, respectively, in addition to derivatives in the symmetry directions $t,\phi$. This shows that one may make a consistent separated ansatz 
$$
\psi^{(\sfrak)}(t,r,\theta,\phi) = e^{-i\omega t} e^{im\phi} \Rsol^{(\sfrak)}(r) \Ssol^{(\sfrak)}(\theta) 
$$
where $\Rsol^{(\sfrak)}$ solves the radial TME 
$$
(\Rop_\sfrak + \lambda_{\sfrak,\omega,\ell,m} ) \Rsol^{(\sfrak)} = 0
$$
where $\lambda_{\sfrak,\omega,\ell,m}$ is an eigenvalue for the angular Teukolsky equation $\Sop_\sfrak \Ssol^{(\sfrak)} = \lambda \Ssol^{(\sfrak)}$, which is the equation for a spin-weighted spheroidal harmonic. 

Although the TSI are usually discussed in terms of separated forms of $\psi^{(\sfrak)}$, we are here interested in the TSI as differential relations between the scalars of extreme spin weights. From this point of view, the TSI expresses the fact that the Debye potential construction starting from the different Maxwell scalars for a given Maxwell field $\phi_{AB}$ yields scalars of the \emph{the same} Maxwell field. The equations for the Maxwell scalars in terms of Debye potentials can be found in Newman-Penrose notation in \cite{cohen:kegeles:1974PhRvD..10.1070C}. These expressions correspond to the components of a symmetry operator of the second kind. See \cite[\S 5.4.2]{aksteiner:thesis} for further discussion, where also the GHP version of the formulas can be found. 
An analogous situation obtains for the case of linearized gravity, see \cite{lousto:whiting:2002PhRvD..66b4026L}. In this case, the TSI are of fourth order.  
Thus, for a Maxwell field, or a solution of the linearized Einstein equations on a Kerr, or more generally a vacuum type $\PetrovD$ background, the pair of Newman-Penrose scalars of extreme spin weights for the field satisfy a system of differential equations consisting of both the TME and the TSI. 

Although the TME is derived from an equation governed by a variational principle, it has been argued by Anco, see the discussion in \cite{perjes:lukacs_2005AIPC..767..306P}, that the Teukolsky system admits no \emph{real} variational principle, due to the fact that the operator $\squareTME_p $ defined by the above fails to be formally self-adjoint. Hence, the issue of real conserved currents for the Teukolsky system, which appear to be necessary for estimates of the solutions, appears to be open. 
However, as we shall demonstrate here, if we consider the \emph{combined} TME and TSI in the spin-$1$ or Maxwell case, as a system of equations for both of the extreme Maxwell scalars $\phi_0, \phi_2$, this system does admit both a conserved current and a conserved stress-energy like tensor.

\subsection{A new conserved tensor for Maxwell} \label{sec:Vab} 
Let $T_{ab}$ be the Maxwell stress-energy tensor, and let 
$\phi_{AB} \to \chi_{AB}$ be the second order symmetry operator of the first kind given by \eqref{eq:SymFirstPot} with $Q = 0$ and $A_{AA'}$ given by \eqref{eq:Adef}. Then the current $T_{ab} \xi^b$ is conserved. In fact, as discussed in \cite[section 6]{2015arXiv150402069A}, it  is equivalent to a current $V_{ab} \xi^b$ defined in terms of a symmetric tensor $V_{ab}$ which we shall now introduce. 
Let 
\begin{align}
\eta_{AA'}\equiv{}&(\sCurlDagger_{2,0} \sExt_{2,0}\phi)_{AA'}. \label{eq:etadef}
\end{align}
where $(\sExt_{2,0}\phi)_{AB}$ is given by \eqref{eq:Ephidef} and define the symmetric tensor $V_{ab}$ by  
\begin{align}
V_{ABA'B'}\equiv{}&\tfrac{1}{2} \eta_{AB'} \bar{\eta}_{A'B}
 + \tfrac{1}{2} \eta_{BA'} \bar{\eta}_{B'A}
 + \tfrac{1}{3} (\sExt_{2,0}\phi)_{AB} (\hat{\mathcal{L}}_{\bar\xi}\bar{\phi})_{A'B'}
 + \tfrac{1}{3} (\bar{\sExt}_{2,0} \bar{\phi})_{A'B'} (\hat{\mathcal{L}}_{\xi}\phi)_{AB}. \label{eq:Vdef}
\end{align}
Then,  as we shall now show, $V_{ab}$ is itself conserved,  
$$
\nabla^a V_{ab} = 0, 
$$
and hence may be viewed as a higher-order stress-energy tensor for the Maxwell field. 
The tensor $V_{ab}$ has several important properties. First of all, if $\Mcal$ is of Petrov type $\PetrovD$, it depends only on the extreme Maxwell scalars $\phi_0, \phi_2$, and hence cancels the static Coulomb Maxwell field \eqref{eq:coulomb} on Kerr which has only the middle scalar non-vanishing. This can be proved using Lemma \ref{lem:Ephi-extreme}, cf. \cite[Corollary 6.2]{2015arXiv150402069A}. Further, the tensor 
$$
U_{AA'BB'} = \tfrac{1}{2} \eta_{AB'} \bar{\eta}_{A'B}
 + \tfrac{1}{2} \eta_{BA'} \bar{\eta}_{B'A}
 $$
is a superenergy tensor for the 1-form field $\eta_{AA'}$, and hence satisfies the dominant energy condition, cf. \cite{bergqvist:1999CMaPh.207..467B,senovilla:2000CQGra..17.2799S}. Note that the notion of superenergy tensor extends to spinors of arbitrary valence. Similarly to the wave equation stress energy, $V_{ab}$ has non-vanishing trace, $V^a{}_a = U^a{}_a = - \bar \eta^a \eta_a$.  

In order to analyze $V_{ab}$, we first collect some properties of the one-form $\eta_{AA'}$ as defined in \eqref{eq:etadef}. 

\begin{lemma}[\protect{\cite[Lemma 2.4]{2014arXiv1412.2960A}}] 
\label{lem:etaeqs} 
Let $\kappa_{AB} \in \KillSpin_{2,0}$, and assume the aligned matter condition holds with respect to $\kappa_{AB}$. Let $\xi_{AA'}$ be given by \eqref{eq:xikappadef}. Further, let $\phi_{AB}$ be a Maxwell field, and let $\eta_{AA'}$ be given by \eqref{eq:etadef}. Then we have 

\begin{subequations} \label{eq:etafacts} 
\begin{align}
(\sDiv_{1,1} \eta)={}&0,\label{diveta1}\\
(\sCurl_{1,1} \eta)_{AB}={}&\tfrac{2}{3} (\hat{\mathcal{L}}_{\xi}\phi)_{AB},\label{curleta1b2}\\
(\sCurlDagger_{1,1} \eta)_{A'B'}={}&0,  \label{curleta2}\\
\eta_{AA'} \xi^{AA'}={}&\kappa^{AB} (\hat{\mathcal{L}}_{\xi}\phi)_{AB}. \label{eq:etaLphi} 
\end{align}
\end{subequations}
\end{lemma} 

The following Lemma gives a general condition, without assumptions on the spacetime geometry, for a tensor constructed along the lines of $V_{ab}$ to be conserved. The proof is a straightforward computation. 
\begin{lemma}[\protect{\cite[Lemma 3.1]{2014arXiv1412.2960A}}] \label{lem:Tconserved}
Assume that $\varphi_{AB}\in \SymSpinSec_{2,0}$ satisfies the system
\begin{subequations} \label{eq:varphieqs} 
\begin{align}
(\sCurlDagger_{1,1} \sCurlDagger_{2,0} \varphi)_{A'B'}={}&0,\label{eq:CurlDgCurlDgvarphi1}\\
(\sCurl_{1,1} \sCurlDagger_{2,0} \varphi)_{AB}={}&\varpi_{AB},\label{eq:CurlCurlDgvarphi1}
\end{align}
\end{subequations}
for some $\varpi_{AB} \in \SymSpinSec_{2,0}$. 
Let   
 \begin{align}\label{eq:etavarphidef} 
\etavarphi_{AA'}={}&(\sCurlDagger_{2,0} \varphi)_{AA'},
\end{align}
and define the symmetric tensor 
$\EMTensorT_{ABA'B'}$ by 
\begin{align}
\EMTensorT_{ABA'B'}={}&\tfrac{1}{2} \etavarphi_{AB'} \bar{\etavarphi}_{A'B}
 + \tfrac{1}{2} \etavarphi_{BA'} \bar{\etavarphi}_{B'A}
 + \tfrac{1}{2} \bar{\varpi}_{A'B'} \varphi_{AB}
 + \tfrac{1}{2} \varpi_{AB} \bar{\varphi}_{A'B'}.
\end{align}
Then
\begin{align} \label{eq:EMTensorcons} 
\nabla^{BB'}\EMTensorT_{ABA'B'}={}&0.
\end{align}
\end{lemma}
We now have the following result, which follows directly from Lemma \ref{lem:Tconserved} and the identities for $\eta_{AA'}$ given in Lemma \ref{lem:etaeqs} together with the above remarks. 
\begin{theorem}[\protect{\cite[Theorem 1.1]{2014arXiv1412.2960A}}] Assume that $(\Mcal, \met_{ab})$ admits a valence $(2,0)$ Killing spinor $\kappa_{AB}$ and assume that the aligned matter condition holds with respect to $\kappa_{AB}$.  
Let $\phi_{AB}$ be a solution of the Maxwell equation. Then the tensor 
$V_{ABA'B'}$ given by \eqref{eq:Vdef} is conserved, i.e. 
$$
\nabla^{AA'} V_{ABA'B'} = 0
$$
If in addition $(\Mcal, \met_{ab})$ is of Petrov type $\PetrovD$, then $V_{ab}$ depends only on the extreme components of $\phi_{AB}$.  
\end{theorem} 

The properties of $V_{ab}$ indicate that $V_{ab}$, rather than the Maxwell stress-energy $T_{ab}$ may be used in proving dispersive estimates for the Maxwell field. In section \ref{sec:morawetz} we shall outline the proof of a Morawetz estimate for the Maxwell field on the Schwarzschild background, making use of a related approach.

\subsection{Teukolsky equation and conservation laws} 
We end this section by pointing out the relation between the fact that $V_{ab}$ is conserved, and the TME and TSI which follow from the Maxwell equation in a Petrov type $\PetrovD$ spacetime. 

A computation shows that the identities  
\begin{subequations} \label{eq:TME-TSI-cov}
\begin{align} 
(\sCurlDagger_{1,1} \sCurlDagger_{2,0} \sExt_{2,0} \phi)_{A'B'} &= 0 
\label{eq:bCdCdEphi}
\\
(\sExt_{2,0}\sCurl_{1,1} \sCurlDagger_{2,0} \sExt_{2,0} \phi)_{AB}={}&\tfrac{2}{3} (\hat{\mathcal{L}}_{\xi}\sExt_{2,0} \phi)_{AB}.
\label{eq:ECCdEphi}
\end{align} 
\end{subequations}
follow from the Maxwell equations, cf. \cite[Eq. (3.5)]{2014arXiv1412.2960A}. We see that this system is equivalent to \eqref{eq:CurlDgCurlDgvarphi1}, \eqref{eq:CurlCurlDgvarphi1}, with $\varphi_{AB}=\sExt_{2,0} \phi$ and $\varpi_{AB}=\tfrac{2}{3} (\hat{\mathcal{L}}_{\xi}\phi)_{AB}$. This shows that the fact that $V_{ab}$ is conserved is a direct consequence of \eqref{eq:TME-TSI-cov}, which in fact are the covariant versions of the TME and TSI. In order to make this clear for the case of the TME, we project \eqref{eq:ECCdEphi}on the dyad. A calculation shows that 
\begin{subequations}
\begin{align}
0={}&- \tho \tho' \varphi_0
 + \rho \tho' \varphi_0
 + \bar{\rho} \tho' \varphi_0
 + \edt \edt' \varphi_0
 -  \tau \edt' \varphi_0
 -  \overline{\tau '} \edt' \varphi_0,\\
0={}&- \rho ' \tho \varphi_2
 -  \overline{\rho '} \tho \varphi_2
 + \tho' \tho \varphi_2
 + \bar{\tau} \edt \varphi_2
 + \tau ' \edt \varphi_2
 -  \edt' \edt \varphi_2.
\end{align}
\end{subequations}
where $\varphi_{0}=-2 \kappa_1 \phi_0$ and $\varphi_{2}=2 \kappa_1 \phi_2$. 
We see from this that \eqref{eq:ECCdEphi} is equivalent to the scalar form of TME for Maxwell given in \eqref{eq:squareTME} above. Further, one can show along the same lines that 
\eqref{eq:bCdCdEphi} is equivalent to the TSI for Maxwell given in scalar form in \cite[\S 5.4.2]{aksteiner:thesis}, cf. \cite[\S 3.1]{2014arXiv1412.2960A}.

\section{A Morawetz estimate for the Maxwell field on Schwarzschild} \label{sec:morawetz} 
In this section, we shall outline the proof of the Morawetz estimate for the Maxwell field on the Schwarzschild spacetime given recently in \cite{2015arXiv150104641A}. 

Assume that $\kappa_{AB}$ is a valence $(2,0)$ Killing spinor, such that $\kappa_{CD} \kappa^{CD}\neq 0$. In the case of the Schwarzschild spacetime, $\kappa_{AB}$ is given by \eqref{eq:kappaKerrNewman}. Define the Killing fields $\xi^{AA'}, \eta^{AA'}$ in terms of $\kappa_{AB}$ by \eqref{eq:xikappadef} and \eqref{eq:etadef}, respectively. 

Now let $\phi_{AB}$ be a solution to the source-free Maxwell equation $(\sCurlDagger_{2,0} \phi)_{AA'} = 0$ and define
\begin{subequations}
\begin{align}
U_{AA'}={}&- \tfrac{1}{2} \nabla_{AA'}\log(-\kappa_{CD} \kappa^{CD}),\\
\Upsilon={}&\kappa^{AB} \phi_{AB},\\
\Theta_{AB}={}&(\sExt_{2,0}\phi)_{AB},\label{ThetaDef}\\
\beta_{AA'}={}&\eta_{AA'}+U^{B}{}_{A'} \Theta_{AB}. \label{betaDef}
\end{align}
\end{subequations}
In the Schwarzschild case, we have 
\begin{align}
\xi^{AA'}={}&(\partial_t)^{AA'},&
U_{AA'}={}&- r^{-1}\nabla_{AA'}r.
\end{align}
Analogously to Lemma~\ref{lem:etaeqs}, we have
\begin{lemma}[\protect{\cite[Lemma 8]{2015arXiv150104641A}}] \label{lem:betaprop}
\begin{subequations}
\begin{align}
\beta_{AA'}={}&- U_{AA'} \Upsilon
 + (\sTwist_{0,0} \Upsilon)_{AA'},\label{eq:betaToUpsilon}\\
(\sDiv_{1,1} \beta)={}&- U^{AA'} \beta_{AA'},\label{eq:Divbeta}\\
(\sCurl_{1,1} \beta)_{AB}={}&U_{(A}{}^{A'}\beta_{B)A'},\label{eq:Curlbeta}\\
(\sCurlDagger_{1,1} \beta)_{A'B'}={}&U^{A}{}_{(A'}\beta_{|A|B')}.\label{eq:CurlDgbeta}
\end{align}
\end{subequations}
\end{lemma}
The superenergy tensors for $\beta_{AA'}$ and $\Theta_{AB}$ are given by
\begin{subequations}
\begin{align}
\EMTensorH_{ABA'B'}={}&\tfrac{1}{2} \beta_{AB'} \overline{\beta}_{A'B}
 + \tfrac{1}{2} \beta_{BA'} \overline{\beta}_{B'A},  \label{eq:Hdef} \\
\mathbf{W}_{ABA'B'} ={}& \Theta_{AB} \overline\Theta_{A'B'} .
\end{align}
\end{subequations}
Choosing the principal tetrad in Schwarzschild given by specializing \eqref{eq:KerrTetrad} to $a = 0$ gives in a standard manner an orthonormal frame, 
\begin{align*}
\OrthFrameT^{AA'}\equiv{}&\tfrac{1}{\sqrt{2}}(o^{A} \bar o^{A'}
 + \iota^{A} \bar\iota^{A'}),&
\OrthFrameX^{AA'}\equiv{}&\tfrac{1}{\sqrt{2}}(\bar o^{A'} \iota^{A}
 + o^{A} \bar\iota^{A'}),\\
\OrthFrameY^{AA'}\equiv{}&\tfrac{i}{\sqrt{2}}(- \bar o^{A'} \iota^{A}
 + o^{A} \bar\iota^{A'}),&
\OrthFrameZ^{AA'}\equiv{}&\tfrac{1}{\sqrt{2}}(o^{A} \bar o^{A'}
 -  \iota^{A} \bar\iota^{A'}).
\end{align*}
The tensor $\EMTensorH_{ABA'B'}$, which agrees up to lower order terms with the conserved tensor $V_{ABA'B'}$ introduced in section \ref{sec:Vab}, is not itself conserved, it yields a conserved energy current.  
\begin{lemma}[\protect{\cite[Lemma 11]{2015arXiv150104641A}}]\label{lem:divH}
For the Schwarzschild spacetime we have
\begin{subequations}\label{eq:divHeqs}
\begin{align}
\nabla^{BB'}\EMTensorH_{ABA'B'}={}&- U_{AA'} \beta^{BB'} \bar{\beta}_{B'B},\label{eq:divH}\\
\xi^{AA'} \nabla^{BB'}\EMTensorH_{ABA'B'}={}&0.\label{eq:xidivH}
\end{align}
\end{subequations}
In particular, $\xi^{BB'}\EMTensorH_{ABA'B'}$ is a future causal conserved current.
\end{lemma}
This result makes use of the fact that the Schwarzschild spacetime is non-rotating. 
For the Kerr spacetime with non-vanishing angular momentum, the 1-form $U_{AA'}$ fails to be real and the current $\EMTensorH_{ab} \xi^b$ is not conserved.

For a vector field $\vecMprimary^a$ and a scalar $\scalMprimary$, define the Morawetz current $\MorawetzCurrJ_a$ by 
\begin{align}
\MorawetzCurrJ_{AA'}={}&\EMTensorH_{ABA'B'} \vecMprimary^{BB'}
 -  \tfrac{1}{2} \scalMprimary \bar{\beta}_{A'}{}^{B} \Theta_{AB}
 -  \tfrac{1}{2} \scalMprimary \beta_{A}{}^{B'} \overline{\Theta}_{A'B'}
  + \tfrac{1}{2} \Theta_{A}{}^{B} \overline{\Theta}_{A'}{}^{B'} (\sTwist_{0,0} \scalMprimary )_{BB'}.\label{eq:MorawetzCurrJDef}
\end{align}
For any spacelike hypersurface $\Sigma$, we define the energy integrals
\begin{align}
E_\xi(\Sigma)={}&\int_{\Sigma}	\EMTensorH_{ab}\xi^{b}
N^{a} \mbox{d}\mu_{\Sigma}, \label{eq:Exidef}\\
E_{\xi+\vecMprimary,\scalMprimary}(\Sigma)
={}&\int_{\Sigma}  \left(\EMTensorH_{ab}\xi^b+\MorawetzCurrJ_a\right)
N^{a} \mbox{d}\mu_{\Sigma}. \label{eq:ExiAqdef}
\end{align}
In view of Lemma \ref{lem:divH}, $E_{\xi}(\Sigma)$ is nonnegative and conserved.

We shall make the following explicit choices of the $\vecMprimary^a$ and $\scalMprimary$,
\begin{subequations}
\begin{align}
\vecMprimary^a={}&\frac{(r - 3 M) (r - 2 M)}{2 r^2}(\partial_r)^a, \label{eq:vecMdef} \\
\scalMprimary ={}&\frac{9 M^2 (r - 2 M) (2 r - 3 M)}{4 r^5}. \label{eq:qdef}
\end{align}
\end{subequations}

\subsection{Positive energy}\label{sec:posenergy}
Before proving the integrated decay estimate, we shall verify that the energy \eqref{eq:ExiAqdef} is be non-negative, and uniformly equivalent to the energy \eqref{eq:Exidef}.

From the properties of spin-weighted spherical harmonics, one derives the inequalities 
\begin{subequations} \label{eq:HardyTheta}
\begin{align}
\int_{S_r} |\varphi_{0}|^2 d\mu_{S_r}
\leq{}& r^2\int_{S_r} |\edt' \varphi_{0}|^2 \mbox{d}\mu_{S_r},\label{eq:HardyTheta0}\\
\int_{S_r} |\varphi_{2}|^2 d\mu_{S_r}
\leq{}& r^2\int_{S_r} |\edt \varphi_{2}|^2 \mbox{d}\mu_{S_r}.\label{eq:HardyTheta2}
\end{align}
\end{subequations}
for the extreme scalars $\varphi_0, \varphi_2$ of a smooth symmetric spinor field $\varphi_{AB}$, cf. \cite[Lemma 6]{2015arXiv150104641A}. Here $S_r$ is a sphere with constant $t,r$ in the Schwarzschild spacetime.

By making use of the Cauchy-Schwarz inequality, and the Hardy type inequalities \eqref{eq:HardyTheta}, we get
\begin{theorem}[\protect{\cite[Theorem 13, Corollary 14]{2015arXiv150104641A}}]\label{thm:PositiveEnergy}
Let $\vecMprimary^{AA'}$ and $\scalMprimary$ be given by \eqref{eq:vecMdef} and \eqref{eq:qdef}.
\begin{enumerate}
\item 
For any constant $|c_1|\leq 10/9$ and any spherically symmetric slice $\Sigma$ with future pointing timelike normal $N^{AA'}$ such that $N^{AA'}N_{AA'}=1$ we have a positive energy
\begin{align}
\int_{\Sigma} N^{AA'} (\EMTensorH_ {ABA'B'} \xi^{BB'} + c_ 1\MorawetzCurrJ_{AA'})  \mbox{d}\mu_{\Sigma_i} \geq{}&0.
\end{align}
\item 
For any spherically symmetric slice $\Sigma$ with future pointing timelike normal $N^{AA'}$ such that $N^{AA'}N_{AA'}=1$ the energies $E_\xi(\Sigma)$ and
$E_{\xi +\vecMprimary, \scalMprimary}(\Sigma)$ are uniformly equivalent,
\begin{align}
\tfrac{1}{10}E_\xi(\Sigma)
\leq{}&
E_{\xi +\vecMprimary, \scalMprimary}(\Sigma)
\leq
\tfrac{19}{10}E_\xi(\Sigma).
\end{align}
\end{enumerate} 
\end{theorem}
In particular, we find that using Theorem \ref{thm:PositiveEnergy}, we can dominate the integral of the bulk term for the Morawetz current over a spacetime domain bounded by Cauchy surfacs $\Sigma_1$, $\Sigma_2$, in terms of the energies $E_{\xi}(\Sigma_1)$, $E_\xi(\Sigma_2)$.  This is the essential step in the proof of an integrated energy decay (or Morawetz) estimate.

We shall apply \eqref{eq:HardyTheta} to $\Theta_{AB}$. 
We have
\begin{subequations}
\begin{align}
|\beta_{\OrthFrameZ}|^2 + | \beta_{\OrthFrameZ}|^2={}&|\edt \Theta_{2}|^2
 + |\edt' \Theta_{0}|^2,\label{eq:TbetaZbetatocomp}\\
\intertext{and}
\mathbf{W}_{\OrthFrameT\OrthFrameT}=\OrthFrameT^{AA'} \OrthFrameT^{BB'} \Theta_{AB}\overline\Theta_{A'B'}={}&
\tfrac{1}{2} |\Theta_{0}|^2
 + \tfrac{1}{2} |\Theta_{2}|^2.\label{eq:E2tocomp}
\end{align}
\end{subequations}
Equations \eqref{eq:TbetaZbetatocomp}, \eqref{eq:E2tocomp} and the inequalities \eqref{eq:HardyTheta} combine to give the estimate 
\begin{align}
\int_{S_r} \mathbf{W}_{\OrthFrameT\OrthFrameT}	d\mu_{S_r}
\leq{}&\frac{r^2}{2} \int_{S_r} |\beta_{\OrthFrameT}|^2 + |\beta_{\OrthFrameZ}|^2 \mbox{d}\mu_{S_r}, 
\end{align}
cf. \cite[Lemma 15]{2015arXiv150104641A}.
From the form \eqref{eq:MorawetzCurrJDef} of the Morawetz current $\MorawetzCurrJ_a$, the definition of $\beta_{AA'}$ and the properties of $\beta_a$ given in Lemma~\ref{lem:betaprop} we get
\begin{align}
- (\sDiv_{1,1} \MorawetzCurrJ)={}&-  \beta^{AA'} \bar{\beta}^{B'B} (\sTwist_{1,1} \vecMprimary)_{ABA'B'}
+\beta^{AA'} \bar{\beta}_{A'A} \bigl(
\tfrac{1}{4} (\sDiv_{1,1} \vecMprimary)
 + \vecMprimary^{BB'} U_{BB'}
- \scalMprimary\bigr)\nonumber\\
& + \Theta_{AB} \overline{\Theta}_{A'B'} \bigl(U^{AA'} (\sTwist_{0,0} \scalMprimary )^{BB'}
 -  \tfrac{1}{2} (\sTwist_{1,1} \sTwist_{0,0} \scalMprimary )^{ABA'B'}\bigr).\label{eq:divJJeq1}
\end{align}
With the explicit choices \eqref{eq:vecMdef} and \eqref{eq:qdef} for the Morawetz vector field $\vecMprimary^a$ and the scalar $\scalMprimary$, respectively, the above estimates now yield 
\begin{align}
\int_{\Omega}-(\sDiv_{1,1} \MorawetzCurrJ)d\mu_{\Omega}\geq{}&
\int_{\Omega}\frac{1}{8}|\beta_{AA'}|^2_{1,\text{deg}}+\frac{M}{100 r^4}|\Theta_{AB}|^2_2 d\mu_{\Omega}, \label{eq:divPint}
\end{align}
for any spherically symmetric spacetime region $\Omega$ of the Schwarzschild spacetime. 

We now make use of Gauss' formula to evaluate the left hand side of \eqref{eq:divPint}. Theorem \ref{thm:PositiveEnergy} and the estimates just proved then yield the following energy bound and Morawetz estimate for the Maxwell field on the Schwarzschild spacetime. 
\begin{theorem}[\protect{\cite[Theorem 2]{2015arXiv150104641A}}]
Let $\Sigma_1$ and $\Sigma_2$ be spherically symmetric spacelike
hypersurfaces in the exterior region of the Schwarzschild spacetime
such that $\Sigma_2$ lies in the future of $\Sigma_1$ and
$\Sigma_2\cup -\Sigma_1$ is the oriented boundary of a spacetime
region $\Omega$.

If $\phi_{AB}$ is a solution of the Maxwell equations on the Schwarzschild exterior,
and $\Theta_{AB}$ and $\beta_{AA'}$
are defined by equations
\eqref{betaDef}-\eqref{ThetaDef}, then
\begin{align}
E_{\xi}(\Sigma_2)={}&E_{\xi}(\Sigma_1),\\
\int_{\Omega} \vert \beta_{AA'}\vert^2_{1,\text{deg}}
+\frac{2 M}{25 r^4}\vert \Theta_{AB}\vert^2_{2} d\mu_{\Omega}
\leq{}& \frac{72}{5}E_{\xi}(\Sigma_1), \label{eq:MorIneq}
\end{align}
where $E_\xi(\Sigma_i)$ is the energy associated with $\xi^a$, evaluated on $\Sigma_i$, and
$|\beta_{AA'}|_{1,deg}$ and $|\Theta_{AB}|_{2}$ are, respectively, the degenerate norm of $\beta_{AA'}$ and the norm of $\Theta_{AB}$ defined by
\begin{subequations}
\begin{align*}
\vert \beta_{AA'}\vert^2_{1,\text{deg}}={}&
\frac{(r - 3 M)^2}{r^3}\bigl(|\beta_{\OrthFrameX}|^2 + |\beta_{\OrthFrameY}|^2\bigr)  + \frac{M (r - 2 M)}{r^3}|\beta_{\OrthFrameZ}|^2
 + \frac{M  (r - 3 M)^2 (r - 2 M)}{r^5}|\beta_{\OrthFrameT}|^2,\\
\vert \Theta_{AB}\vert^2_{2}={}& \frac{ (r - 2 M)}{r} \mathbf{W}_{\OrthFrameT\OrthFrameT}.
\end{align*}
\end{subequations}
\qed
\end{theorem}

\subsection*{Acknowledgements} We are grateful to Steffen Aksteiner, Siyuan Ma, Marc Mars and Claudio Paganini for helpful remarks. 
PB and TB were supported by EPSRC grant EP/J011142/1. LA thanks Institut Henri Poincar\'e, Paris, for hospitality and support during part of the work on these notes.

\newcommand{\arxivref}[1]{\href{http://www.arxiv.org/abs/#1}{{arXiv.org:#1}}}
\newcommand{\mnras}{Monthly Notices of the Royal Astronomical Society}
\newcommand{\prd}{Phys. Rev. D} 
\bibliographystyle{abbrv} %

\end{document}